\documentclass[a4paper,reqno, 11pt]{amsart}
\usepackage{fullpage}

\usepackage{color}
\usepackage[capitalize,nameinlink,noabbrev,nosort]{cleveref}
\usepackage{graphicx}
\usepackage{tikz-cd}
\usepackage{mathtools}

\usepackage{enumerate}

\usepackage{caption}
\usepackage{subcaption}

\usepackage{graphicx,color,marginnote}
\usepackage{amsmath}
\usepackage[utf8]{inputenc}
\usepackage[T1]{fontenc}
\usepackage[noadjust]{cite}
\usepackage[english]{babel}
\usepackage[normalem]{ulem} 

\usepackage{amsfonts}
\usepackage{amssymb}
\usepackage{bbm}
\usepackage{epstopdf}
\usepackage{color}
\numberwithin{equation}{section}

\newcommand{\C}{\mathbb{C}}
\newcommand{\R}{\mathbb{R}}

\newcommand{\DAz}{\frac{1}{\sqrt{2}}\mathbb{D}}

\renewcommand{\d}{ {\delta} }

\newcommand{\tb}{\bullet}
\newcommand{\tw}{\circ}

\usepackage{amsthm}
\theoremstyle{plain}
\newtheorem{theorem}{Theorem}[section]

\newtheorem{corollary}[theorem]{Corollary}
\newtheorem{assumption}[theorem]{Assumption}
\newtheorem*{assumption*}{Assumption}
\newtheorem{claim}[theorem]{Claim}
\newtheorem{lemma}[theorem]{Lemma}
\newtheorem{definition}[theorem]{Definition}
\newtheorem{proposition}[theorem]{Proposition}

\theoremstyle{remark}
\newtheorem{remark}[theorem]{Remark}

\newcommand{\old}[1]{}

\newcommand{\G}{\mathcal{G}}
\newcommand{\T}{\mathcal{T}}
\newcommand{\Or}{\mathcal{O}}

\newcommand{\eps}{\varepsilon}

\newcommand{\ZZ}{\mathbb{Z}}

\newcommand{\CC}{\mathbb{C}}

\newcommand{\Ordo}{\mathcal{O}}

\DeclareMathOperator{\im}{Im}
\DeclareMathOperator{\re}{Re}

\renewcommand{\d}{\,\mathrm{d}}
\renewcommand{\i}{\mathrm{i}}
\newcommand{\e}{\mathrm{e}}

\DeclareMathOperator{\dist}{dist}

\def\LipKd{{\mbox{\textsc{Lip(}$\kappa$\textsc{,}$\delta$\textsc{)}}}}

\def\ExpFat{{{\mbox{\textsc{Exp-Fat(}$\delta$\textsc{)}}}}}

\newcommand{\TG}{\widetilde{A}}
\newcommand{\HS}{H}

\usepackage{bm}
\usepackage{color}

\title[title]{Perfect t-embeddings of uniformly weighted Aztec diamonds and tower graphs}

\author[Tomas Berggren]{Tomas Berggren$^\mathrm{a}$}

\author[Matthew Nicoletti]{Matthew Nicoletti$^\mathrm{a}$}

\author[Marianna Russkikh]{Marianna Russkikh$^\mathrm{b}$}

\thanks{\textsc{${}^\mathrm{A}$ Massachusetts Institute of Technology, Department of Mathematics, 77 Massachusetts Avenue, Cambridge, Massachusetts, 02139–4307, United States of America}}
\thanks{\textsc{${}^\mathrm{B}$  California Institute of Technology, Division of Physics, Mathematics and Astronomy, 1200 E California Blvd, Pasadena, CA 91125, United States of America}}

\thanks{\texttt{tomasb@mit.edu}, \texttt{mnicolet@mit.edu}, \texttt{russkikh@caltech.edu}}

\begin{document}

\begin{abstract}
In this work we study a sequence of perfect t-embeddings of uniformly weighted Aztec diamonds. We show that these perfect t-embeddings can be used to prove convergence of gradients of height fluctuations to those of the Gaussian free field. In particular we provide a first proof of the existence of a model satisfying all conditions of the main theorem of~\cite{CLR2}. This confirms the prediction of~\cite{Ch-R}. An important part of our proof is to exhibit exact integral formulas for perfect t-embeddings of uniformly weighted Aztec diamonds.

In addition, we construct and analyze perfect t-embeddings of another sequence of uniformly weighted finite graphs called tower graphs. Although we do not check all technical assumptions of the mentioned theorem for these graphs, we use perfect t-embeddings to derive a simple transformation which identifies height fluctuations on the tower graph with those of the Aztec diamond.
\end{abstract}



\maketitle

\tableofcontents


\section{Introduction}\label{sec:intro}

\begin{figure}
 \begin{center}
 \vspace{-120pt}
\includegraphics[scale=.46, trim={3cm 8cm 2cm 7cm}, clip]{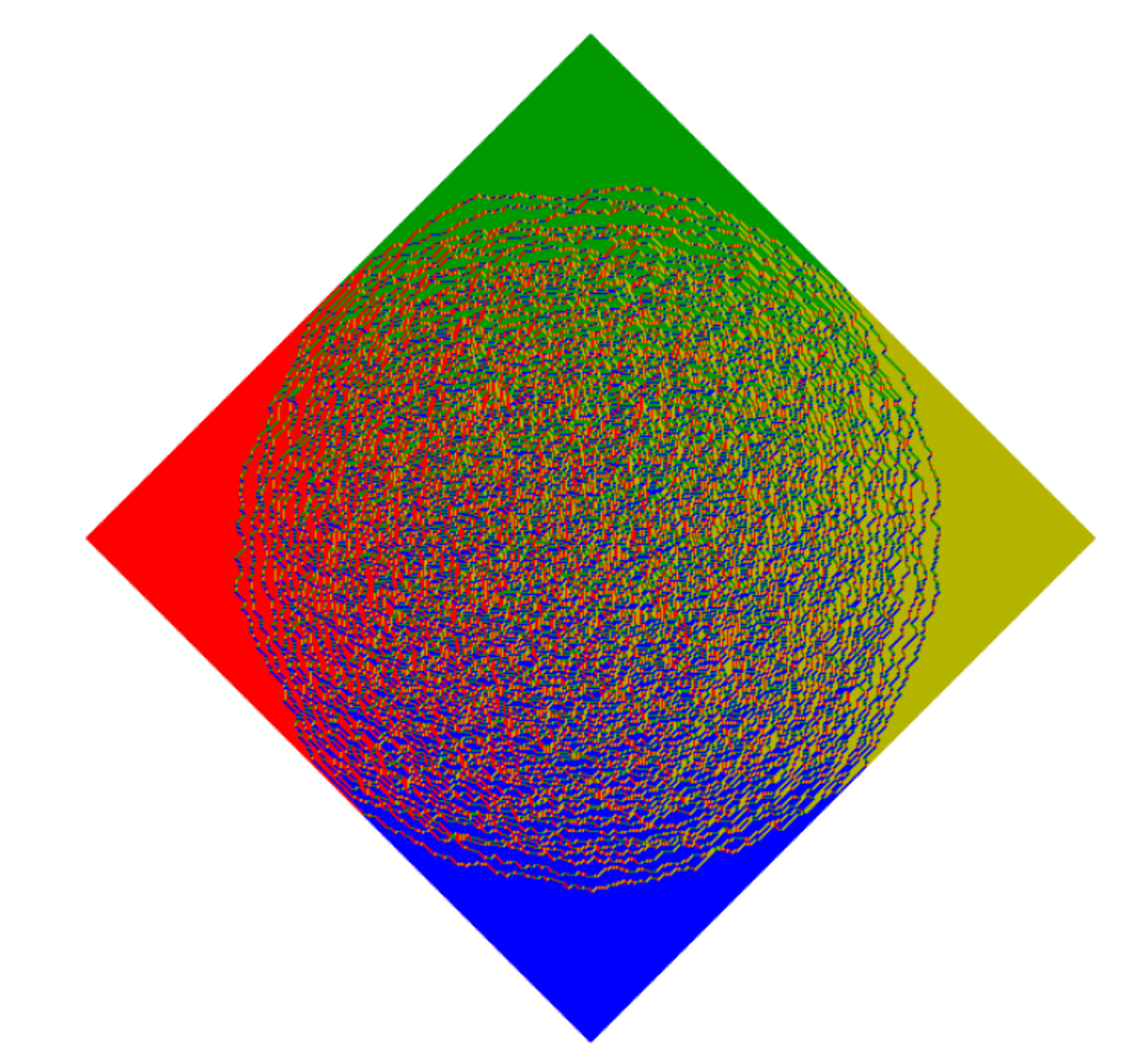}
\qquad
\includegraphics[scale=.35]{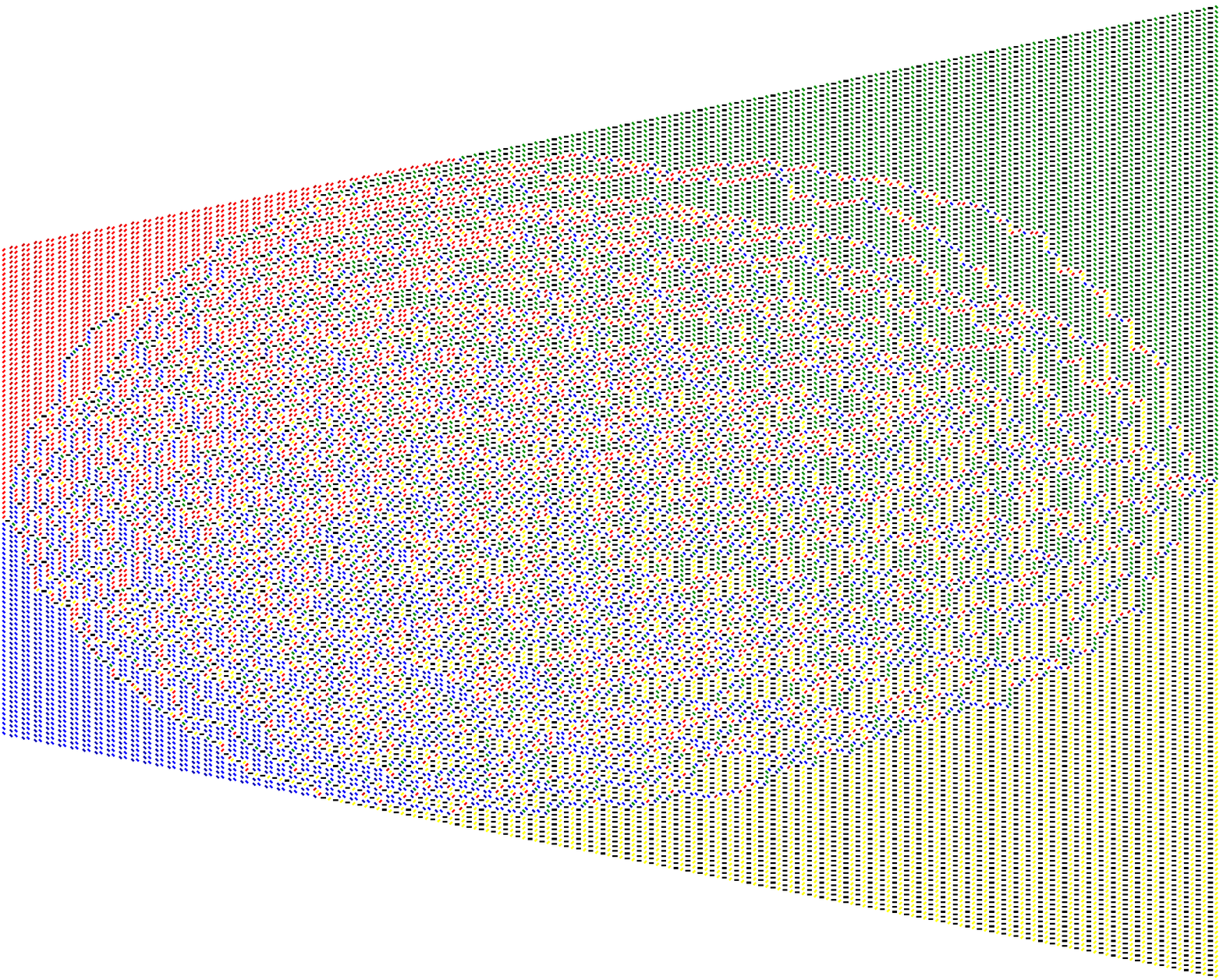}
  \caption{Uniform domino tiling of Aztec diamond (left) and tower graph~(right).}
   \label{fig:rocket}
 \end{center}
\end{figure}

A \emph{dimer cover}, or \emph{perfect matching} of a finite, bipartite, planar graph $\mathcal{G}$ is a subset of edges such that each vertex is incident to exactly one edge. Suppose in addition~$\mathcal{G}$ is equipped with edge weights $\nu : E \rightarrow  \mathbb{R}_{+}$. Then the \emph{dimer model} on~$\mathcal{G}$ is the probability measure on perfect matchings~$m$ defined by
$$\mathbb{P}(m) =\frac{\prod\limits_{e \in m} \nu(e)}{\sum\limits_{m'} \prod\limits_{e \in m'} \nu(e)} .$$
Dimer covers can also be viewed as stepped surfaces via \emph{Thurston's height function} which is defined on the faces of the graph, and the asymptotic behavior of the dimer model can be described via the corresponding random surfaces; see the surveys \cite{gorin2021lectures, Kenyon2007Lecture, kenyon2003introduction} and references therein.

In this work we focus on the uniformly weighted Aztec diamond~$A_n$, or equivalently on uniformly random domino tilings of a certain domain, see Figure~\ref{fig:rocket}, left. For~$n$ large,~$\frac{1}{n} A_n$ approximates the continuous domain~$\mathcal{A} = \{|x| + |y| < 1\}$. It is known that the disk~$\frac{1}{\sqrt{2}}\mathbb{D}$ inscribed in~$\mathcal{A}$ is the \emph{liquid region}, also known as the \emph{rough} region; outside of this region the height function~$h_n$ asymptotically does not fluctuate, whereas in its interior, long range correlations of~$h_n - \mathbb{E}[h_n]$ persist~\cite{jockusch1998random, cohn-elki-prop-96}. In particular, there is a nontrivial complex structure on the liquid region obtained by solving the \emph{complex Burgers equation}, and height fluctuations~$h_n - \mathbb{E}[h_n]$ converge to a Dirichlet Gaussian free field \emph{in this complex structure} on~$\frac{1}{\sqrt{2}}\mathbb{D}$, see~\cite{bufetov2016fluctuations}.

This phenomenon is a special case of a general conjecture of Kenyon-Okounkov~\cite{OkounkovKenyon2007Limit}, which describes the fluctuations of the dimer model height function on subgraphs of~$\mathbb{Z}^2$-periodic lattices with~$\mathbb{Z}^2$-periodic edge weights. Even in this generality, the appearance of Gaussian free field in the complex structure induced by the solution of a certain Burgers equation is expected. Proofs of the Kenyon-Okounkov conjecture have been obtained in a variety of special cases using algebraic and probabilistic techniques \cite{BorFerr2008DF, Duits2011GFF, Petrov2012, Petrov2012GFF,  BoutillierLiSHL, BufetovKnizel, bufetov2016fluctuations, li2021asymptotics, HuangGFF}.

 In some other cases, an appropriate notion of discrete holomorphicity has been used to rigorously analyze the scaling limit of the dimer model. There are several examples of 2-dimensional lattice models in statistical physics, including the critical Ising model on the square lattice and isoradial graphs, 
 which have been analyzed using some discretization of complex analysis  and potential theory; see for example \cite{chelkak2011discrete, chelkak2012universality, ChHI, ChI, HS, SmTowards, SmConfInv, Smirnov2010DiscreteCA} and references therein. Examples where  discretizations of complex analysis
 have been developed and utilized to prove fluctuation results for the dimer model include \cite{Kenyon2000Conformal, Kenyon2004Height, russkikh2018dimers, russkikh2020dominos, Laslier21}.

Recently in~\cite{KLRR, CLR1}, a new type of dimer graph embedding called a \emph{t-embedding} (or \emph{Coulomb gauges}) and a corresponding mapping of the plane called the \emph{origami map} were introduced, and it was shown in~\cite{KLRR} that these structures are preserved under the elementary graph transformations. 
In~\cite{CLR1} a theory of discrete complex analysis on t-embeddings is developed, and in~\cite{CLR2} it is extended and the notion of \emph{perfect t-embeddings} is introduced. Using these ingredients, \cite{CLR2} proves that if a sequence~$(\mathcal{T}_n, \mathcal{O}_n)$ of perfect t-embeddings and associated origami maps satisfies certain regularity conditions, and if the graphs of~$\mathcal{O}_n$ over~$\mathcal{T}_n$ converge to a \emph{maximal surface~$S$ in the Minkowski space~$\mathbb{R}^{2, 1}$}~(the term~\emph{Lorentz-minimal surface} was used instead in an earlier preprint version of this work, and also in~\cite{CLR2} and~\cite{Ch-R}), see Section~\ref{subsec:CLR_result} for the definition, then 
 the gradients of $n$-point correlation functions converge to those of the standard Gaussian free field in the intrinsic metric of the surface~$S$.
In particular, the theorem of~\cite{CLR2} provides a new description of the complex structure inducing the limiting Gaussian free field; this description is not restricted to the~$\mathbb{Z}^2$-periodic setting, whereas the characterization using the Burgers equation is.

\begin{figure}
 \begin{center}
\includegraphics[width=0.45\textwidth]{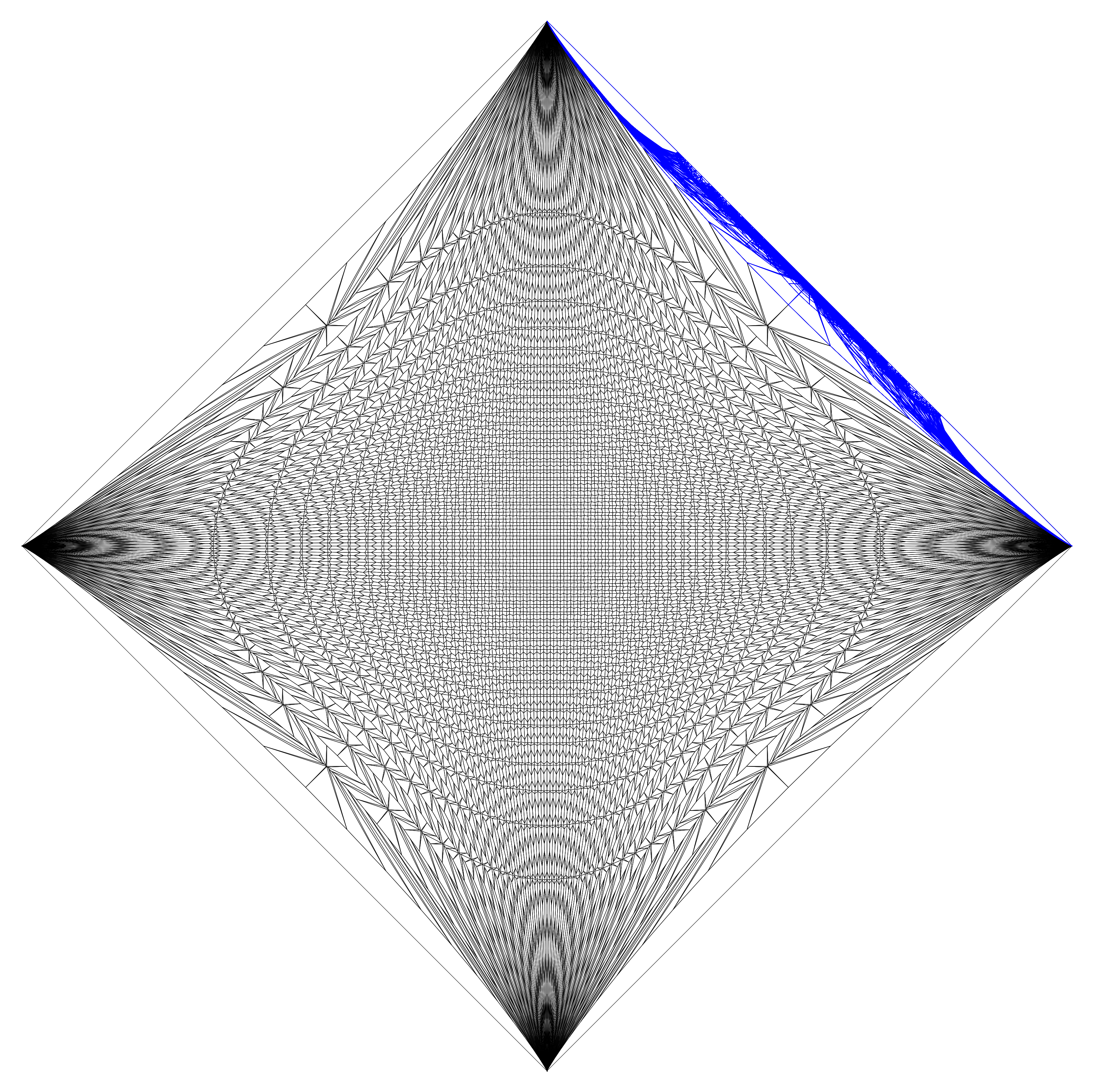}
$\quad\quad$
     \includegraphics[width=0.45\textwidth]{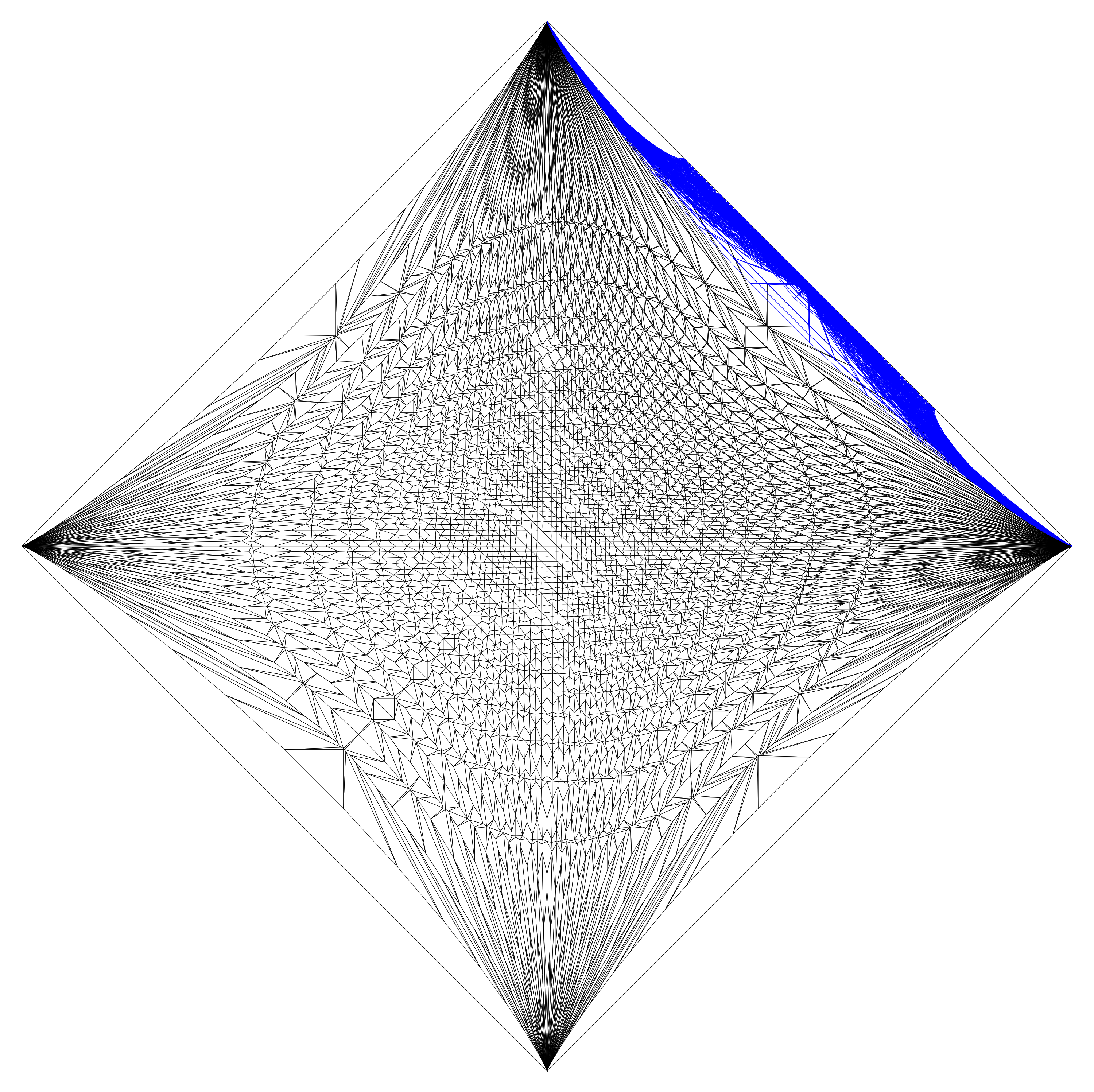}
  \caption{Perfect t-embedding (black) and the image of the embedded graph under the origami map (blue) of Aztec diamond (left) of size $100$ and of tower graph (right) of size $50$.} 
  \label{fig:AdT}
 \end{center}
\end{figure}

The existence of perfect t-embeddings in general, and the validity of the hypotheses of the theorem in any particular case, are left as open questions; this is the motivation for our work.
An important contribution towards that end is the work~\cite{Ch-R}: the authors derive 
 a recurrence relation for  a sequence of perfect t-embeddings for the uniformly weighted Aztec diamond of increasing size. Noting that the recurrence relation can be viewed as a discretization of the wave equation, the authors argue that the scaling limit of the corresponding origami maps is a maximal surface in~$\R^{2,1}$. They also identify this surface and explain that the conformal structure of the limiting maximal surface coincides with the known (Kenyon-Okounkov) conformal structure for the Aztec diamond height fluctuations. 

In this work we study the perfect t-embeddings of uniformly weighted Aztec diamonds, see Figure~\ref{fig:AdT}, with the goal of verifying all assumptions of the main theorem of~\cite{CLR2}. 
 We solve the recurrence relation derived in~\cite{Ch-R} by obtaining a simple expression for the perfect t-embeddings in terms of certain edge inclusion probabilities. Using the solution we give a proof of the convergence of the origami maps to the maximal surface predicted in~\cite{Ch-R}.
We also produce a rigorous verification of~$\LipKd$ and $\ExpFat$, Assumptions~\ref{assumption:Lip} and~\ref{assumption:Exp_fat} described in Section~\ref{subsec:CLR_result}. This provides a first example in which all requirements of the theorem of~\cite{CLR2} are fully verified.
Furthermore, we construct and study perfect t-embeddings of uniformly weighted \emph{tower graphs}, see Figures~\ref{fig:rocket} and~\ref{fig:AdT}, and we find that the graphs of origami maps converge to the \emph{same} maximal surface as in the Aztec diamond case.

We believe our work is an important illustration of the validity and the power of the machinery developed in~\cite{CLR1, CLR2}. On the other hand, our proof of~$\LipKd$ and~$\ExpFat$ for the Aztec diamond, as well as our computation of the conformal structure for the tower graph using perfect t-embeddings, heavily relies on the integrable structures available in these two cases. For example, we utilize the exact formulas for the \emph{inverse Kasteleyn matrix} of the Aztec diamond, which could instead be used to directly prove GFF fluctuations in the Aztec diamond, as is outlined in~\cite{CKBAztec}. In future work, we plan to return to the question of applying the framework of~\cite{CLR1, CLR2} to study the dimer model in other domains, where such rich structure may not be available.

\subsection{Statement of results}
\label{subsec:results}
We divide the discussion of our results into two parts; in the first part we state our results about perfect t-embeddings of the uniformly weighted Aztec diamond. In the second part we state results about perfect t-embeddings of uniformly weighted tower graphs, and we describe a connection between the two sequences of embeddings.

The main objects of our study are perfect t-embeddings and the origami maps associated to weighted bipartite graphs. Let us start with a brief definition of these objects and refer the reader to Section~\ref{sec:definitions} for precise definitions. Suppose we have a planar bipartite graph~$\mathcal{G} = (V, E)$ which is equipped with a set of positive edge weights~$\{\nu(e)\}_{e \in E}$. We denote by~$\mathcal{G}^*$ the \emph{augmented dual} of~$\mathcal{G}$, see Figure~\ref{fig:augmented}.
\begin{itemize}
\item \textbf{Perfect t-embeddings.} A \emph{t-embedding}~$\mathcal{T}$ of~$\mathcal{G}$ is a proper embedding of the augmented dual graph~$\mathcal{G}^*$ into~$\mathbb{C}$ satisfying \emph{the angle condition}, and a certain \emph{gauge equivalence} between geometric edge weights and the given edge weights~$\nu(e)$. A t-embedding is \emph{perfect} if in addition the image under~$\T$ of the outer vertices of~$\mathcal{G}^*$ is a tangential polygon to the unit circle, and edges between outer vertices and interior ones lie on angle bisectors. See Figure~\ref{fig_period} for an example. 
\item \textbf{Origami maps.} The origami map~$\mathcal{O}(z)$ is a mapping from the subset of~$\mathbb{C}$ bounded by the image under~$\mathcal{T}$ of the outer vertices of~$\mathcal G^*$, defined up to rotation and translation, which can be obtained from~$\T$ as follows; first choose a reference face~$v_0$ which is fixed by~$\Or$; then for~$z$ in another face~$v$, to compute~$\mathcal{O}(z)$ one reflects across each edge of the t-embedding along a face path from~$v$ to~$v_0$. Well-definedness of this procedure follows from the angle condition.
\end{itemize}

The main goal of this paper is to verify the assumptions 
of~\cite[Theorem 1.4.]{CLR2}. The assumptions of the main result of~\cite{CLR2} are rather technical. We postpone the precise statement until Section~\ref{subsec:CLR_result}. Instead we provide here an imprecise statement of the assumptions of the theorem. The theorem assumes the existence of a sequence of \emph{perfect} t-embeddings~$\mathcal{T}_n$ of bipartite graphs~$\mathcal{G}_n$, with origami maps~$\Or_n$, satisfying the following three properties: 
\begin{enumerate}[(I)]
\item The pair $(\T_n,\Or_n) \to (z,\vartheta(z))$, as $n\to \infty$, where the graph $(z,\vartheta(z)) \in\mathbb{R}^2\times\mathbb{R}$ is a maximal surface in the Minkowski space~$\R^{2,1}$. \label{ass:informal_lorentz}
\item At the discrete level the origami map is Lipschitz continuous with constant strictly less than one. \label{ass:informal_lip}
\item For almost every face, the radius of the largest circle which can be inscribed in the face cannot decay exponentially fast as~$n \rightarrow \infty$.\label{ass:informal_exp-fat}
\end{enumerate}
 It was shown in~\cite{CLR2}  that under these assumptions the gradients of $n$-point correlation functions converge to those of the standard Gaussian free field in the intrinsic metric of the maximal surface. See Theorem~\ref{thm:CLR2thm} for the precise statement.

\begin{figure}
 \begin{center}
 \includegraphics[width=0.6\textwidth]{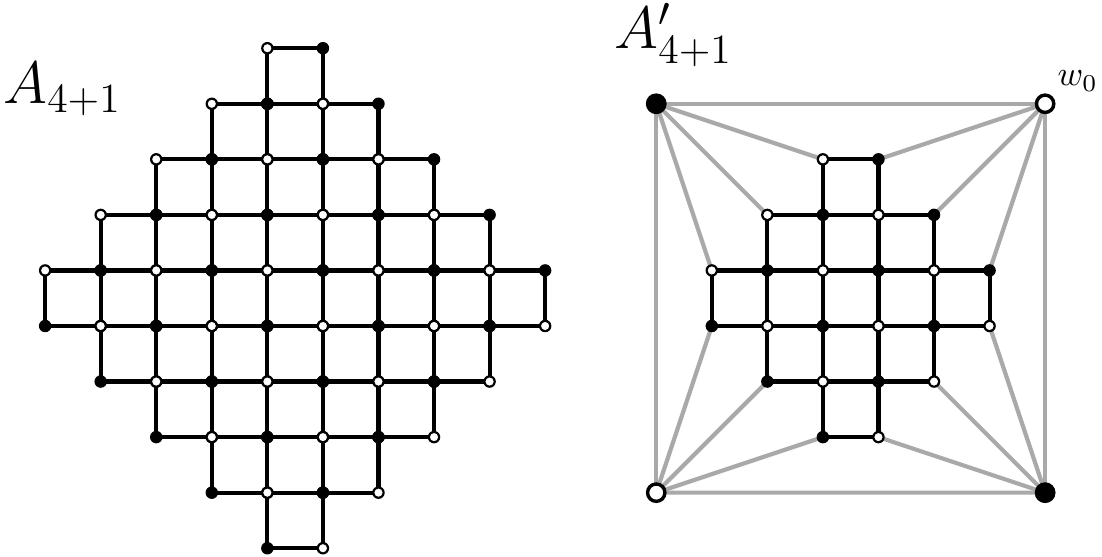}
 $\quad$
 \includegraphics[width=0.3\textwidth]{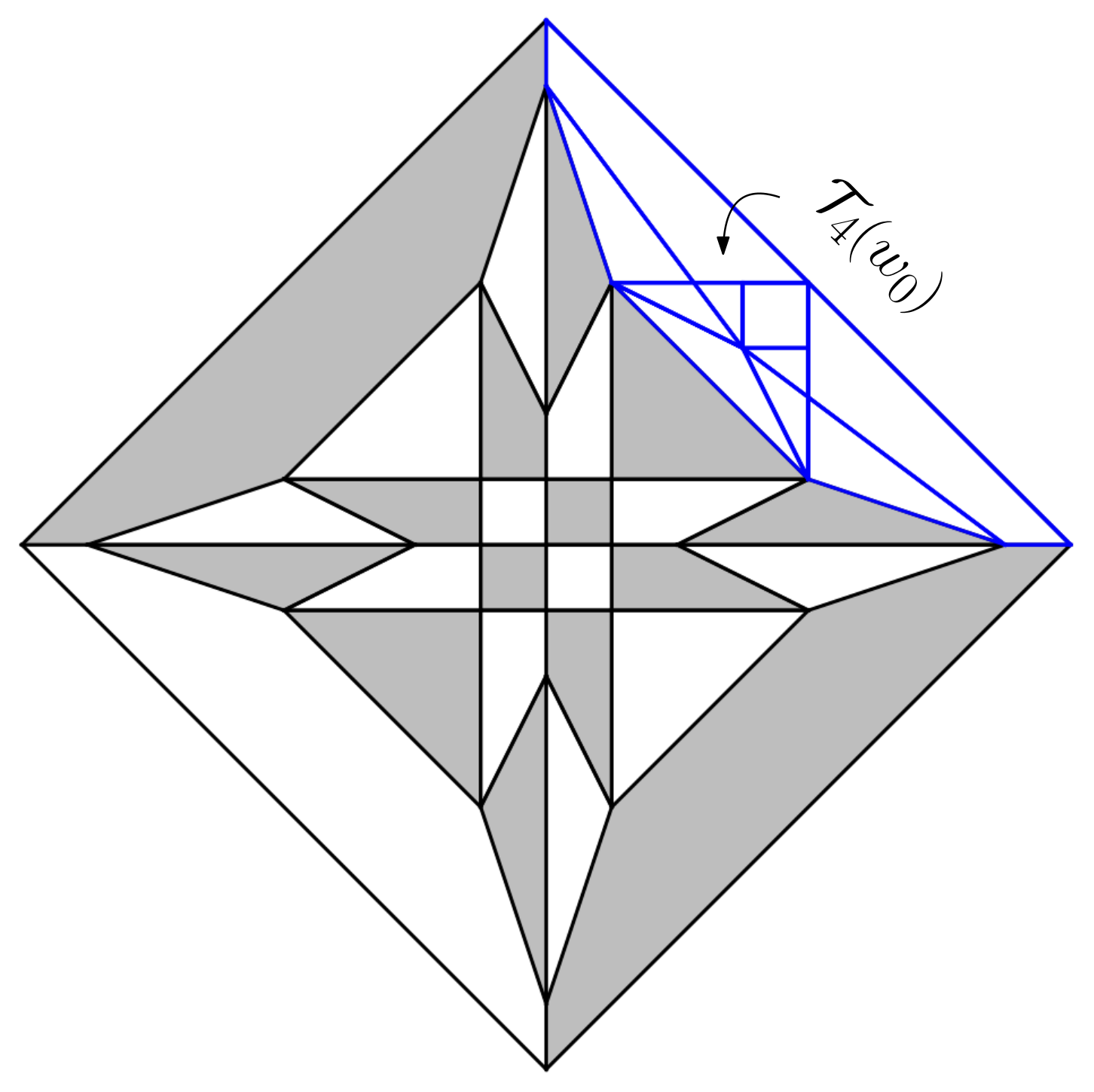}
  \caption{{\bf Left:} Aztec diamond $A_{4+1}$ of size $4+1$. 
  {\bf Middle:} Its reduced Aztec diamond $A'_{4+1}$.  {\bf Left and middle:} Weights on black edges are~$1$ and weights on grey ones are~$2$.  {\bf Right:} perfect t-embedding~$\T_{ 4}((A'_{4+1})^*)$ shown in black and its origami map~$\Or_4((A'_{4+1})^*)$ is shown in blue. The face $\T_4(w_0)$ is the reference face for the origami map.
  }\label{fig_period}
 \end{center}
\end{figure}

\subsubsection{Uniformly weighted Aztec diamond.} 
Our main theorem, Theorem~\ref{thm:Aztec_assumptions}, states that a certain sequence of perfect t-embeddings of the uniformly weighted Aztec diamond satisfies all assumptions of~\cite[Theorem 1.4.]{CLR2}. In order to show this, we found formulas for the perfect t-embeddings and their origami maps in terms of edge probabilities of the dimer model on the uniformly weighted Aztec diamond. These expressions yield simple exact contour integral formulas for perfect t-embeddings of the uniformly weighted Aztec diamond and their origami maps. 
Using these formulas, we show that under the perfect t-embedding, all frozen regions asymptotically are collapsed to points.  

Let us index faces of the shifted square lattice~$(\mathbb{Z}+\frac{1}{2})^2$ by integer pairs~$(j, k)$. Then the Aztec diamond~$A_n$ of size~$n$ is the subgraph defined by the union of faces~$\{(j, k) : |j| + |k| < n\}$, equipped with edge weights which are identically~$1$. The \emph{reduced Aztec diamond} of size~$n+1$, denoted~$A_{n+1}'$, is obtained from~$A_{n+1}$ by a sequence of contractions of degree two vertices on the boundary, see Section~\ref{sec:aztec} for a precise definition. The augmented dual~$(A_{n+1}')^*$ has four boundary vertices indexed by~$(n, 0), (0, n), (-n, 0), (0, -n)$. 
 We study a perfect t-embedding~$\T_n$ of~$A_{n+1}'$
 constructed in~\cite{Ch-R} (we review this construction in Section~\ref{sec:aztec}), such that these four vertices of~$(A_{n+1}')^*$ map to~$1, \i, -1, -\i$, respectively. 
  Define~$\Or_n$ as the corresponding origami map, specified by the reference face of~$\T_n$ adjacent to the boundary vertices~$\T_n(n,0)$ and~$\T_n(0,n)$. See Figure~\ref{fig_period} for an example with~$n=4$. We will show that the limit of the origami maps~$\Or_n$ lies on a segment between the points~$1$ and~$i$, see Figure~\ref{fig:AdT} for the illustration of this convergence. So it is natural to define another version of the origami map by $\Or'_n=\e^{i\tfrac{\pi}{4}}\left(\Or_n - \tfrac{1+i}{2}\right),$
which is a composition of a translation and rotation of the initial one. Note that $\Or_n'$ maps the vertices ~$(n, 0), (0, n), (-n, 0), (0, -n)$ of~$(A_{n+1}')^*$ to~$1/\sqrt{2}$,~$-1/\sqrt{2}$,~$1/\sqrt{2}$, and~$-1/\sqrt{2}$, respectively.

Since we postpone stating precisely the assumptions~\eqref{ass:informal_lorentz}-\eqref{ass:informal_exp-fat} until Section~\ref{sec:definitions}, we give an informal statement of our main result, and refer the reader to Theorem~\ref{thm:main_thm} for a precise statement. 
\begin{theorem}[Informal version]\label{thm:Aztec_assumptions}
Let~$\T_n$ be the sequence of perfect t-embeddings of the reduced uniformly weighted Aztec diamonds~$A_{n+1}'$, with corresponding origami maps~$\Or_n'$. All assumptions of the main theorem of~\cite{CLR2} hold for the sequence~$\T_n$.
\end{theorem}

The perfect t-embeddings~$\T_n$ of the uniformly weighted Aztec diamonds are described in~\cite{Ch-R} in terms of a recurrence relation. We recall this construction in Section~\ref{sec:aztec_rec}. An important discovery of the present paper is that the recurrence for~$\T_n$ can be solved in terms of other dimer model observables, which in turn have representations as double contour integrals. 
 We use this exact solvability as a main ingredient to prove Theorem~\ref{thm:Aztec_assumptions}.
In the rest of this section we state our formula for~$\T_n$, and then describe how it is used to prove that~\eqref{ass:informal_lorentz}-\eqref{ass:informal_exp-fat} hold.

There are four types of edges in the Aztec diamond~$A_n$; for a face~$(j, k)$ of~$A_n$ with~$j + k + n$ odd, the East edge is the vertical edge on the east boundary of the face. We similarly define North, West, and South edges. 
The following theorem provides a surprising relation between the perfect t-embedding~$\T_n$ and origami map~$\Or_n'$ of~$A_{n+1}'$ and certain edge inclusion probabilities of~$A_n$. Note that the interior faces~${\{(j, k) : |j| + |k| < n\}}$ of the Aztec diamond~$A_n$ of size~$n$ are in bijection with the interior vertices of~$(A_{n+1}')^*$. Denote by~$\mathcal{T}_{n}(j, k)$ the image under the perfect t-embedding of vertex~$(j, k) \in (A_{n+1}')^*$ and by~$\Or_n'(j, k)$ the image of~$\mathcal{T}_{n}(j, k) \in \C$ under~$\Or_n'$.
\begin{theorem}\label{thm:T_edge_probabilities}
For~$|j| + |k| < n$ and~$j + k + n$ odd, let~$p_E(j, k, n)$ ($p_N(j, k, n)$, $p_W(j, k, n)$, $p_S(j, k, n)$, respectively) denote the probability that the edge on the East (North, West, South, respectively) boundary of the face~$(j, k)$ is present in a uniformly random dimer cover of~$A_n$. Then for any interior vertex~$(j, k) \in (A_{n+1}')^*$ with~$j + k + n$ odd, the perfect t-embedding~$\T_{n}$ of~$A_{n+1}'$ and its origami map~$\Or_{n}'$ are given by
\begin{equation}\label{eq:T_p_e}
\T_{n}(j,k)=p_E(j,k,n)+\i p_N(j,k,n)-p_W(j,k,n)-\i  p_S(j,k,n)
\end{equation}
and
\begin{align}\label{eq:O_p_e}
\Or_{n}'(j,k)&= \tfrac{1}{\sqrt{2}} \left( p_E(j,k,n)- p_N(j,k,n)+p_W(j,k,n) -  p_S(j,k,n) \right) \\
&+ \i \tfrac{1}{\sqrt{2}} \left(  -1 + p_E(j,k,n) + p_N(j,k,n) + p_S(j,k,n) + p_W(j,k,n)  \right). \notag
\end{align}
\end{theorem}

\old{
\begin{figure}
\centering
\includegraphics[scale=.7]{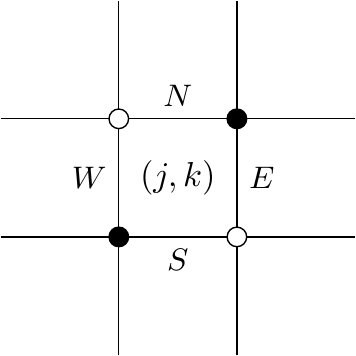}
\caption{The edges on the East, North, West, and South boundary of face~$(j, k)$. The edge probability $p_E(j, k, n)$ is the probability that the East edge of face $(j, k)$ is contained in a random dimer cover of $A_n$, and similarly for~$p_N, p_W, p_S$.}\label{fig:pE}
\end{figure}}

A direct corollary of the above theorem is that each frozen region in the Aztec diamond is mapped to a point under~$\T_n$ in the large~$n$ limit. The East frozen region of~$A_n$ is the region which asymptotically consists of only East edges, or in other words, where the edge probabilities~$p_E\to 1$, and~$p_N,p_W,p_S\to 0$, as~$n\to \infty$. The North, West and South frozen regions are defined similarly. It is well known~\cite{cohn-elki-prop-96}, and we recall this fact in Proposition~\ref{prop:edge_prob_convergence} in the text, that the edge probabilities converge uniformly on compact subsets. 
\begin{corollary}\label{Cor_main_frozen}
As~$n \rightarrow \infty$, the perfect t-embedding of~$A_{n+1}'$ maps each frozen region to a point. More precisely, for all~$j+k+n$ odd,
one has
\begin{align*}
(\mathcal{T}_{n}(j, k),\mathcal{O}_{n}'(j, k))&=o(1)+
\begin{cases}
(1,1/\sqrt{2} ) &\text{ if }\, (j/n, k/n) \in \text{East frozen region}\\
(i,-1/\sqrt{2} ) &\text{ if }\, (j/n, k/n) \in \text{North frozen region}\\
(-1,1/\sqrt{2} ) &\text{ if }\, (j/n, k/n) \in \text{West frozen region}\\
(-i,-1/\sqrt{2} ) &\text{ if }\, (j/n, k/n) \in \text{South frozen region}
\end{cases} 
\end{align*}
where the~$o(1)$ error is uniform for~$(j/n, k/n)$ in compact subsets as~$n\to\infty$. 
\end{corollary}

The edge probabilities of the uniformly weighted Aztec diamond are well understood. Indeed, the edge probabilities can be expressed in terms of the inverse Kasteleyn matrix, which is known to admit a double integral formula. Theorem~\ref{thm:T_edge_probabilities} therefore provides us with expressions of~$\T_n$ and~$\Or_n'$ in terms of double integrals. The integral expressions allows for asymptotic analysis, using a classical steepest descent analysis. To prove Theorem~\ref{thm:Aztec_assumptions} we need the first two terms in the series expansion of~$\T_n$.

The first order term in the asymptotic expansion of~$\T_n$ and~$\Or_n'$ is enough to prove that assumption~\eqref{ass:informal_lorentz} holds. More precisely, using the leading order terms we prove a conjecture stated in~\cite[Section 3]{Ch-R}, stated as a proposition below.
\begin{proposition}\label{prop:eqn:to_conv}
Let~$ \mathcal{T}_{n}(j,k)$ denote the image of the vertex~$(j, k) \in (A_{n+1}')^*$ under the perfect t-embedding of~$A_{n+1}'$ and~$\mathcal{O}_{n}'(j,k) $ denote the image of this vertex under the origami map. Then
\begin{align}\label{eqn:tconv}
\mathcal{T}_{n}(j,k) &= z(j/n, k/n) + o(1) \\
\mathcal{O}_{n}'(j,k) &= \vartheta(j/n, k/n) + o(1) \label{eqn:oconv}
\end{align}
where~$z(x,y)$ and~$\vartheta(x, y)$ are explicit smooth functions (see Section~\ref{sec:T_n_conv} for the precise definition of the functions~$z, \vartheta$) defined on the rescaled Aztec domain~$\mathcal{A} \coloneqq \{|x| + |y| < 1\}$. The~$o(1)$ error is uniform for~$(j/n, k/n)$ in compact subsets of~$\mathcal{A}$.
\end{proposition}

It was shown in~\cite{Ch-R} that the surface~$\{(z(x,y),\vartheta(x,y)) : (x,y) \in \mathcal{A}\}$ is indeed a maximal surface in the Minkowski space~$\mathbb{R}^{2,1}$.

Since the assumptions~\eqref{ass:informal_lip} and~\eqref{ass:informal_exp-fat} are conditions on the discrete level, the leading term in the asymptotic expansion of~$\T_n$ is not sufficient to prove these assumptions. Instead we need to go to the second order term. We use it to prove a certain \emph{rigidity} of the graph~$\T_n\left((A_{n+1}')^*\right)$.
\begin{lemma}\label{lem:bdd}
Let~$\T_n$ be the perfect t-embedding of the reduced Aztec diamond~$A_{n+1}'$. In compact subsets of the interior of the domain covered by~$\T_n\left((A_{n+1}')^*\right)$, the edge lengths of~$\T_n$ are of order~$1/n$, and the angles are bounded away from~$0$ and~$\pi$. 
\end{lemma}
For a precise version of the above lemma we refer to Lemmas~\ref{lem:bound_edges} and~\ref{lem:bound_angles}. In Section~\ref{sec:lip_expfat} we prove that the above rigidity of the graph is sufficient to prove assumptions~\eqref{ass:informal_lip} and~\eqref{ass:informal_exp-fat}. In fact, we prove that such rigidity is sufficient for any bipartite graph under some mild assumptions.

\begin{figure}
 \begin{center}
 \includegraphics[width=0.35\textwidth]{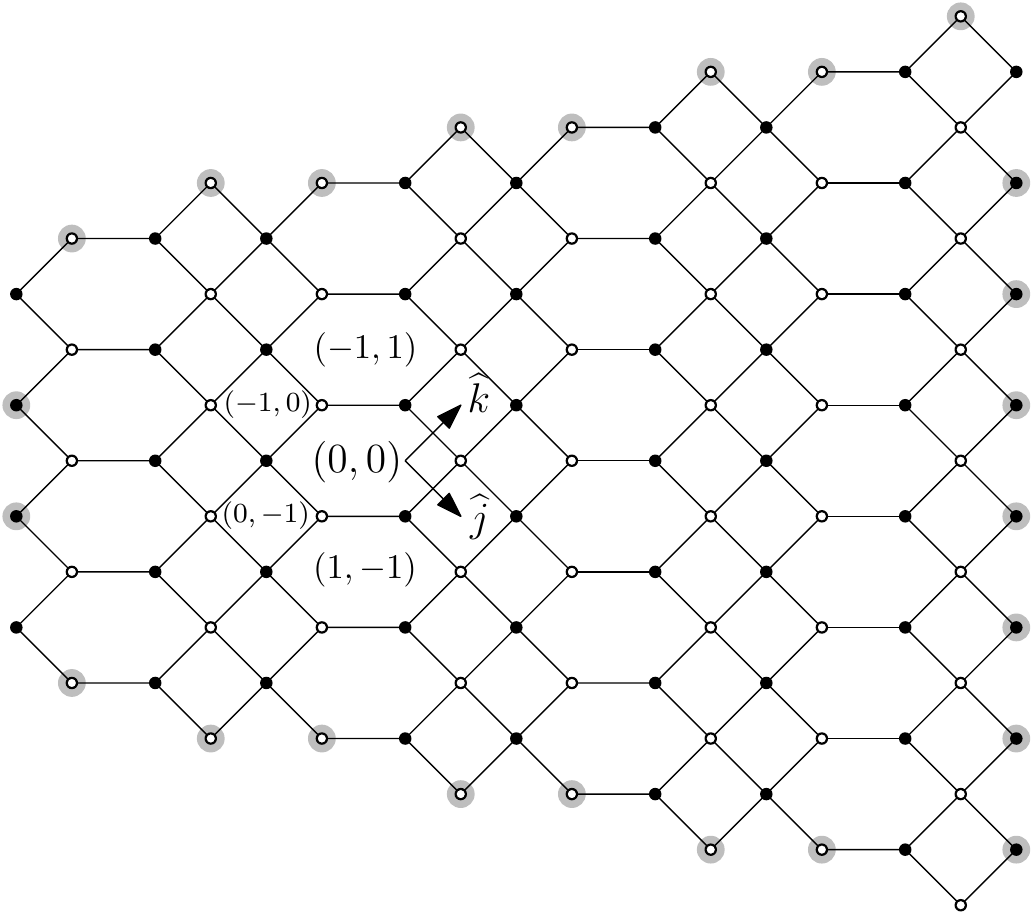}
 $\quad\quad$
  \includegraphics[width=0.45\textwidth]{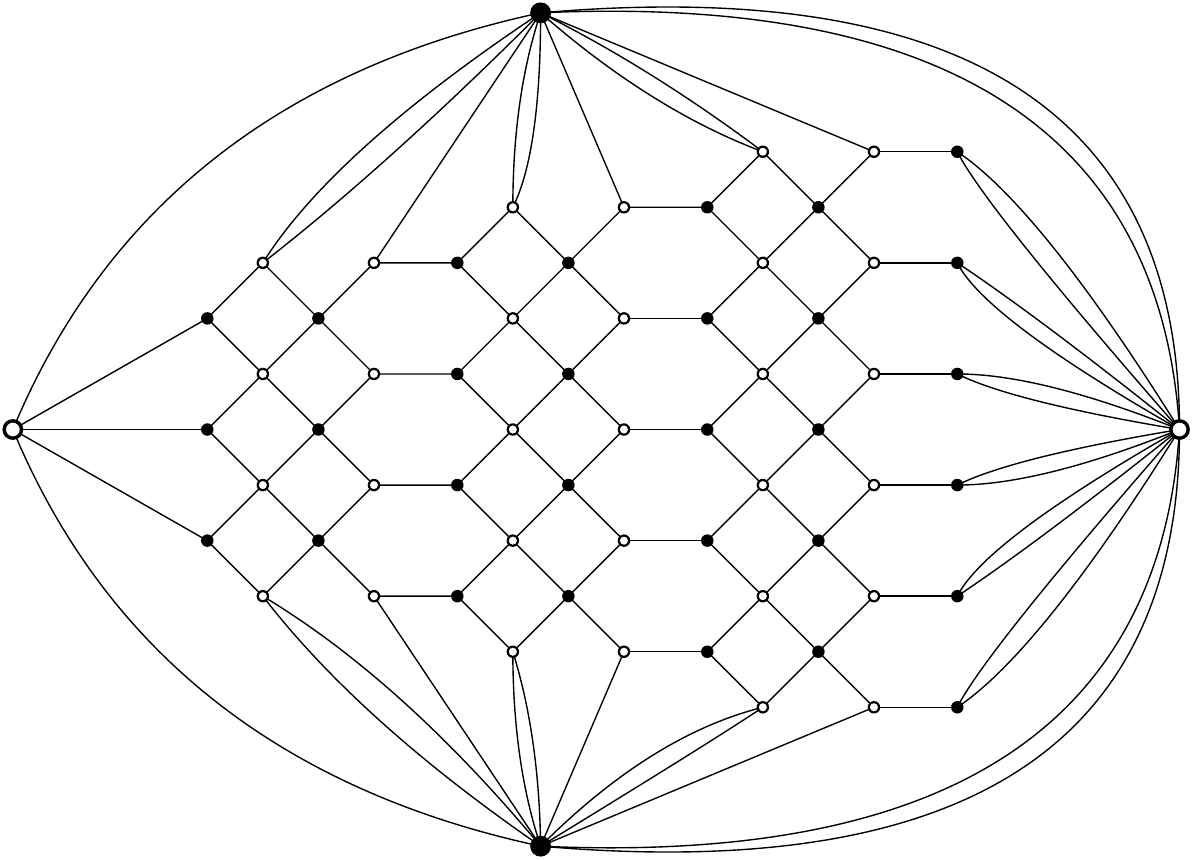}
   \includegraphics[width=0.4\textwidth]{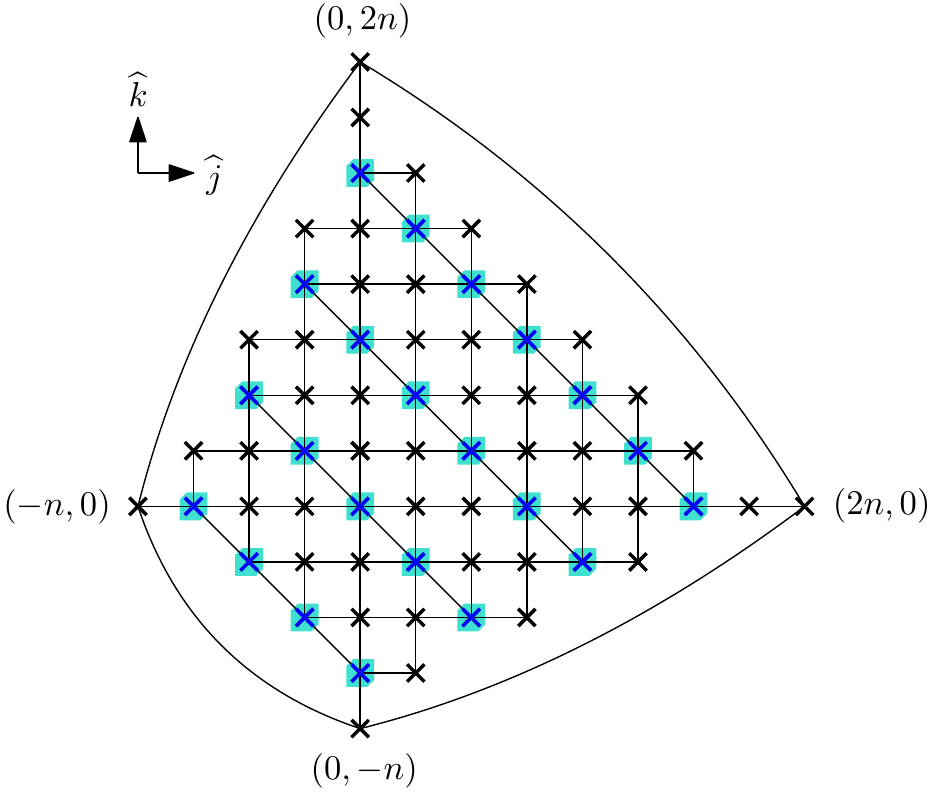}
    $\quad\quad\quad\quad$
   {\includegraphics[width=0.35\textwidth]{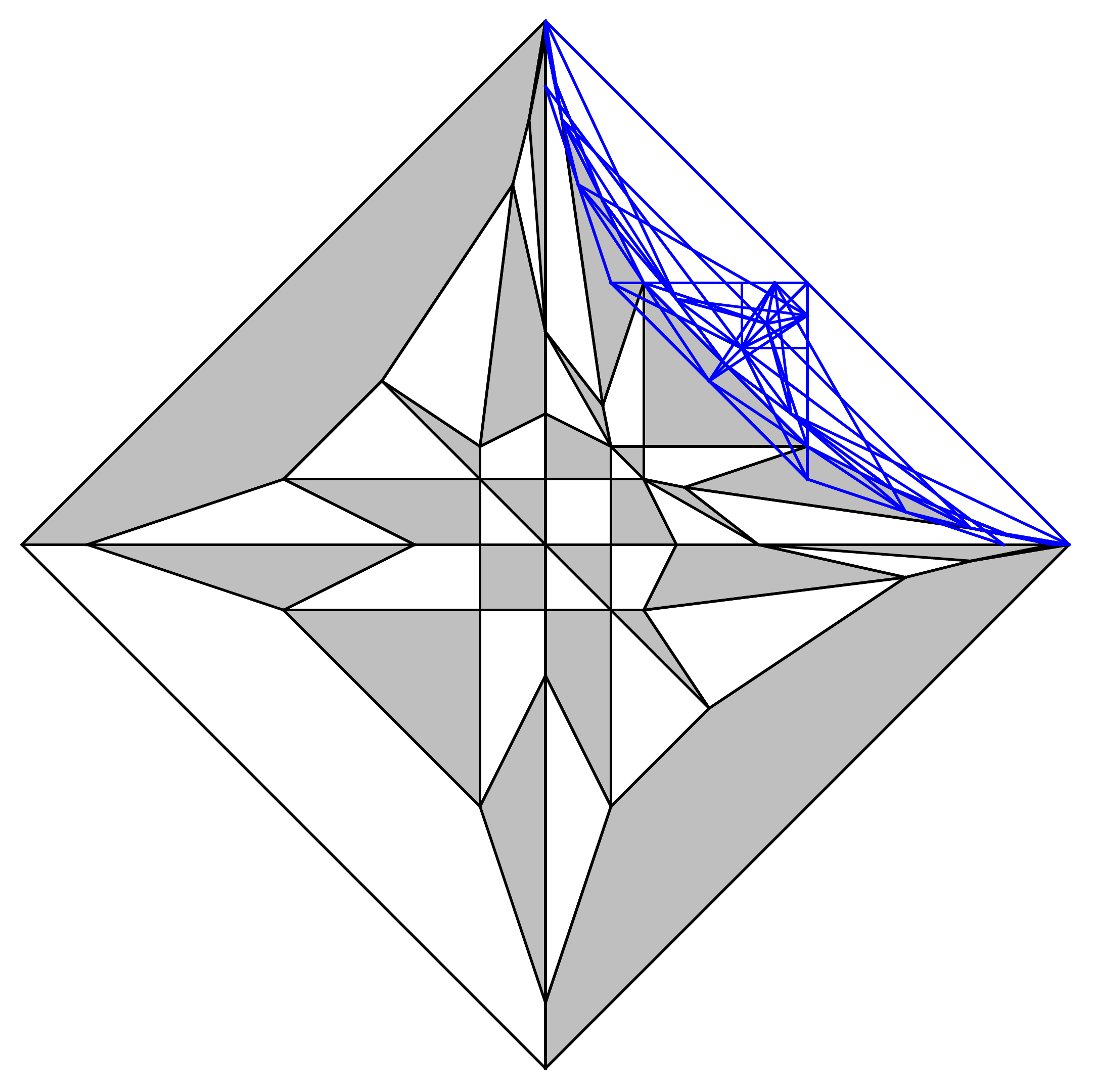}}
 \end{center}
\caption{{\bf Top:} A tower graph~$\TG_4$ of size~$4$ and its reduced one~$\TG'_4$. To obtain our coordinates on faces, one may shrink horizontal edges to~$0$, and then use coordinates~$j, k$ corresponding to basis vectors~$\hat{j}, \hat{k}$ shown in the figure. To get a reduced graph contract degree two vertices emphasised in grey.\\ {\bf Bottom: } Augmented dual of ~$\TG'_4$ (vertices of the augmented dual corresponding to hexagonal faces of~$\TG_4$ are emphasized in turquoise).
Perfect t-embedding~$\widetilde\T_4$ (black) together with the origami map ~$\widetilde\Or_4$  (blue)
of the augmented dual~$(\TG'_4)^*$.}\label{fig:tower}
\end{figure}


\subsubsection{Uniformly weighted tower graphs}
The other sequence of graphs we study in this work are \emph{tower graphs}; these graphs were introduced in~\cite{borodin2015random}. We show that a sequence of perfect t-embeddings of uniformly weighted tower graphs can be expressed in terms of perfect t-embeddings of the uniformly weighted Aztec diamond. As a corollary, we obtain that these perfect t-embeddings satisfy the assumption~\eqref{ass:informal_lorentz} that graphs of the origami maps converge to a maximal surface in the Minkowski space~$\R^{2, 1}$.

The faces of the tower graph have the structure of the square lattice; see Figure~\ref{fig:tower} for the exact coordinate system we use to index faces. The set of faces of a tower graph $\TG_n$ of size $n$ is given by $(j, k)$ such that 
\begin{equation}\label{eq:tower_coord_intro}
\begin{cases}
 j+k < 2n \quad &\text{ for } j, k \geq 0,\\
 |j|+|k| <n \quad &\text{ for } j, k \leq 0,\\
 k-2j < 2n \quad &\text{ for } j<0, k>0,\\
 j-2k < 2n \quad &\text{ for } j>0, k<0 .
\end{cases}
\end{equation}
We construct and study a sequence of perfect t-embeddings and origami maps~$\widetilde\T_{n}, \widetilde\Or_{n}$ of the sequence of  uniformly weighted \emph{reduced tower graphs}~$\TG_n'$ of size~$n$, see Figure~\ref{fig:tower} for an example. As for the 
Aztec diamond, we introduce a shifted and rotated origami map~${\widetilde\Or'_{n} \coloneqq e^{\i \frac{\pi}{4}}(\widetilde\Or_{n}-\frac{1+\i}{2})}$.  For a dual vertex~$(j,k)\in (\TG_n')^*$, denote by~$\widetilde\T_{n}(j, k), \widetilde\Or'_{n}(j, k)$ the image of that vertex under~$\widetilde\T_{n}, \widetilde\Or'_{n}$, respectively. The construction of~$\widetilde\T_{n},\widetilde\Or'_{n}$ utilizes a shuffling algorithm similar to that of the Aztec diamond; however, in this case there are two rounds of spider moves required to construct~$\widetilde\T_{n+1}$ from~$\widetilde\T_{n}$. Thus, it is natural to consider the t-embedding and origami map of an intermediate graph after the first round, so we make a variable change~$m = 2 n$, and we  introduce
functions~$\widetilde{\T}(j, k, m), \widetilde{\mathcal{O}}'(j, k, m)$ which are defined for all integers~$m \geq 1$, and satisfy~$\widetilde\T_{n}(j, k)=\widetilde{\T}(j, k, 2 n)$ and~$ \widetilde\Or_{n}'(j, k)=\widetilde{\Or}'(j, k, 2 n)$. See Section~\ref{sec:tower} for details. As in the case of the Aztec diamond, the boundary vertices of~$(\TG_n')^*$ are mapped to~$1, \i, -1, -\i$ and~$\frac{1}{\sqrt{2}}, - \frac{1}{\sqrt{2}}, \frac{1}{\sqrt{2}}, -\frac{1}{\sqrt{2}}$ by~$\widetilde\T_n$ and~$\widetilde \Or'_{n}$, respectively.

\begin{figure}
 \begin{center}
 { \includegraphics[scale=.5]{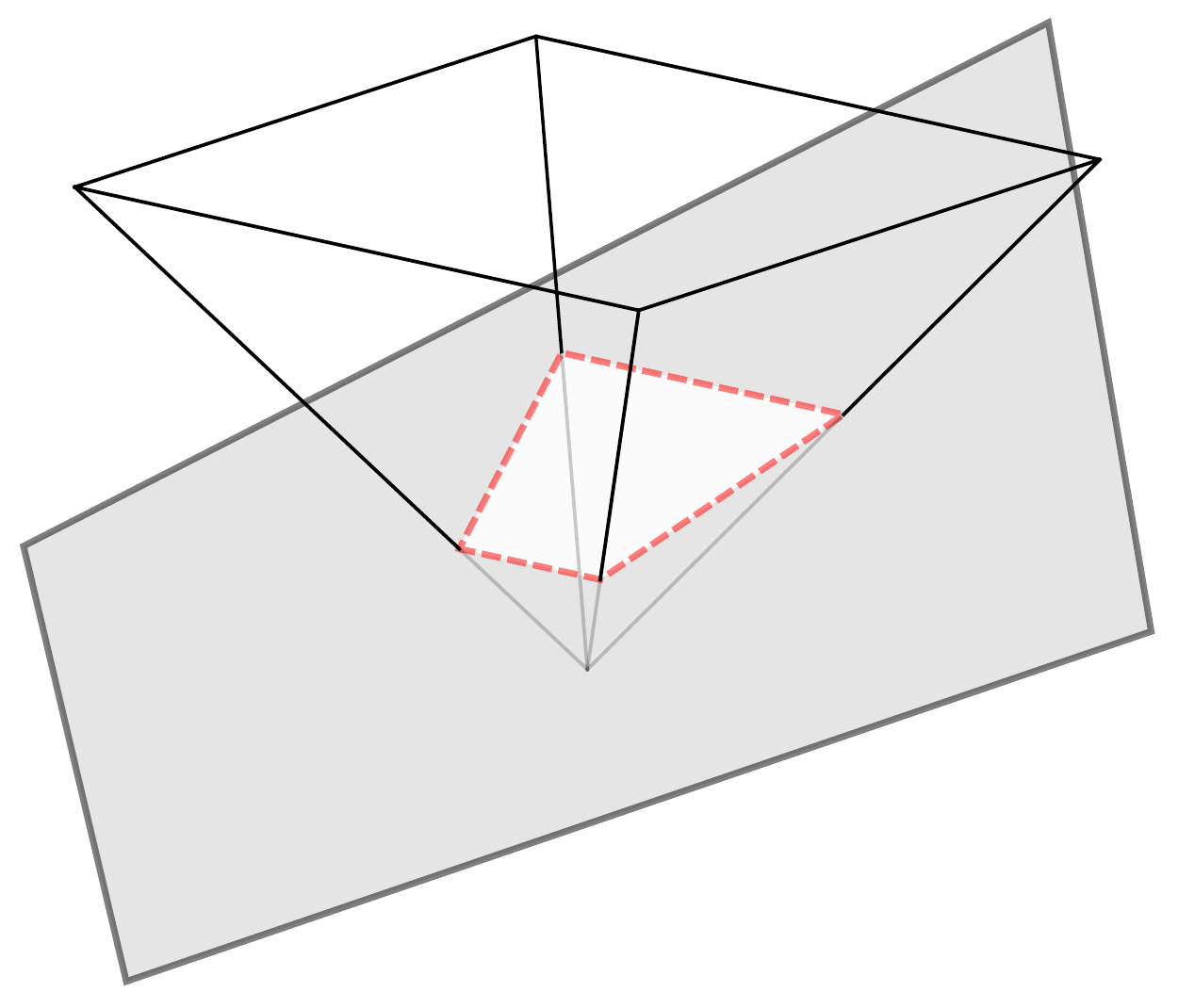} \hspace{20pt}
  \includegraphics[scale=.45]{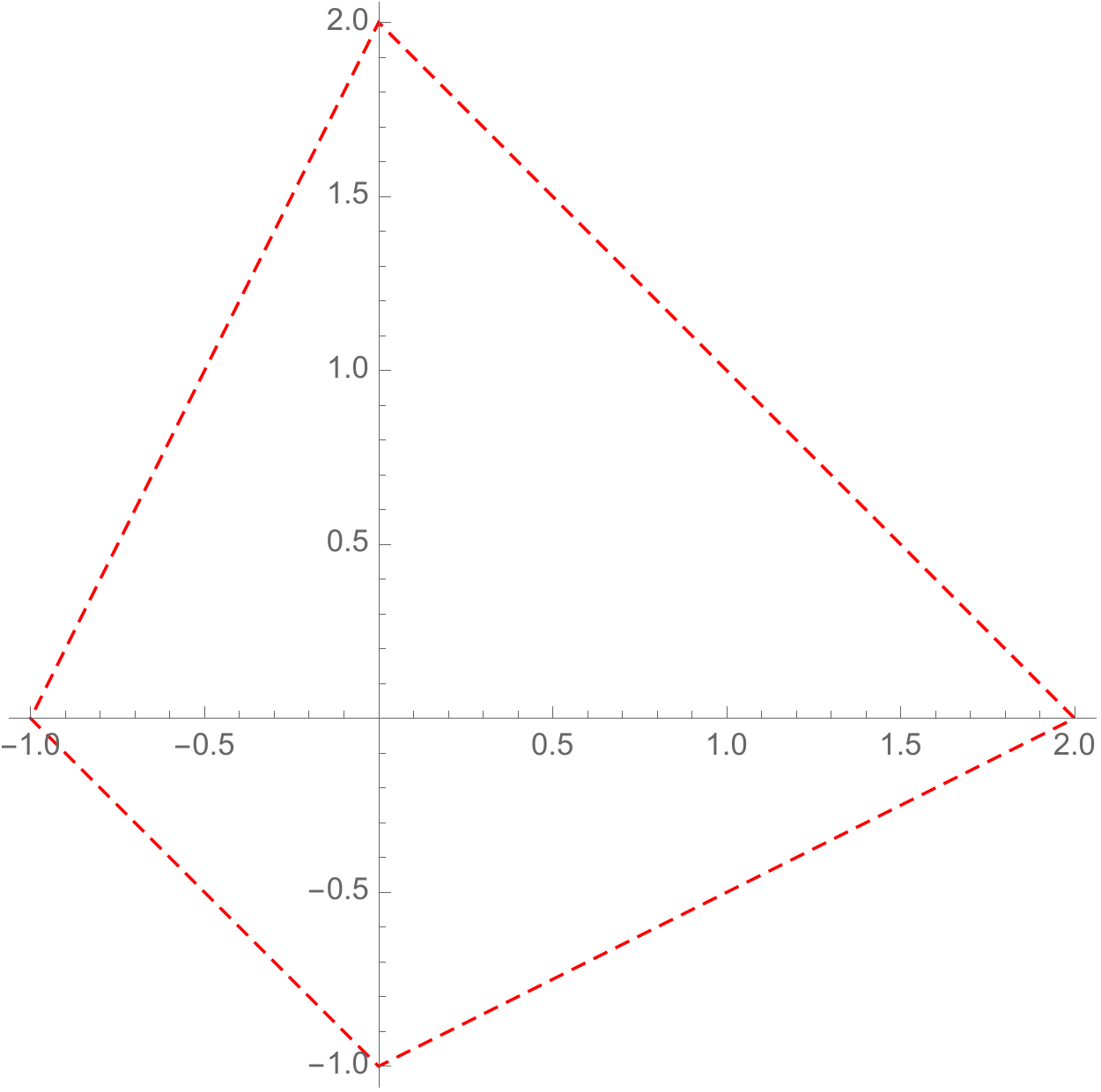}}
  \caption{Left: The plane along which limiting values of~$\T_{n}$ are sampled to obtain~$\widetilde\T$ according to Theorem~\ref{thm:Aztec_Tower}. In the figure on the left, the coordinates~$(x, y, t)$ correspond to integer coordinates~$(\lfloor \frac{x}{\epsilon} \rfloor, \lfloor \frac{y}{\epsilon} \rfloor)$ in the Aztec diamond of size~$n = \lfloor \frac{t}{\epsilon} \rfloor$, for some small~$\epsilon > 0$. The domain of the Aztec t-embedding, in continuous~$(x,y,t)$ coordinates, is the cone~$|x|+|y| <t$, and the plane slicing it is~$t = \frac{1}{3}(x+y+4)$. Right: An illustration of the rescaled augmented dual of the tower graph in the plane (compare with Figure~\ref{fig:tower}).  }
  \label{fig:tower_slice}
 \end{center}
\end{figure}

The following theorem relates the perfect t-embeddings and origami maps of  uniformly weighted tower graphs with those of  uniformly weighted Aztec diamonds. 
\begin{theorem}\label{thm:Aztec_Tower}
Let~$\T_{n}$ denote the perfect t-embedding of the reduced Aztec diamond~$A_{n+1}'$, and let~$\widetilde{\T}$ define the perfect t-embeddings of the reduced tower graphs~$\TG_n'$ as described above. For all interior vertices~$(j, k)$ of~$(A_{n+1}')^*$ with~$j+k+n$ odd
\begin{align*}
\T_{n}(j,k)&=\widetilde{\T}\left(j, \,k, \, \tfrac32 n - \tfrac12(j+k)-\tfrac12\right), \\
\Or_{n}'(j,k)&=\widetilde{\Or}'\left(j, \,k, \,\tfrac32 n - \tfrac12(j+k)-\tfrac12\right) .
\end{align*}
\end{theorem}
The change of coordinates used to write the perfect t-embeddings of tower graphs in terms of those of the Aztec diamonds, obtained by inverting the coordinate change above, is visualized in Figure~\ref{fig:tower_slice}.

Using the theorem above and the convergence from Proposition~\ref{prop:eqn:to_conv}, we obtain the following corollary. Similarly to Theorem~\ref{thm:Aztec_assumptions}, we state the corollary informally here, and refer the reader to Corollary~\ref{cor:tower_coord} in the text for the precise statement.

\begin{corollary}\label{cor:convergence_tower}
The sequence of tower graph t-embeddings and origami maps
satisfies assumption~\eqref{ass:informal_lorentz}. Moreover, the graphs~$(\widetilde\T_n, \widetilde\Or_{n}')$ converge to the same maximal surface which arises as the limit of Aztec diamond t-embeddings and origami maps.
\end{corollary}


See Figure~\ref{fig:surface_both}, which illustrates that the graphs of~$\Or_{n}', \widetilde\Or_{n}'$ both approximate the same surface.

\begin{remark}
The fluctuations of the dimer model height function on certain finite subgraphs of lattices made up of alternating layers of squares and hexagons was studied in~\cite{BoutillierLiSHL}, and their class of graphs includes tower graphs. In Section~\ref{sec:T_n_conv}, we compare the conformal structure of the maximal surface obtained via perfect t-embeddings with the conformal structure for the Gaussian free field derived in \cite{BoutillierLiSHL}, and indeed we find that they match (and furthermore, both match with the conformal structure of Kenyon-Okounkov). In particular, the conformal structures of the Gaussian free field on the Aztec diamond and tower graph are identified under this change of coordinates, and thus so are the limiting height fluctuations.
\end{remark}

\begin{remark}
Using Theorem~\ref{thm:Aztec_Tower}, and the fact that asymptotically frozen regions map to frozen regions under the change of coordinates there (see the discussion after Corollary~\ref{cor:tower_coord}, in particular Claim~\ref{claim:same_structure}), one obtains an analog of Corollary~\ref{Cor_main_frozen} for tower graphs; each frozen region maps to a point as~$n \rightarrow \infty$.
\end{remark}

\begin{figure}
 \begin{center}
\begin{minipage}{0.7\columnwidth}
  \hspace{-3.cm}
   {\includegraphics[width=1.2\textwidth]{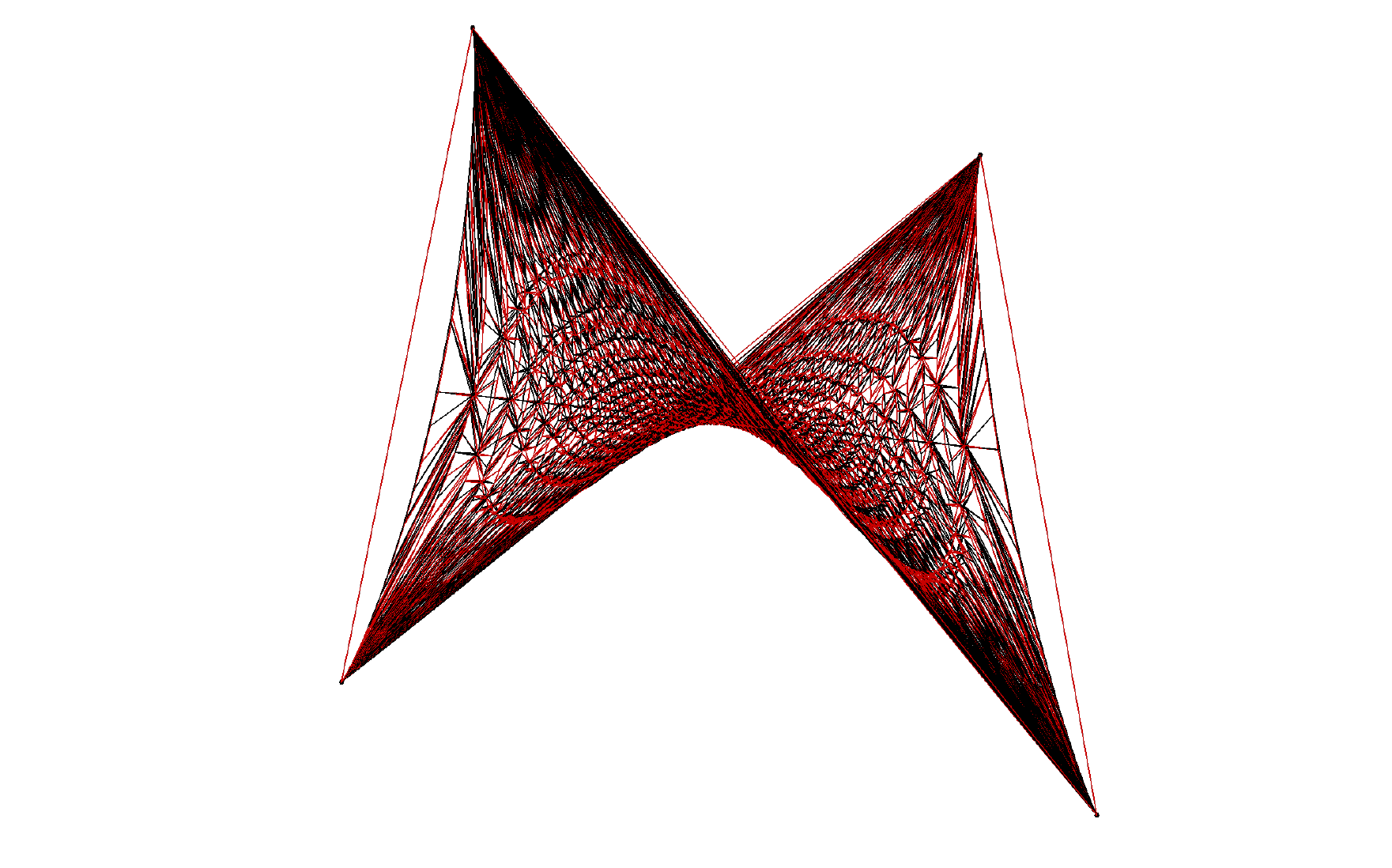}}
\end{minipage}
 \hspace{-3.cm}
\begin{minipage}{0.29\columnwidth}
   {\includegraphics[width=1.4\textwidth]{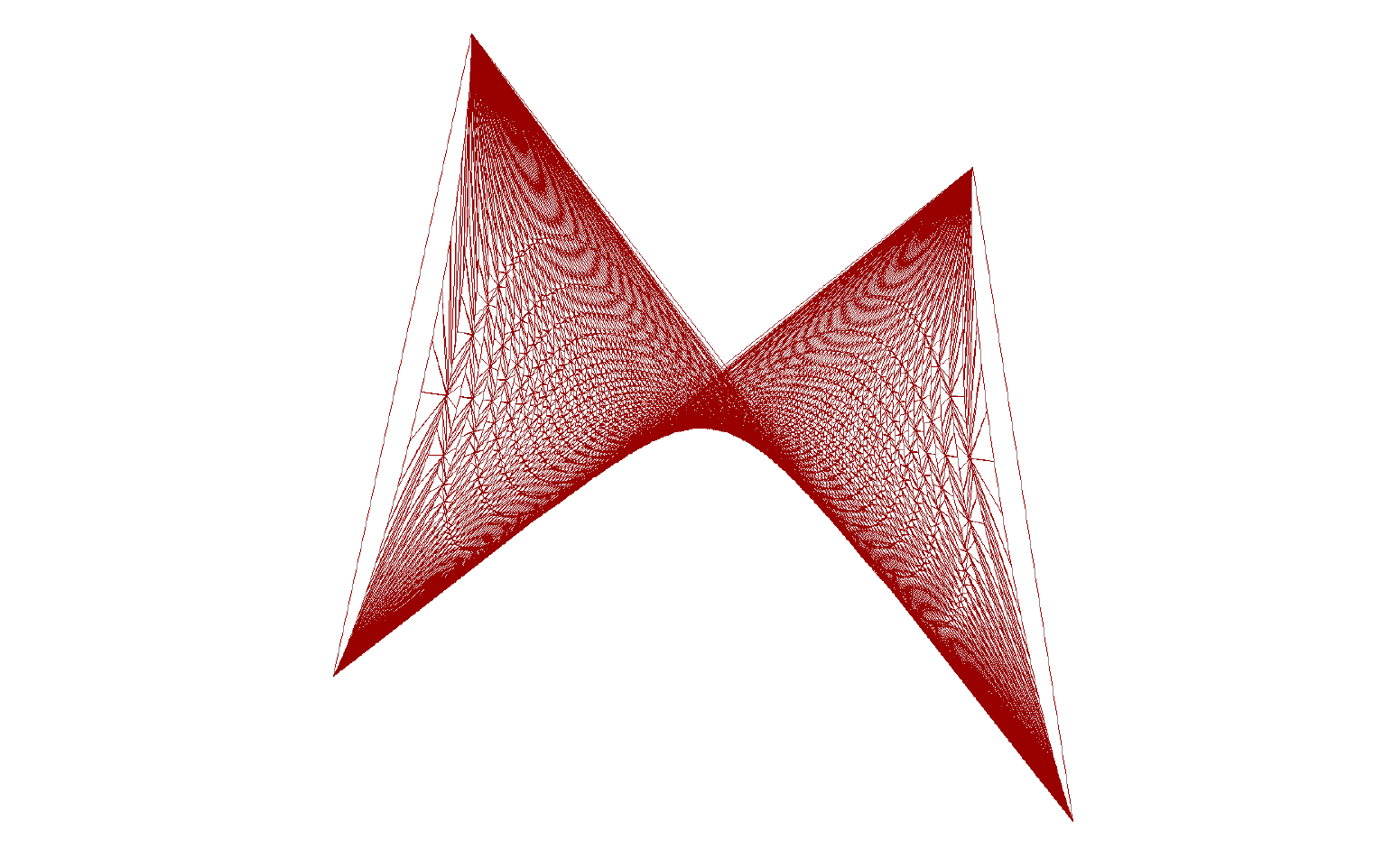}}
   {
   \includegraphics[width=1.4\textwidth]{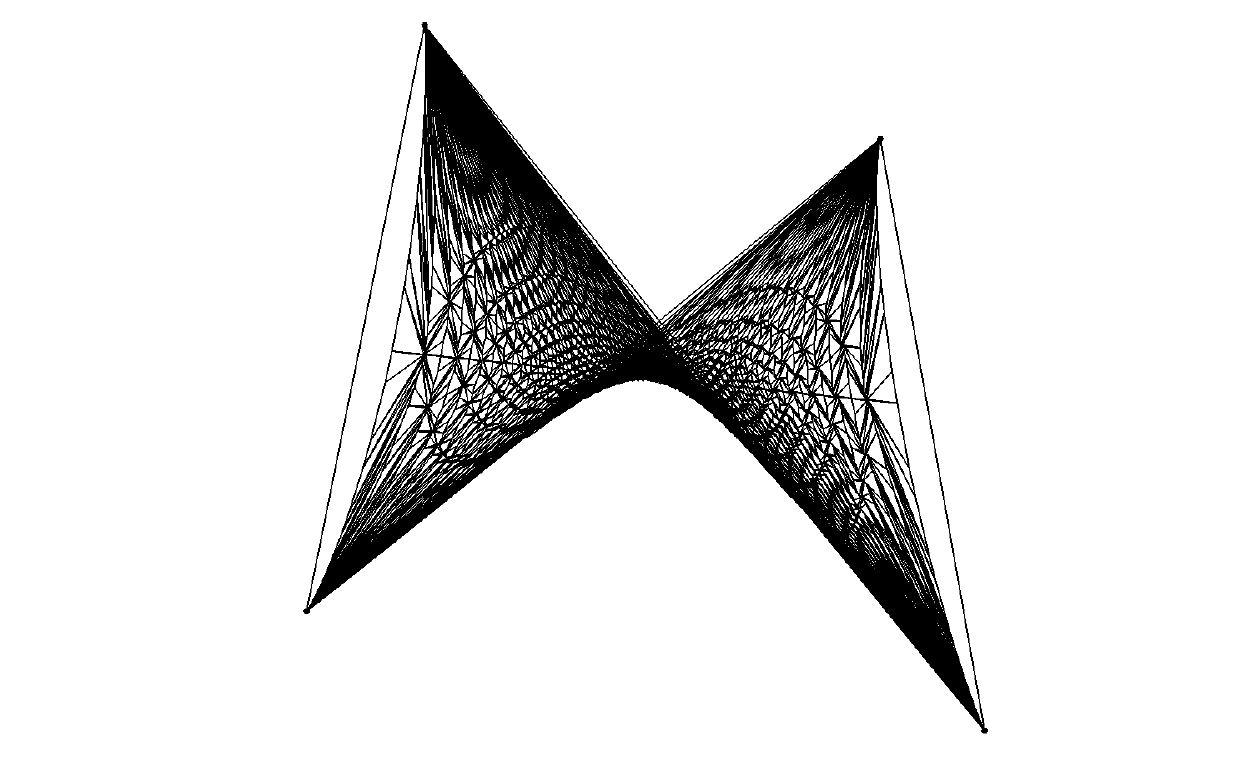}}
 \end{minipage}
  \caption{The graphs of origami maps~$(\T_n, \Or'_n)$,~$(\widetilde\T_n, \widetilde \Or'_n)$ for the~$n = 100$ Aztec diamond, shown in red, and~$n = 50$ tower graph, shown in black. Both graphs approximate the same surface~$S_\diamondsuit$. In the picture, we plot the real part of the origami maps, since they are only real valued in the limit.}
  \label{fig:surface_both}
 \end{center}
\end{figure}

\addtocontents{toc}{\protect\setcounter{tocdepth}{1}}
\subsection*{Organization.} The paper is organized as follows. In Section~\ref{sec:definitions}, we recall the essential facts and objects in the study of the dimer model, and we precisely define perfect t-embeddings. In Section~\ref{subsec:CLR_result} we state the main theorem of~\cite{CLR2}, and in particular we describe in detail its assumptions, whose verification is the content of 
Theorem~\ref{thm:Aztec_assumptions}. Then we define elementary transformations, and state known results about how perfect t-embeddings and origami maps evolve under elementary transformations. In Section~\ref{sec:aztec} we recall the construction of the perfect t-embedding and origami map for the uniformly weighted Aztec diamonds, 
and we prove Theorem~\ref{thm:T_edge_probabilities}. In Section~\ref{sec:tower} we introduce a similar construction of the perfect t-embedding and origami map of a uniformly weighted tower graph, and show how to obtain exact formulas for the tower graph perfect t-embeddings and origami maps, 
and we prove Theorem~\ref{thm:Aztec_Tower}. In Section~\ref{sec:assumptions} we deduce the convergence of the Aztec and tower polygonal surfaces~$(\mathcal{T}_n, \mathcal{O}'_n)$ to the maximal surface, in other words, we prove Proposition~\ref{prop:eqn:to_conv} and Corollary~\ref{cor:convergence_tower}. In the process we also obtain Corollary~\ref{Cor_main_frozen}. In addition, we prove that if we have any sequence of t-embeddings which satisfies a condition that we call the \emph{rigidity condition}, then the conditions~$\LipKd$ and~$\ExpFat$ hold on compact subsets of~$\Omega_{\T_n}$ with an appropriate choice of~$\delta = \delta_n$. Then using these results and Lemma~\ref{lem:bdd}, which states that the perfect t-embeddings of the Aztec diamond satisfy the rigidity condition, we prove Theorem~\ref{thm:Aztec_assumptions}. In Section~\ref{sec:complicated_stuff}, we prove Lemma~\ref{lem:bdd}, thus completing the proof of Theorem~\ref{thm:Aztec_assumptions}.
\addtocontents{toc}{\protect\setcounter{tocdepth}{2}}

\addtocontents{toc}{\protect\setcounter{tocdepth}{1}}
\subsection*{Acknowledgements}  The authors are grateful to Alexei Borodin for his support and interest. 
We also would like to thank Leonid Petrov and Patrik Ferrari for the help with the simulations. MR is grateful to Mikhail Basok for many helpful discussions and sharing his ideas on checking Lipschitz assumption. 
TB~was supported by the Knut and Alice Wallenberg Foundation grant KAW~2019.0523.
MR~was partially supported by the Swiss NSF grant P400P2-194429 and by A.~Borodin's Simons Investigator grant.
\addtocontents{toc}{\protect\setcounter{tocdepth}{2}}

\section{Preliminaries}\label{sec:definitions}
\subsection{Dimer model}
Let~$(\G, \nu)$ be a weighted planar bipartite graph, where~$\nu$ is a positive real-valued function defined on the set of edges of the graph~$\G$. A dimer configuration of a given graph is a subset of edges that covers every vertex exactly once. Let us denote the set of all dimer configurations of~$\G$ by~$\mathcal{M}(\G)$. Then the probability measure on dimer configurations of~$(\G, \nu)$ is given by
 \[\mathbb{P}(m)=\frac{\prod\limits_{e\in m} \nu(e)}{\sum\limits_{m\in\mathcal{M}}\prod\limits_{e\in m} \nu(e)}.\]
 
Let us recall the definition of a Kasteleyn matrix, which is a fundamental object in the Dimer model. To define the Kasteleyn matrix we first need to introduce Kasteleyn signs. Complex signs~$\tau(e)\in\mathbb{C}, |\tau(e)|=1$ on edges of a bipartite graph~$\G$ are called Kasteleyn signs if around a face of~$\G$ of degree~$2k$ the alternating product of signs over the edges of this face is~$(-1)^{k+1}$. A Kasteleyn matrix~$K$ is a weighted, complex-signed adjacency matrix whose
rows index the black vertices and columns index the white vertices, and the signs
are chosen to be Kasteleyn (i.e. for~$b,w\in \G$ one has~$K(b,w)=\tau(bw)\nu(bw)$ if~$b$ and~$w$ are adjacent
and~$K(b,w)=0$ otherwise). It is known, that for any choice of Kasteleyn signs the partition function~$\sum_{m\in\mathcal{M}}\prod_{e\in m} \nu(e)$ is given by~$|\det K|$ and all local statistics of the measure on dimer configurations can be evaluated in terms of the inverse Kasteleyn matrix.

 Recall that two weight functions~$\nu, \nu': E(\G)\to \mathbb{R}_{+}$ are gauge equivalent if there exists a pair of gauge functions~$(F^\bullet_{\operatorname{gauge}}, F^\circ_{\operatorname{gauge}})$ 
such that 
$\nu'(bw)=F^\circ_{\operatorname{gauge}}(w)\nu(bw)F^\bullet_{\operatorname{gauge}}(b).$  Note that gauge equivalent weight functions define the same probability measure on dimer configurations.
Given edge weights~$\nu$ on a bipartite graph, one can associate a face weight~$X_{v^*}$ to each face of~$\G$ by
\[X_{v^*}:=\prod_{s=1}^d\frac{\nu({b_sw_s})}{\nu({b_sw_{s+1}})},\]
where the face $v^*$ has degree $2d$ with vertices denoted by~$w_1, b_1, \ldots , w_d, b_d$ in clockwise order. Note that two edge weight functions are gauge equivalent if and only if they correspond to the same face weights.

\subsection{Perfect t-embeddings and origami maps} In this section we remind of basic properties and definitions of \emph{t-embeddings}, also known as \emph{Coulomb gauges} (except that a Coulomb gauge need not be a \emph{proper} embedding). We refer an interested reader to \cite{KLRR, CLR1, CLR2} for more details.

\begin{figure}
 \begin{center}
 \includegraphics[scale=.59]{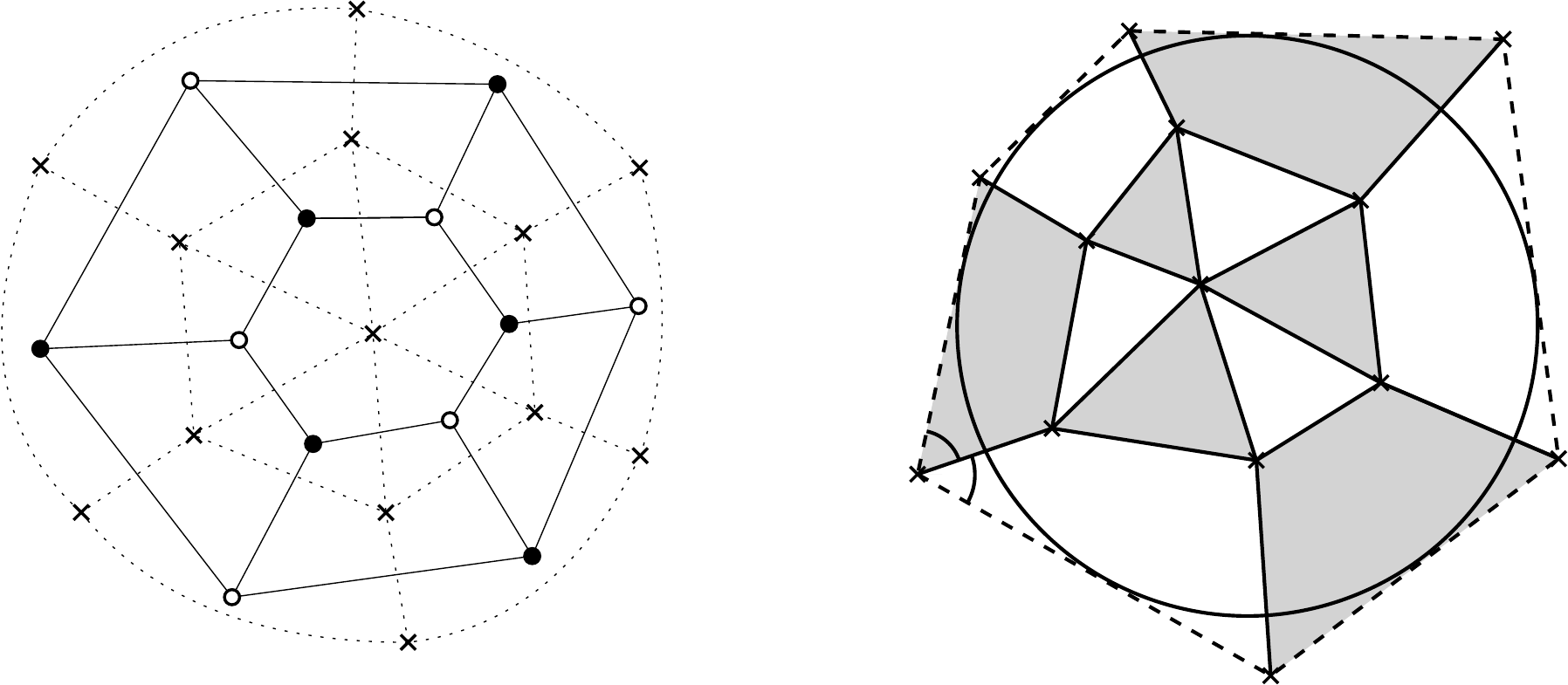}
  \caption{Left: A planar bipartite graph, together with its augmented dual. Right: A perfect t-embedding of the graph, with emphasis on the required boundary conditions.}
  \label{fig:augmented}
 \end{center}
\end{figure}
 

In the finite case a t-embedding is an embedding of the {\emph{augmented dual}}~$\G^*$ of~$\G$. To get an augmented dual~$\G^*$ one should add a vertex~$v_{\operatorname{out}}$ to the graph~$\G$ and connect it with all vertices of the outer face of~$\G$; the augmented dual is then the dual graph of~$\G\cup v_{\operatorname{out}}$. See Figure~\ref{fig:augmented} for an example.

\begin{definition}\label{def:temp} Given weighted planar bipartite graph~$(\G, \nu)$, a t-embedding $\T(\G^*)$ is a proper embedding of an augmented dual graph
such that the following conditions are satisfied:
\begin{itemize}
\item[$\bullet$] the sum of the angles at each inner vertex of $\T(\G^*)$ at the corners corresponding to black faces is equal to $\pi$ (and similarly for white faces),
\item[$\bullet$] the geometric weights (dual edge lengths) $|\T(v^*_1)-\T(v^*_2)|$ are gauge equivalent to $\nu_e$, where $v^*_{1,2}$ are vertices of the dual edge $e^*$.
\end{itemize}
\end{definition}

\begin{remark}The last condition of the above definition can be equivalently formulated in terms of face weights:
In the setup of Definition~\ref{def:temp},
 let $v^*$ be a face of
degree~$2d$ with vertices labeled by~$w_1, b_1, \ldots , w_d, b_d$ (in clockwise order). Let~$X$ be face weights corresponding to the weight function~$\nu$.
Then the second condition of Definition~\ref{def:temp} is equivalent to 
\[X_{v^*}=\prod_{s=1}^d\frac{|\T(v^*)-\T(v^*_{2s-1})|}{|\T(v^*_{2s})-\T(v^*)|} \quad \text{ for all inner } v^*\in\G^*,\]
where by~$v^*_1, v^*_2, \ldots, v^*_{2d}$ we denote dual vertices adjacent to~$v^*$, such that~$v^*v^*_{2s-1}=(w_s b_s^*)$
and~$v^*v^*_{2s}=(b_sw_{s+1}^*)$.   

Moreover, the angle condition 
 implies that 
\[X_{v^*}=(-1)^{d+1}\prod_{s=1}^d\frac{\T(v^*)-\T(v^*_{2s-1})}{\T(v^*_{2s})-\T(v^*)}.\]
\end{remark}

The existence of t-embeddings was shown in~\cite{KLRR} for the following cases:
\begin{itemize}
\item Let~$\G$ be a planar bipartite weighted graph with outer face of
degree~$4$. Fix a convex quadrilateral~$P$. Then there are~$2$ t-embeddings
of the augmented dual of~$\G$ which produce weights that are gauge
equivalent to the original weights and such that the four outer dual
vertices are mapped to the vertices of~$P$.
\item Let~$\G$ be an infinite weighted bipartite graph, periodic in two directions. Then periodic t-embeddings of the dual of~$\G$ producing edge weights that are gauge equivalent to the original ones are in bijection with liquid ergodic Gibbs measures on~$\G$.
\end{itemize}
In general, if the boundary conditions are not specified and there are~$2k$ outer boundary vertices, it is expected to have many (the number is a function only of the graph~$\G$) t-embeddings, though properness of at least one of them is still an open question, i.e. it is possible that the edges of these t-embeddings intersect. The existence of at least one (maybe non-proper) t-embedding in this setup is shown in~\cite[Section 3.4]{KLRR}.

We are interested in t-embeddings with special boundary conditions. The notion of \emph{perfect t-embedding} was introduced in~\cite{CLR2}. Due to~\cite{CLR2} these embeddings are expected  to recover the conformal structure of the dimer model height fluctuations. 

\begin{definition}
A t-embedding is perfect if the outer face of $\T(\G^*)$ is a tangential polygon
to a circle and all the non-boundary edges adjacent to boundary vertices lie on bisectors of the corresponding angles.
\end{definition}

The existence and uniqueness (up to some isomorphisms) of perfect t-embeddings of a given abstract planar  bipartite graph carrying the dimer model is still an open question, we refer an interested reader to a discussion in~\cite[Section 4]{CLR2}.

To each t-embedding $\T(\G^*)$ one can associate the so-called origami map $\Or:\G^*\to\mathbb{C}$.
Let us start with a definition of the origami map in terms of the origami square root function introduced in~\cite[Section 2.2]{CLR1}. Following~\cite{CLR1}, for each $w\sim b$ denote by $d\T(bw^*)$ the complex number $\T(v')-\T(v)$, where $vv'$ is an oriented edge $(bw^*)$ with $w$ on the left.  Note that the angle condition for t-embeddings implies that~$K_{\T}$, defined by~{$K_{\T}(b, w):=d\T(bw^*)$}  if~$b$ and~$w$ are adjacent
and~{$K_{\T}(b,w)=0$} otherwise, is a Kasteleyn matrix. Denoting by $\mathbb{T}$ the unit circle in $\mathbb{C}$, the origami square root function $\eta: B\cup W\to \mathbb{T}$  is defined by 
\begin{equation}\label{eq:origami_sq_roor}
\overline{\eta}_b\overline{\eta}_w=\frac{d\T(bw^*)}{|d\T(bw^*)|},
\end{equation}
see~\cite[Definition 2.4]{CLR1}.  The function $\eta$ itself branches over every vertex, however $\eta^2$ is well-defined. Note also, that $\eta: W \to \mathbb{T}$ is defined up to a global rotation, so we can choose~{$\eta_{w_0}=1.$} 

Denote by~$\Omega_{\T}$ the discrete domain corresponding to t-embedding~$\T(\G^*)$. For each vertex~$v$ of a bipartite graph~$\G$ we denote the corresponding face of the t--embedding by~$\T(v)$. Then~$\Omega_{\T}$ can be seen as~$\cup_{v\in\G}\T(v)$.

\begin{definition}\label{def:O_dO}
The origami map $\Or$ is a primitive of the piece-wise constant differential form~$d\Or$ defined by 
\begin{equation}\label{eq:dO}
d\Or(z) = 
\begin{cases}
\eta_w^2 \,d z \text{ if } z\in \T(w) \\
\overline{\eta}_b^2 \,d\overline{z} \text{ if } z\in \T(b).
\end{cases}
\end{equation}
\end{definition}

Let us also remind of an alternative definition of the origami map. 
\begin{definition}\label{def:O} 
Fix a white root face~$w_0$, and set~$\Or(z)=\T(z)$ for all~$z\in \T(w_0)$. For any other face~$v$ of the $\G^*$ consider a face path~$\gamma$ from~$\T(w_0)$ to~$\T(v)$ and reflect (all points of)~$\T(v)$ along each edge of the t-embedding crossing~$\gamma$ to get~$\Or(v)$.
\end{definition}

Because of the angle condition, this procedure is independent of the path~$\gamma$, so~$\Or$ is well-defined. Informally speaking, to get an origami map~$\Or(\G^*)$ from~$\T(\G^*)$ one can choose a white root
face~$w_0$, set~$\Or(w_0)=\T(w_0)$, and fold the plane along every edge of the t-embedding. Using this interpretation it is easy to see that if~$\T$ is a perfect t-embedding, then~$\Or$ maps the outer face of~$\G^*$ onto a line.

\begin{remark}
The primitive of~$d\Or(z)$ is defined up to a global additive constant. To match Definition~\ref{def:O_dO} with  Definition~\ref{def:O} set~$\eta_{w_0}=1$ and~$\Or(w_0)=\T(w_0)$.
\end{remark}

\subsection{Main result of \cite{CLR2}}
\label{subsec:CLR_result}

\begin{figure}
 \begin{center}
 \includegraphics[scale=.65]{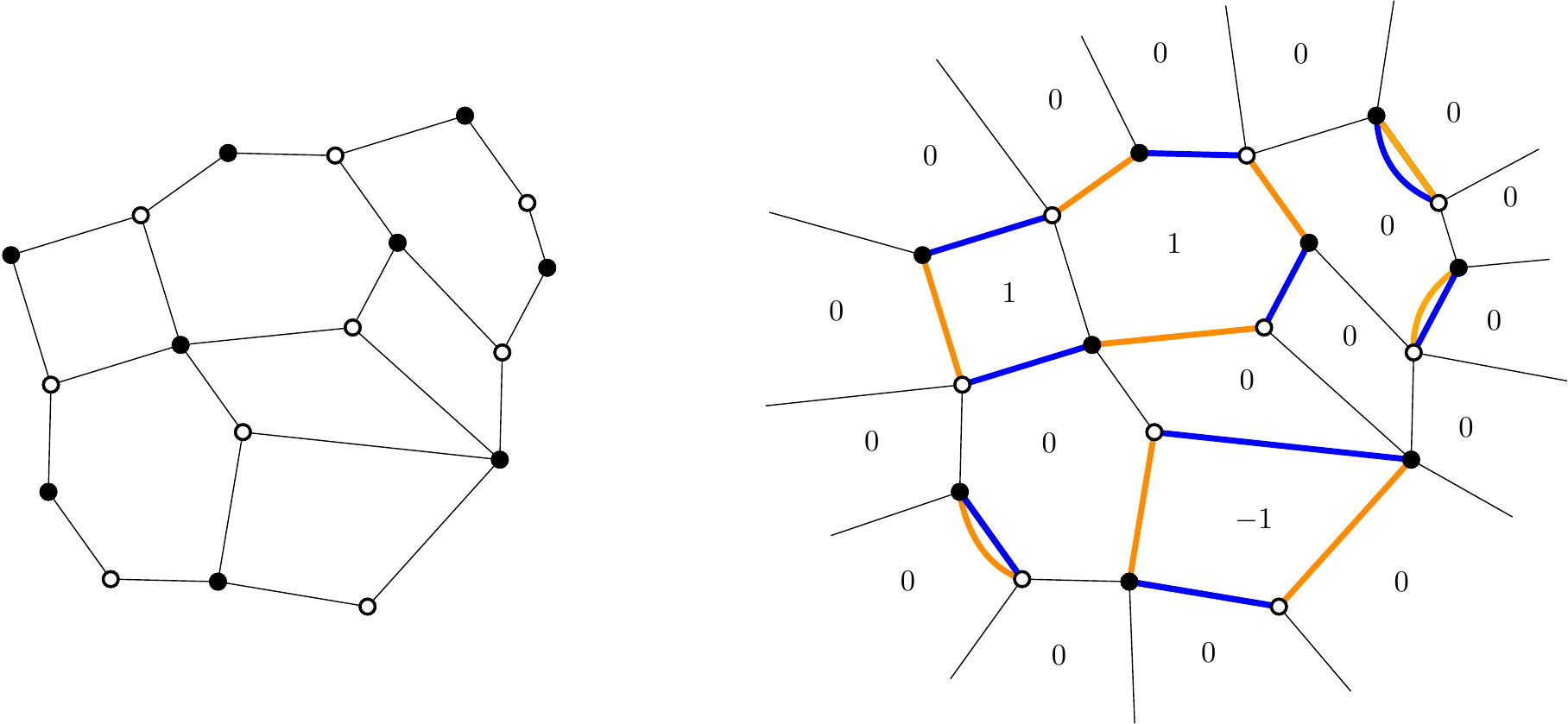}
  \caption{A planar bipartite graph~$\mathcal{G}$, and a height function on~$\mathcal{G}$ together with its corresponding perfect matching~$M$ (orange), and a fixed reference matching~$M_0$ (blue).
  }
  \label{fig:height_fn}
 \end{center}
\end{figure}

Recall that one of our main goals is to verify all assumptions of~\cite[Theorem 1.4]{CLR2} for the sequence of perfect t-embeddings of uniformly weighted Aztec diamonds.


First, we recall the definition of the height function~$h$ corresponding to a random matching~$M$ sampled from the dimer model on~$\mathcal{G}$. The height~$h(v^*)$ is a function on vertices of the augmented dual~$\mathcal{G}^*$. We fix some reference matching~$M_0$ of~$\mathcal{G}$, and fix the boundary values~$h(v^*) =0$ for outer vertices of~$\mathcal{G}^*$. Then, for each edge~$(b w^*) = u^* v^*$, oriented from~$u^*$ to~$v^*$ so that the white vertex is on the left, define~$h(v^*)-h(u^*) = -\mathbf{1}_{bw \in M} + \mathbf{1}_{bw \in M_0}$, see Figure~\ref{fig:height_fn}. Note that~$h$ depends on the reference matching~$M_0$, but~$h - \mathbb{E}[h]$ does not.

Now to fix notation, define the pseudo-Riemannian manifold~$\mathbb{R}^{2,1}$ as~$\mathbb{R}^3$ equipped with the metric~$d x^2 + d y^2 - d \vartheta^2$. Recall that we denote by~$\Omega_{\T}$ the region covered by the union of faces of~$\mathcal{T}$. We are now ready to formulate the main theorem of \cite{CLR2}. 
\begin{theorem}[\cite{CLR2}]\label{thm:CLR2thm}
Suppose there exists a sequence of perfect t-embeddings~$\mathcal{T}_n$ of weighted planar bipartite graphs~$\mathcal{G}_n$ satisfying the following properties:
\begin{enumerate}
\makeatletter
\renewcommand{\labelenumi}{\textbf{\theenumi}}
\renewcommand{\theenumi}{\textbf{\arabic{enumi}}}
\makeatother
\item \label{item:a} The regions~$\Omega_{\T_n}$ approximate (in the Hausdorff sense) a simply connected domain~$\Omega \subset \mathbb{C}$ as~$n \rightarrow \infty$.
\item \label{item:b} We have the convergence of origami maps~$\mathcal{O}_n(z) \rightarrow \vartheta(z)$, uniformly on compact subsets, for some function~$\vartheta : \Omega \rightarrow \mathbb{R}$. Furthermore, the limiting graph~$S_{\Omega}=\{(z, \vartheta(z))\}_{z \in \Omega}$ of~$\mathcal{O}_n$ over~$\mathcal{T}_n$ is a \emph{space-like surface with zero mean curvature in~$\mathbb{R}^{2,1}$.}
\item \label{item:c} There is a sequence of scales~$\delta = \delta_n \rightarrow 0$, such that the sequence~$\T_n$ satisfies Assumptions~\ref{assumption:Lip} and ~\ref{assumption:Exp_fat} on compact subsets (see the next paragraph for precise descriptions of the assumptions).
\end{enumerate}
Let~$\zeta_{m, i} \in \mathbb{D}, m = 1,\dots, k, i = 1,2$ be points of the disk, and let~$v_{m, i}^{(n)}$ be vertices of~$\T_n$ which approximate the points~$z(\zeta_{m, i})$, where~$(z(\zeta), \vartheta(\zeta ))$ is a conformal parameterization of~$S_{\Omega}$. Suppose that for~$i = 1, 2$,  the points~$v_{1, i}^{(n)},\dots,v_{k, i}^{(n)}$ stay uniformly away from each other and from  the boundary~$\partial \Omega$. Then, 
$$\mathbb{E}\left[\prod_{m=1}^k (\bar{h}_n(v_{m, 2}^{(n)}) - \bar{h}_n(v_{m, 1}^{(n)}))\right] \rightarrow \sum_{r_1,\dots,r_k \in\{1,2\}}(-1)^{r_1+\cdots+r_k} G_k(\zeta_{1, r_1}, \dots, \zeta_{k, r_k}),$$
where~$\bar{h}_n(v) = h_n(v)- \mathbb{E}[h_n(v)]$ denotes the mean subtracted height function on~$\T_n$, and~$G_k$ is the~$k$-point correlation function of the Gaussian free field on~$\mathbb{D}$.
\end{theorem}

\begin{remark}
In the setting of the theorem, the surface~$S_{\Omega}$ carries a canonical conformal structure, namely, the conformal structure of the induced metric from the embedding in~$\R^{2,1}$. Thus, a conformal parameterization exists and is unique up to biholomorphic maps of the disc.
\end{remark}

\begin{remark}
The conclusion of the theorem above is slightly weaker than convergence of height fluctuations to a Gaussian free field; see~\cite[Remark 1.5]{CLR2} for a further discussion.
\end{remark}

We continue on to precisely state the two technical assumptions
in item~\ref{item:c} of the theorem above in more detail, as verifying these is a central piece of our main theorem. The first assumption roughly speaking requires that (at the discrete level) the origami map is Lipschitz with constant strictly less than~$1$. Given constants~$\delta>0$ and~$\kappa\in (0,1)$, we say that a t-embedding~$\T$ satisfies the assumption~$\LipKd$ on a compact subset~$K \subset \Omega$ if
 \[|\Or(z')-\Or(z)|\leq \kappa |z'-z|\]
 for all~$z,z'\in K$ such that~$|z'-z|\geq\delta.$

\begin{assumption}\label{assumption:Lip}
Given a sequence~$\delta=\delta_n\to 0$, we say that a sequence of t-embeddings~$\T_n$ satisfies the Lipschitz assumption on compact subsets if for each compact~$\mathcal{K} \subset \Omega$ there exists~${\kappa \in (0, 1)}$ such that for~$n$ large enough~$\mathcal{T}_n$ satisfies~$\LipKd$ on~$\mathcal{K}$.
\end{assumption}

The second assumption is a non-degeneracy condition on the faces of the sequence of t-embeddings, which roughly speaking asserts that for almost every face, the radius of the largest circle which can be inscribed in the face cannot decay exponentially fast as~$n \rightarrow \infty$. We say a face of~$\mathcal{T}$ is~$\rho$-fat if there exists a disk of radius~$\rho$ contained in the face.  
 To formulate the second assumption, we also need to recall the definition of a splitting of a t-embedding introduced in~\cite{CLR1}. A \emph{splitting}~$\mathcal{T}_{\mathrm{spl}}^\circ$ (resp.~$\mathcal{T}_{\mathrm{spl}}^{\tb}$) of~$\mathcal{T}$ is obtained from~$\mathcal{T}$ by adding diagonals to each white (resp. black) face of degree larger than $3$ such that all of these faces are split into triangles.
\begin{assumption}\label{assumption:Exp_fat}
Given a sequence~$\delta=\delta_n\to 0$, we say that a sequence of t-embeddings~$\T_n$ satisfies the assumption~$\ExpFat$ on compact subsets if for each compact~$\mathcal{K} \subset \Omega$, there exists some sequence~$\delta' = \delta'_n \rightarrow 0$ such that
\begin{itemize}
\item There exist splittings~$(\mathcal{T}_{n})_{\mathrm{spl}}^{\tb}$ such that after removing all $\exp(-\delta' \delta^{-1})$-fat white faces and black triangles, the size of any remaining vertex connected component in~$\mathcal{K}$ converges to $0$ as $n \rightarrow \infty$.

\item There exists splittings $(\mathcal{T}_{n})_{\mathrm{spl}}^{\tw}$ satisfying the same condition.
\end{itemize}

\end{assumption}

\subsection{Perfect t-embeddings and elementary transformations of bipartite graphs}
There are~$3$ types of elementary transformations of the planar bipartite graph that, in a certain sense, preserve partition and correlation functions of the dimer model, see Figure~\ref{fig:elem}.
One of the crucial (on the discrete level) properties of t-embeddings is related to elementary transformations of bipartite graphs. 
To state this property we need to remind a definition of a \emph{central move} introduced in~\cite{KLRR}. Let  $u, u_1, u_2, u_3, u_4 \in \mathbb{C}$ be distinct points. Note that the equation 
\begin{equation}\label{eq:cental_move}
\frac{(u_1-z)(u_3-z)}{(u_2-z)(u_4-z)}=\frac{(u_1-u)(u_3-u)}{(u_2-u)(u_4-u)}
\end{equation}
has two roots $z = u$ and $z = \tilde{u}$.
The map $u\mapsto\tilde{u}$ is called a central move with respect to $u_1, \ldots, u_4$. Let us also recall one of the geometrical properties of a central move given in \cite[Remark 6]{KLRR}.

\begin{figure}
 \begin{center}
\includegraphics[width=0.23\textwidth]{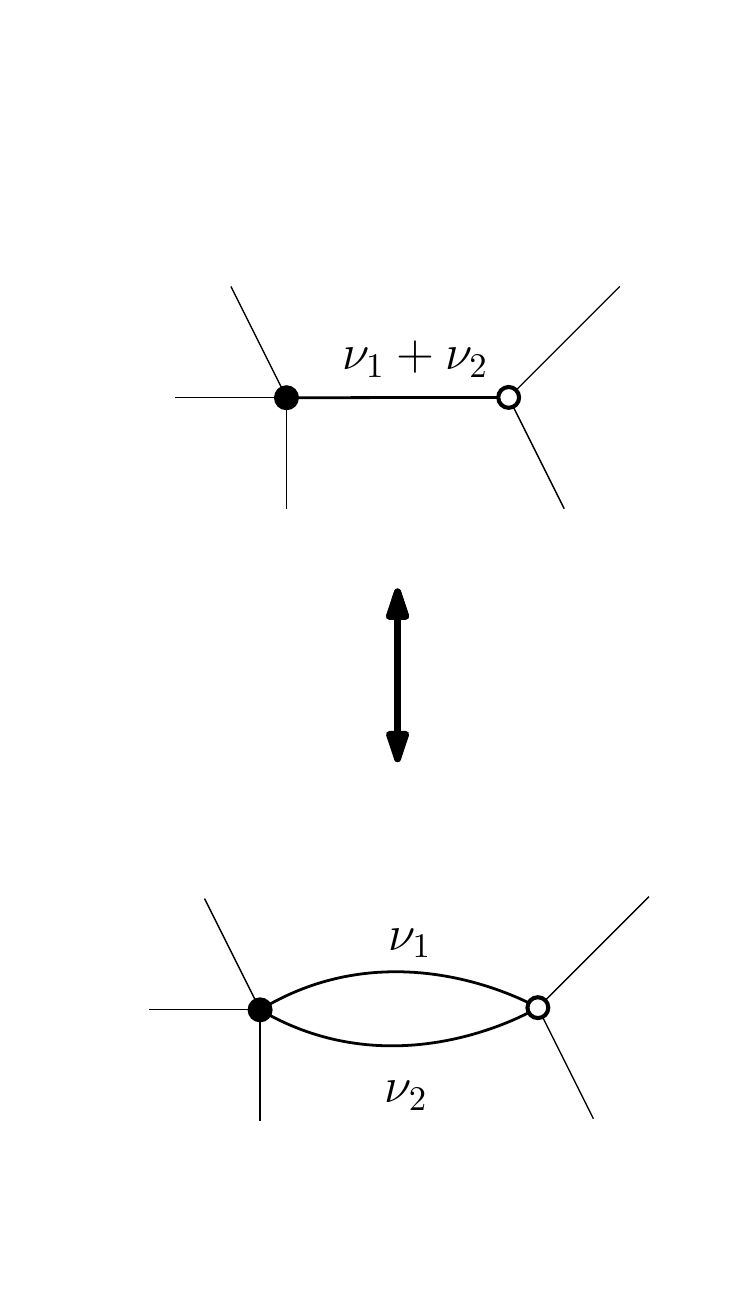}
$\quad\quad\quad$
\includegraphics[width=0.23\textwidth]{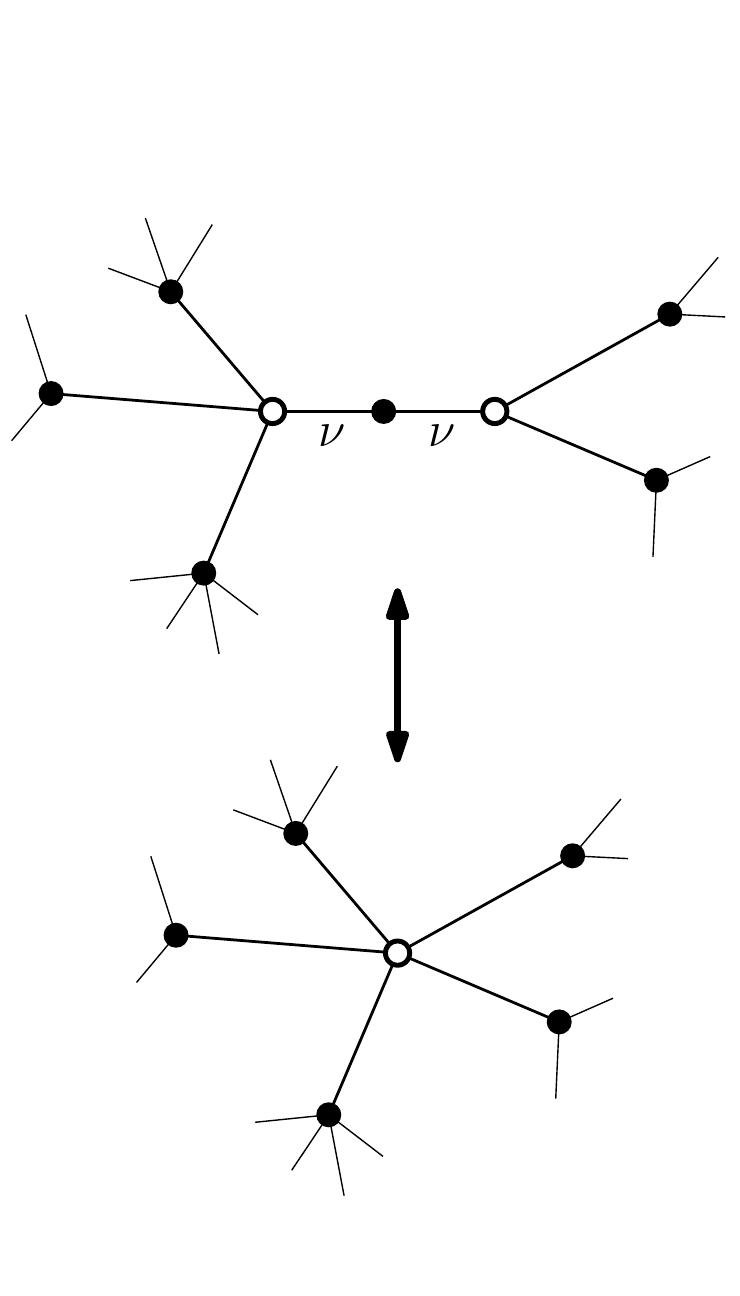}
$\quad\quad\quad$
\includegraphics[width=0.23\textwidth]{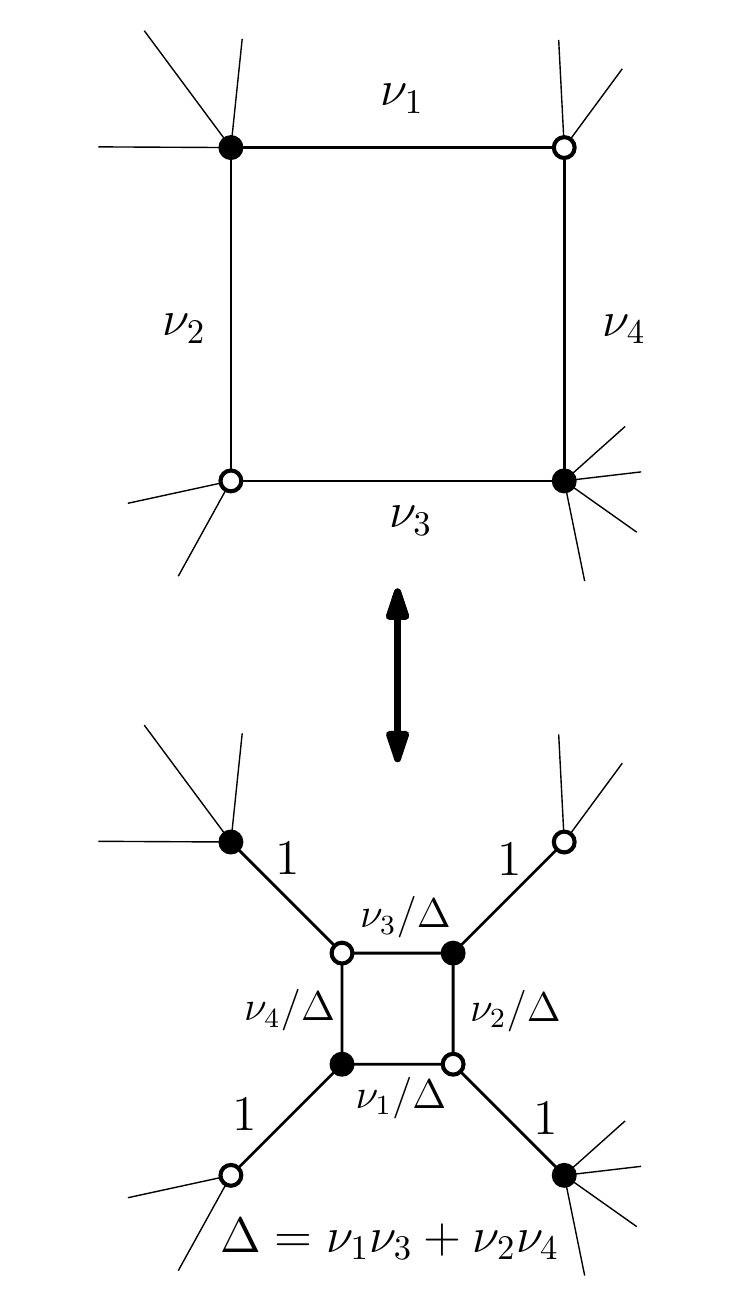}
 \end{center}
\caption{Elementary transformations of weighted bipartite graphs:
(1)~a single edge with weight~$\nu_1 + \nu_2$ can be replaced by parallel edges with weights~$\nu_1, \nu_2$; double edges  with weights~$\nu_1, \nu_2$ can be merge to a single edge with weight~$\nu_1 + \nu_2$;
(2)~contracting a degree~$2$ vertex whose edges have equal weights; splitting a vertex of degree $d_1+d_2$ to two vertices of degrees $d_1+1$ and $d_2+1$ and adding a degree two vertex between them;
(3)~spider move, with the weight change as shown.
}\label{fig:elem}
\end{figure}

\begin{figure}
 \begin{center}
  \includegraphics[width=0.25\textwidth]{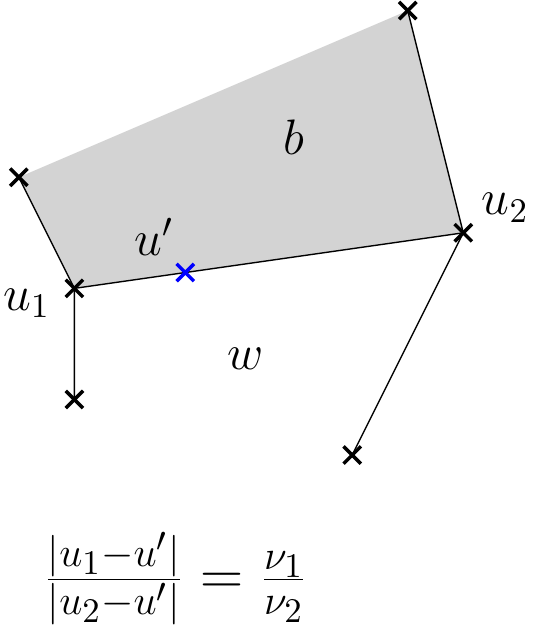}
$\,$
  \vline
  $\quad$
  \includegraphics[width=0.25\textwidth]{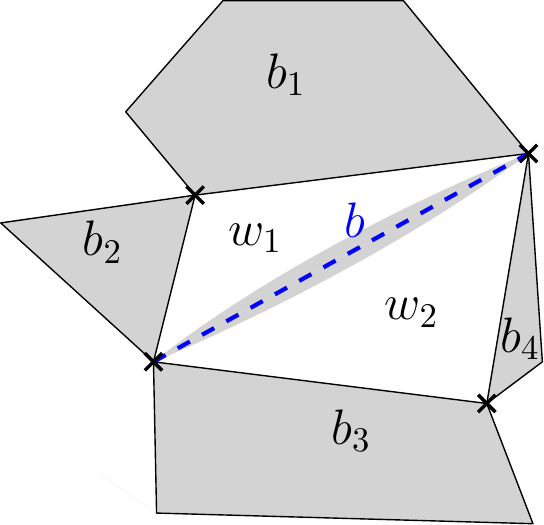}
    $\quad$
  \vline
  $\quad$
    \includegraphics[width=0.35\textwidth]{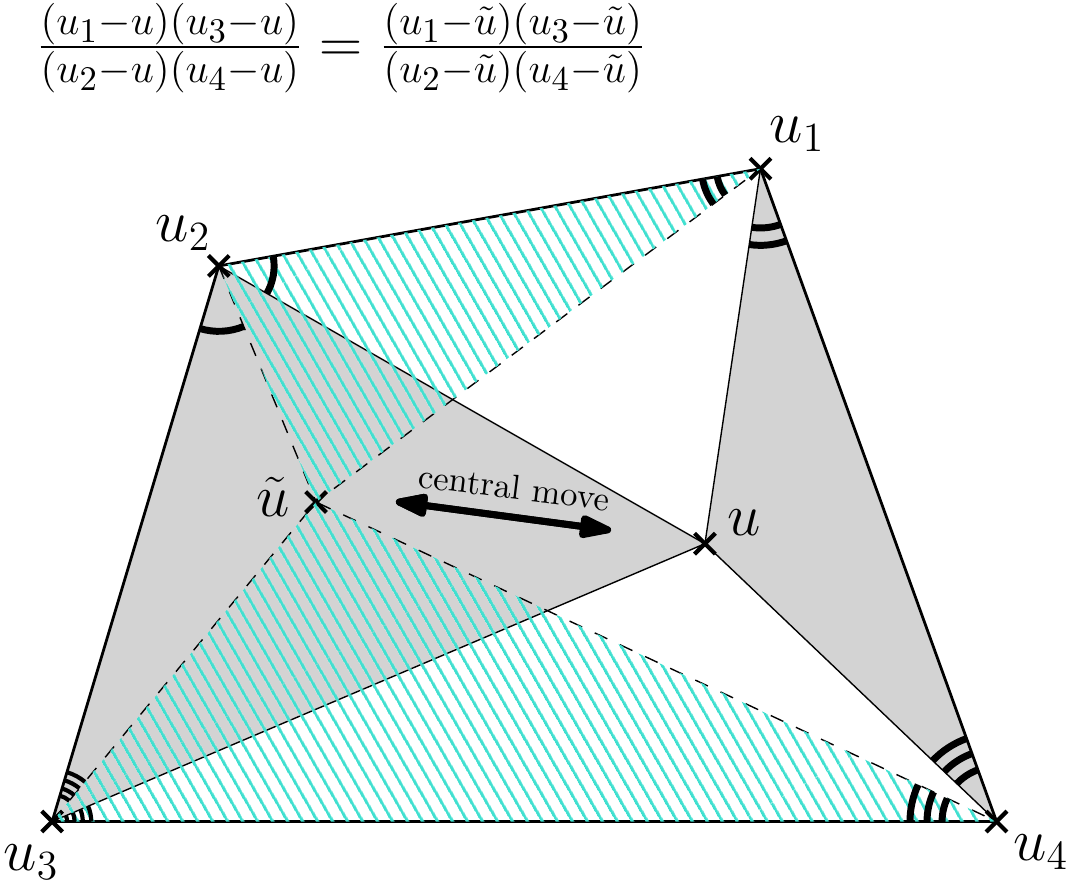}
  \caption{Transformations of t-embeddings: (1) adding / removing vertex of degree $2$; (2) adding / removing diagonal of a face of a t-embedding; (3) central move (points obtained by a central move are isogonal conjugate). 
  }\label{fig:elem_t_emb}
 \end{center}
\end{figure}

\begin{remark}[\cite{KLRR}]\label{isog}
The isogonal conjugate of a point $u$ with respect to a quadrilateral $u_1u_2u_3u_4$ is given by reflecting the lines $uu_i$ about the angle bisectors of $u_i$. If these four reflected lines happen to intersect at one point, then this point  is called the isogonal conjugate of $u$. The points obtained by a central move are isogonal conjugate, see Figure~\ref{fig:elem_t_emb}.
\end{remark}

\begin{proposition}[\cite{KLRR}]\label{prop:elem}
T-embeddings of~$\G^*$ are preserved under elementary transformations of~$\G$. More precisely, 
\begin{enumerate}
\item replacement of a single edge with weight~$\nu_1 + \nu_2$ by parallel edges with weights~$\nu_1, \nu_2$ corresponds to adding a point dividing corresponding edges of the t-embedding in proportion~$[\nu_1: \nu_2]$; merging double edges corresponds to removing a degree~$2$ vertex of the t-embedding, which due to the angle condition has to lie on a line with its two adjacent vertices;
\item contracting a degree~$2$ vertex corresponds to removing a degree two face of a t-embedding which can be seen as a diagonal of a face of the t-embedding; splitting a vertex of degree $d_1+d_2$ to two vertices of degrees $d_1+1$ and $d_2+1$ and adding a degree two vertex between them corresponds to adding a diagonal to the corresponding face of the t-embedding with respect to the structure of the splitted graph;
\item a spider move corresponds to a central move of points of t-embedding.
\end{enumerate}
\end{proposition}

\begin{remark} Note that perfect t-embeddings are preserved under elementary transformations that preserve boundary vertices in all intermediate graphs. 
\end{remark}

\begin{corollary}\label{cor:or_elem}
In the same sense as in Proposition~\ref{prop:elem}, origami maps are preserved under elementary transformations of $\G$.
\end{corollary}

\begin{proof} 
Adding a point dividing an edge of the t-embedding in proportion~$[\nu_1: \nu_2]$ corresponds to adding a point dividing corresponding edge of the origami map in the same proportion. 

Diagonal of a face of the t-embedding is actually a face itself which has degree~$2$ and folding along both of them doesn't change the rest of the origami map.  Therefore adding / removing a diagonal of a face of the t-embedding implies adding / removing a diagonal of a corresponding face of the origami map.

Let $u_1, u_2, u_3, u_4$ and~$u$ be vertices of the t-embedding as shown on Figure~\ref{fig:or_spider}. Recall that a spider move of the initial bipartite graph corresponds to a central move for corresponding vertices of a t-embedding, see~\cite[Equation (6)]{KLRR}.

Let us first show that the positions of points $u_i$ under the origami map stay the same after applying a central move to the perfect t-embedding. To see this, recall that the points $u$ and $\tilde{u}$ are isogonal conjugate (see Remark~\ref{isog}), therefore the angle between the lines~$\Or(u_{i-1})\Or(u_{i})$ and~$\Or(u_{i})\Or(u_{i+1})$ is equal to the angle between the lines~$uu_i$ and~$u_i\tilde{u}$.
Since $|\Or(u_i)\Or(u_{i+1})|=|u_iu_{i+1}|$, we obtain that  the positions of points $u_i$ under the origami map stay the same after applying a central move up to a global translation, and this global translation is fixed to be the same by the  boundary conditions, for example by choosing~$\eta_{w_0}=1$ for a fixed face~$w_0$.

Finally, let~$\tilde{u}$ be the vertex obtained from~$u$ by a central move, i.e.
\[\frac{(u_1-u)(u_3-u)}{(u_2-u)(u_4-u)}=\frac{(u_1-\tilde{u})(u_3-\tilde{u})}{(u_2-\tilde{u})(u_4-\tilde{u})},\]
we want to show that
\[\frac{(\Or(u_1)-\Or(u))(\Or(u_3)-\Or(u))}{(\Or(u_2)-\Or(u))(\Or(u_4)-\Or(u))}=\frac{(\Or(u_1)-\Or(\tilde{u}))(\Or(u_3)-\Or(\tilde{u}))}{(\Or(u_2)-\Or(\tilde{u}))(\Or(u_4)-\Or(\tilde{u}))},\]
which is a central move for the vertices of the origami map.
Indeed, note that  $|\T(v)\T(v')|=|\Or(v)\Or(v')|$ for any adjacent vertices $v, v'\in\G^*$ and 
the sum of the angles at each inner vertex of $\Or(\G^*)$ at the corners corresponding to black faces is equal to~$\pi$ (since these angles are equal to the angles of the corresponding t-embedding). 
\end{proof}

\begin{figure}
 \begin{center}
 \includegraphics[width=0.6\textwidth]{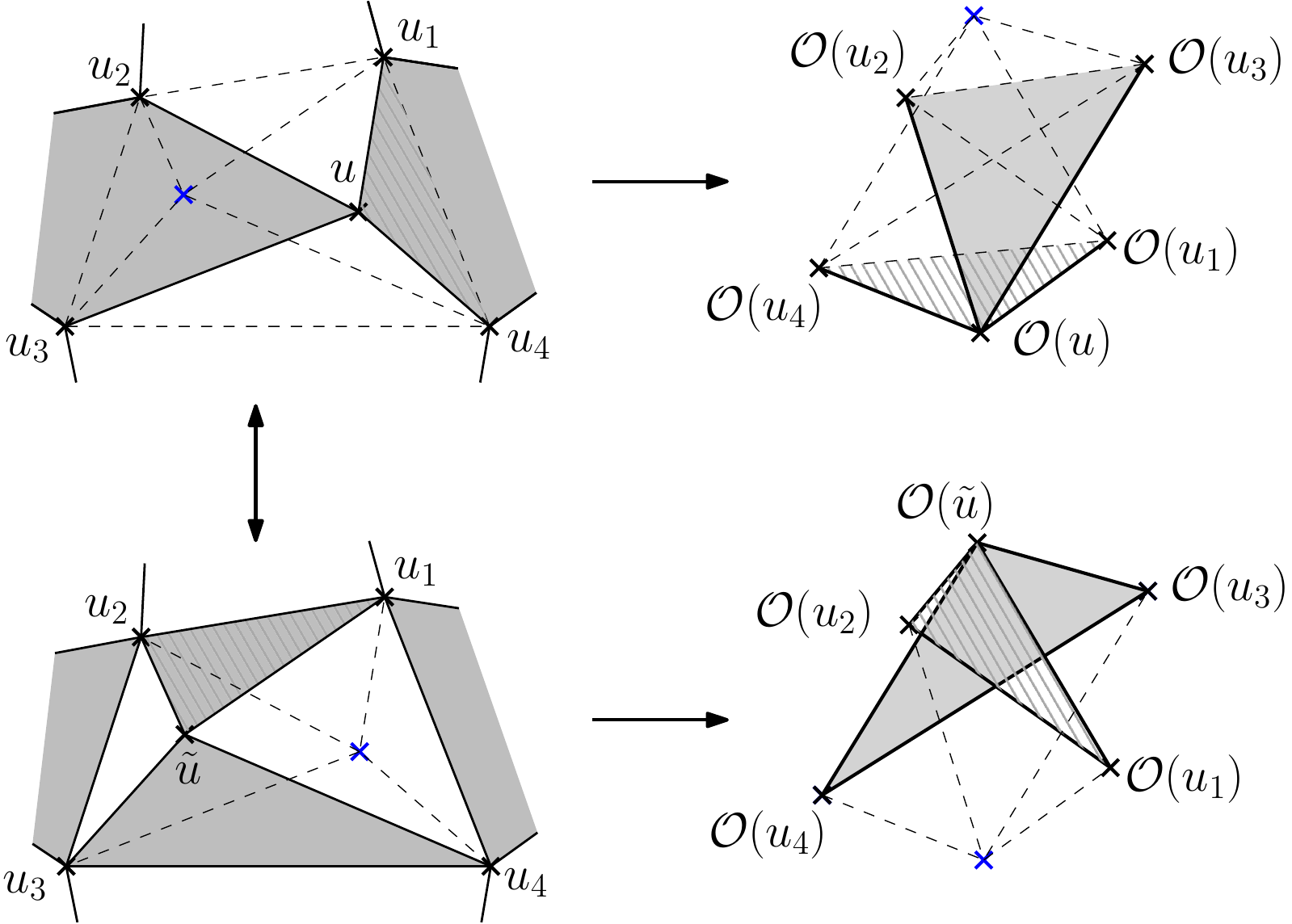}
  \caption{Central move of a vertex of a t-embedding and corresponding origami maps. 
  }\label{fig:or_spider}
 \end{center}
\end{figure}


\section{Perfect t-embeddings of Aztec diamonds}\label{sec:aztec}
We first recall the construction of a perfect t-embedding of Aztec diamond with homogeneous dimer weights given in~\cite{Ch-R}. In Section~\ref{sec:aztec_prob} we show a connection of perfect t-embeddings of Aztec diamond with edge probabilities. We use this connection in Section~\ref{sec:assumptions} to show that perfect t-embeddings of the uniform dimer model on Aztec diamond satisfy all assumptions of Theorem~\ref{thm:CLR2thm}.

\subsection{The Aztec diamond and its reduction.}\label{sec:aztec_def}
Consider the square grid~$(\mathbb{Z}+\frac12)^2$. The faces of such a grid can be naturally indexed by pairs~$(j,k)\in\mathbb{Z}^2.$ Let~$n$ be a positive integer. The \emph{Aztec diamond}~$A_{n+1}$  of size~$n+1$ is a subset of faces~$(j,k)$ of the square grid~$(\mathbb{Z}+\frac12)^2$ such that~$|j|+|k| \leq n$. 
 We also choose the bipartite coloring of vertices of the Aztec diamond such that all North-Eastern corners are black, see Figure~\ref{fig_period}.

Following~\cite{Ch-R} let us also define \emph{reduced Aztec diamond}~$A'_{n+1}$  of size~$n+1$. To obtain the reduced Aztec diamond $A'_{n+1}$ from $A_{n+1}$ one should make the following sequence of moves:
\begin{itemize}
\item[$\bullet$] contract black vertices $(j\pm\frac12,k\pm\frac12)$ of $A_{n+1}$ with $j+k=\pm n$;
\item[$\bullet$] contract white vertices $(j\pm\frac12,k\mp\frac12)$ of $A_{n+1}$ with $\begin{cases} j-k=\pm n\\  
\{|j|, |k|\} \neq \{0,n\} \end{cases}$;  
 \item[$\bullet$] merge pairwise all the $4n$ obtained pairs of parallel edges.
\end{itemize}

Note that inner vertices of the augmented dual $(A'_{n+1})^*$ are in natural correspondence with inner faces of $A_n$ and therefore can be indexed by $(j,k)\in\mathbb{Z}^2$ with $|j|+|k| < n$. The boundary vertex of augmented dual $(A'_{n+1})^*$ adjacent to the dual vertex indexed by~$(n-1,0)$ we index by~$(n,0)$. Similarly we index the other three boundary vertices of the augmented dual~$(A'_{n+1})^*$ by~$(0,n)$, $(-n,0)$ and $(0,-n)$.


\subsection{Recurrence relation for t-embeddings}\label{sec:aztec_rec} In this section we remind the reader of a construction of a perfect t-embedding of the reduced Aztec diamond with homogeneous dimer weights introduced in~\cite{Ch-R}. For the rest of the paper we focus on uniform measure on dimer configurations, in other words we fix all face weights to be~$1$.

The Aztec diamond of size $n+1$ can be obtained from the Aztec diamond of size $n$ using a sequence of elementary transformations. The same holds for the reduced Aztec diamonds, see Figure~\ref{fig:shuffling}. This fact together with Proposition~\ref{prop:elem} was used in~\cite[Proposition 2.4]{Ch-R} to construct perfect t-embeddings~$\T((A'_{n+1})^*)$. Denote by~$\T_{n}(j,k)$ the position of the perfect t-embedding of the vertex of~$(A'_{n+1})^*$ indexed by~$(j,k)$. The following proposition gives a constructive way to define a sequence of perfect t-embeddings~$\T_n$ of reduced Aztec diamonds. 

\begin{figure}
 \begin{center}
\includegraphics[width=0.9\textwidth]{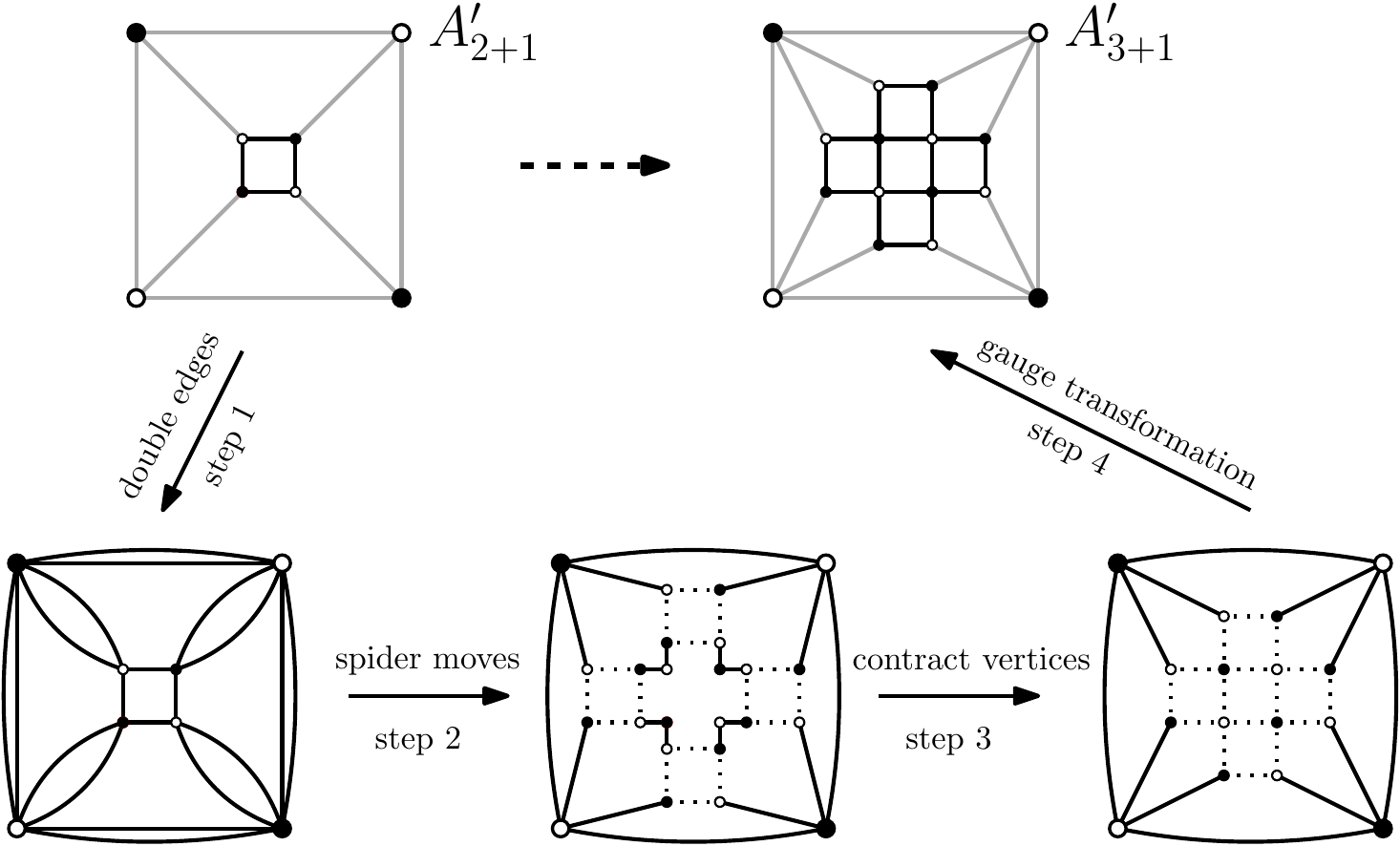}
  \caption{ The steps to get~$A'_{3+1}$ from~$A'_{2+1}$. Weights of black dotted edges are~$\tfrac{1}{2}$, of black filled ones are~$1$ and of grey filled edges are~$2$. The first step is to split all edges adjacent to boundary vertices (note that all these edges have weight~$2$) to double edges with weight~$1$ each. The second step is to apply a spider move at each face indexed by~$(j, k)$ with~$j+k+n$ odd. Then one contract all obtained by spider moves vertices of degree~$2$ as the third step. And finally, the last forth step is to apply a gauge transformation at each white vertex: multiply all edge weights by~$2$.
  }\label{fig:shuffling}
 \end{center}
\end{figure}

\begin{proposition}[\cite{Ch-R}]\label{prop:t_rec_aztec} 
The perfect t-embedding~$\T_{n+1}(j,k)$ can be obtained from~$\T_{n}(j,k)$ using the following update rules
\begin{enumerate}
\item $\T_{n+1}(0,\pm(n+1)) =\pm \i$ 
 and $\T_{n+1}(\pm(n+1),0) = \pm 1$.
\item For $\{|j|, |k|\}=  \{0,n\}$
\begin{equation*}
\T_{n+1}(\pm n,0)=\frac{1}{2}\Big(\T_n(\pm n,0)+\T_n(\pm(n-1),0)\Big),
\end{equation*}
\begin{equation*}
\T_{n+1}(0,\pm n)=\frac{1}{2}\Big(\T_n(0,\pm n)+\T_n(0,\pm(n-1))\Big).
\end{equation*}
\item
For $1\leq j\leq n-1$, $|j|+|k|=n$, $\{|j|, |k|\} \neq \{0,n\}$
\[
\T_{n+1}(j,\pm(n-j))
=\frac{1}{2}\Big(\T_n(j-1,\pm(n-j))+\T_n(j,\pm(n-j-1))\Big).
\]
For $-(n-1)\leq j\leq -1$,  $|j|+|k|=n$
\[
\T_{n+1}(j,\pm(n+j))
=\frac{1}{2}\Big(\T_n(j,\pm(n+j-1))+\T_n(j+1,\pm(n+j))\Big).
\]
\item For $|j|+|k|<n$ and $j+k+n$ even, $\T_{n+1}(j,k)=\T_n(j,k)$.
\item For $|j|+|k|<n$ and $j+k+n$ odd,
\begin{multline*}
\T_{n+1}(j,k)+\T_n(j,k)=\frac{1}{2}\Big(\T_{n+1}(j-1,k)+\T_{n+1}(j+1,k) 
+\T_{n+1}(j,k+1)+\T_{n+1}(j,k-1)\Big).
\end{multline*}
\end{enumerate}
\end{proposition}

For the reader’s convenience, let us briefly outline the proof of this statement.

\begin{proof}
The steps to get reduced Aztec diamond of size $n+1$ from the reduced Aztec diamond of size $n$ are shown on Figure~\ref{fig:shuffling}.
Due to Proposition~\ref{prop:elem}, the first step corresponds to adding mid-points of the corresponding edges of~$\T_n$, i.e. these are update rules~$(2)$ and~$(3)$. 
Update rule~$(4)$ reflects the fact that, for~$j + k + n$ even, the inner faces~$(j, k)$ of~$\T_n$ are not destroyed by spider moves and hence the positions of the corresponding dual vertices in the t-embedding~$\T_{n+1}$ remain the same as in~$\T_n$. The update~$(5)$ corresponds to spider moves at faces for which $j + k + n$ is odd. Here we use the fact that the face weight at the face~$(j,k)$ of~$A'_{n+1}$ is equal to~$1$, and the explicit description~\cite[Equation (6)]{KLRR} of the spider move (the \emph{central move}) in terms of the t-embedding. Finally, the last two steps 
do not affect positions of the dual vertices in the t-embedding.
\end{proof}

Let~$\Or_{n}(j,k)$ be the origami map of~$(A'_{n+1})^*$, where the root white face is the one adjacent to~$\T_n(n,0)$ and~$\T_n(0,n)$. Then the following holds.

\begin{proposition}[\cite{Ch-R}]\label{prop:o_rec}
For $n \geq 1$ the origami map $\Or_{n+1}$ can be constructed from $\Or_n$ using the same update rules (2)--(5) as in Proposition~\ref{prop:t_rec_aztec}, with boundary conditions 
\begin{equation}\label{eq:or_bdry}
 \Or_n(n,0) = \Or_n(-n,0) = 1\quad\text{ and } \quad\Or_n(0,n)=\Or_n(0,-n) = \i.
 \end{equation}
\end{proposition}
\begin{proof}
This proposition follows from Proposition~\ref{prop:t_rec_aztec} combined with Corollary~\ref{cor:or_elem}.
\end{proof}

As it is shown in~\cite{Ch-R} the update rules (2), (3) and (5) from Proposition~\ref{prop:t_rec_aztec}
 can be rewritten in a more universal way. 
Let $\Lambda = \{(j,k,n)\in \ZZ^2\times \ZZ;j+k+n \text{ odd}\}.$
Given~$(b_0,b_E,b_N,b_W,b_S) \in \CC^5 $ 
the following conditions define the function $f:\Lambda \to \CC$ uniquely. 
\begin{enumerate}
\item For all $(j,k)\in\mathbb{Z}^2$
\begin{equation} \label{bc}
f(j,k,0)=0 \quad\text{ and }\quad f(j,k,-1)=0;
\end{equation}
\item If $j+k+n$ is even
\begin{equation} \label{eq:rec-unif}
\begin{split}
f&(j,k,n+1)+f(j,k,n-1)\\
 &- \frac12\big(f(j-1,k,n)+f(j+1,k,n)+f(j,k+1,n)+f(j,k-1,n)\big)\\
&\quad=\delta_{j,0}\delta_{k,0}\delta_{n,0}\cdot b_0 +
\frac12\big(\delta_{j,n}\delta_{k,0}\cdot b_E+
\delta_{j,-n}\delta_{k,0}\cdot b_W+
\delta_{j,0}\delta_{k,n}\cdot b_N+
\delta_{j,0}\delta_{k,-n}\cdot b_S\big).
\end{split}
\end{equation}\end{enumerate}

Note that $f(j, k, n) = 0$ if $n\leq 0$  or $|j| + |k|\geq n$.
\old{$\begin{cases} |j| + |k|\geq n\\
\{|j|, |k|\} \neq \{0,n\}
\end{cases}$. \textcolor{purple}{[I think that the condition $\{|j|, |k|\} \neq \{0,n\}$ for $f = 0$ should be removed.]}}

\begin{remark}\label{rmk:aztec_f}
Note that $\T_n(j,k)$ can be seen as the solution~$f(j,k,n)$ of the above system with the boundary conditions~$(0, 1,\i , -1, -\i )$.
To see this we use that~$\T_{n+1}(j,k)=\T_n(j,k)$ for~$j+k+n$ even. Similarly, the origami map corresponds to the boundary conditions~$(0, 1,\i , 1, \i )$. 
\end{remark}

Let $f_0$ be the fundamental solution to~\eqref{bc}--\eqref{eq:rec-unif}, i.e. the solution with the boundary condition \[(b_0,b_E,b_N,b_W,b_S)=(1,0,0,0,0).\]  
Let $f_E$, $f_N$, $f_W$ and $f_S$  be the solutions with boundary condition~$(0,1,0,0,0)$
and so on. Then the solution to the equation with boundary conditions~$(b_0,b_E,b_N,b_W,b_S)$
is given by
\begin{equation}
f=
b_0f_0+
b_Ef_E+b_Nf_N+b_Wf_W+b_Sf_S.
\end{equation} 
In particular, for $(j,k,n)\in \Lambda$
\begin{equation}\label{eq:T_f_e}
\T_n(j,k)=f_E(j,k,n)+\i f_N(j,k,n)-f_W(j,k,n)-\i  f_S(j,k,n)
\end{equation}
and
\begin{equation}\label{eq:O_f_e}
\Or_n(j,k)=f_E(j,k,n)+\i f_N(j,k,n)+f_W(j,k,n)+\i  f_S(j,k,n).
\end{equation}

As stated in~\cite{Ch-R}, we have the following lemma. 

\begin{lemma}[\cite{Ch-R}]\label{lem:f_e_unif} The function $f_E$ can be expressed via the fundamental solution~$f_0$ in the following way
\begin{equation}\label{eq:f_E_unif}
f_E(j,k,n)=\tfrac12\sum_{s=0}^{n+1} f_{0}(j-s,k,n-s).
\end{equation}
Similarly for $f_W$, $f_N$ and $f_S$ one has
\[f_W(j,k,n)=\tfrac12\sum_{s=0}^{n+1} f_{0}(j+s,k,n-s),\]
\[f_N(j,k,n)=\tfrac12\sum_{s=0}^{n+1} f_{0}(j,k-s,n-s) \quad\text{and}\quad 
f_S(j,k,n)=\tfrac12\sum_{s=0}^{n+1} f_{0}(j,k+s,n-s).
\]
\end{lemma}
\begin{proof} 
Note that the RHS of~\eqref{eq:f_E_unif} satisfies~\eqref{eq:rec-unif} with $(b_0,b_E,b_N,b_W,b_S)=(0,1,0,0,0)$ and vanishes at $n=0$ and $n=-1$. This conditions identify the function $f_E$ uniquely. 
\end{proof}

\subsection{Perfect t-embeddings and edge probabilities.}\label{sec:aztec_prob}
In this section we relate the values of the perfect t-embedding~$\T_n(j,k)$ of a reduced Aztec diamond $A'_{n+1}$ to certain edge inclusion probabilities on Aztec diamond of size~$n$, which is a step towards deriving exact formulas for~$\T_n(j,k)$. Recall that inner faces of~$A'_{n+1}$ are in natural correspondence with inner faces of~$A_n$.

Recall that the Aztec diamond~$A_{n}$ of size $n$ is made up of faces $(j,k)$ in the square grid $(\mathbb{Z}+\frac12)^2$ such that $|j|+|k| < n$. For~$(j, k, n) \in \Lambda$ (which means that~$j + k + n$ is odd), define~$p_E(j, k, n)$ as the probability that the east edge of the face~$(j, k)$ is present in a uniformly random dimer cover of Aztec diamond~$A_n$. 
We define~$p_E(j, k, n) = 0$ if~$n \leq 0$ or if the edge being referred to is not an edge of~$A_n$. Define~$p_W,p_N,p_S$ analogously. We again stress that $\mathcal{T}_n$ is a function on inner faces of $A'_{n+1}$, whereas $p_E(\cdot, \cdot, n)$ (and similarly for $W,N,S$) is a function on faces of $A_n$. However, we will consider them both as functions on $\Lambda$.

\old{
\begin{figure}
 \begin{center}
 \includegraphics[width=0.7\textwidth]{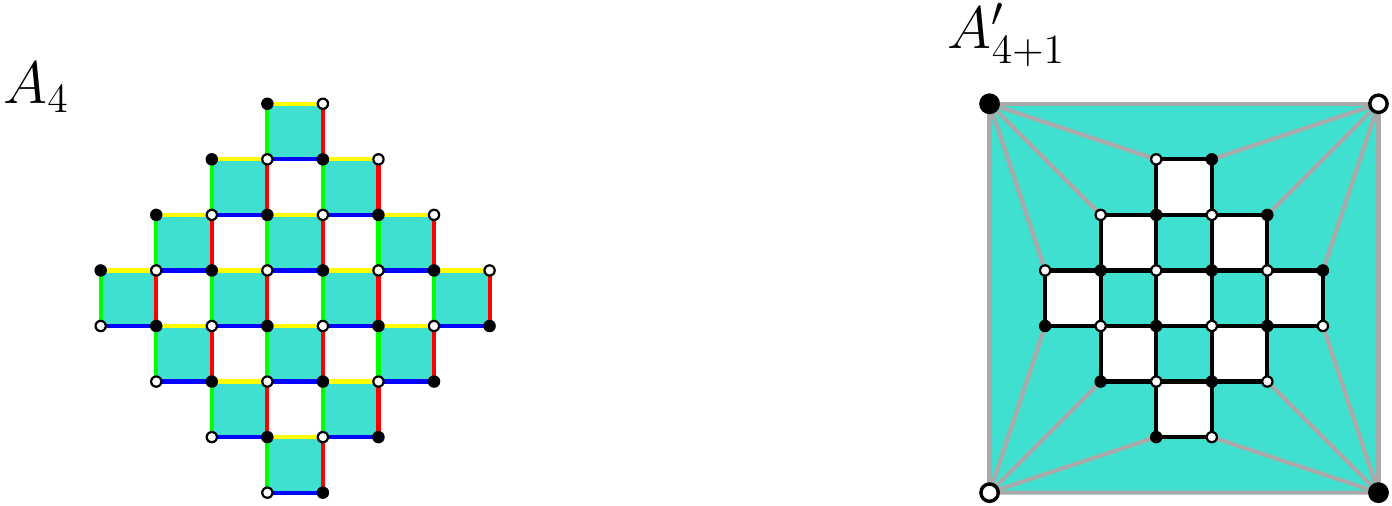}
  \caption{Above is a size $n =4 $ Aztec diamond (left) and a size $5 = 4 +1$ reduced Aztec diamond (right). The faces with $j + k + n$ odd are shaded in turquoise. In the Aztec diamond, the east edges are shaded in red, north edges are shaded in yellow, south in blue, and west in green.}\label{fig:NWSE_edges}
 \end{center}
\end{figure}
}

Define the function $f_{n\geq 1}:\Lambda\to\mathbb{R}$ by 
\[f_{n\geq 1}(j, k, n) := 1-p_E(j-1, k, n-1)  - p_W(j+1, k, n-1) - p_N(j, k-1, n-1)-p_S(j, k+1, n-1).\]
Due to the \emph{shuffling algorithm} for $n\geq 1$ the edge probability update rule for the uniform dimer model on the Aztec diamond is given by
\begin{align}\label{eq:p_e_shuffling}
\begin{split}
p_E(j, k, n) &= p_E(j-1, k, n-1) + \tfrac{1}{2} f_{n\geq 1}(j, k, n), \\
p_S(j, k, n) &= p_S(j, k+1, n-1)+ \tfrac{1}{2} f_{n\geq 1}(j, k, n), \\
p_N(j, k, n) &= p_N(j, k-1, n-1)+ \tfrac{1}{2} f_{n\geq 1}(j, k, n),\\
p_W(j, k, n) &= p_W(j+1, k, n-1) + \tfrac{1}{2} f_{n\geq 1}(j, k, n),
\end{split}
\end{align}
see e.g.~\cite[Section 3]{propp2003generalized}.
Therefore, for $n\geq 1$ one has 
\begin{align*}
 f_{n\geq 1}(j, k, n)&=2 (p_E(j, k, n) - p_E(j-1, k, n-1)) = 2(p_S(j, k, n) - p_S(j, k+1, n-1)) \\
&= 2 (p_N(j, k, n) - p_N(j, k-1, n-1)) = 2 (p_W(j, k, n) - p_W(j+1, k, n-1)).
\end{align*}

 \begin{lemma}\label{lem:fund}
 Let $f(j, k, n) :=2 (p_E(j, k, n) - p_E(j-1,k, n-1))$. Then $f$ is the fundamental solution of the recurrence~\eqref{eq:rec-unif} with the initial condition~\eqref{bc}.
 \end{lemma}
 
\begin{proof}
Note that for~$n\leq 0$ the function $f(j, k, n)$ vanishes and therefore satisfies the initial condition~\eqref{bc}. It remains to check that $f$ satisfies recurrence~\eqref{eq:rec-unif} with 
\[(b_0,b_E,b_N,b_W,b_S)=(1,0,0,0,0).\]

It is easy to see, that~$f$ satisfies recurrence~\eqref{eq:rec-unif} for $n\leq 0$. Indeed, note that for~$n\leq 0$ the function~$f(j, k, n)$ vanishes and~$f(j, k, 1)=\delta_{j,0}\delta_{k,0}$ since~$p_E(0,0,1) =\frac{1}{2}$, $p_E(j, k, 1) =0$ for~$(j, k) \neq (0, 0)$, and~$p_E(j, k, 0) = 0$ for all~$j, k$.

Assume now that $n \geq 1$; then the function $f(j, k, n)$ coincides with the function $f_{n\geq 1}(j, k, n)$ defined above. Therefore, the edge probability update rule~\eqref{eq:p_e_shuffling} implies that  
\begin{align}\label{eq:ffp_1}
\begin{split}
&f(j, k, n+1) =  1 - \frac12\big(f(j-1,k,n)+f(j+1,k,n)+f(j,k+1,n)+f(j,k-1,n)\big)\\
&\,\,-p_E(j-2, k, n-1)  - p_W(j+2, k, n-1) - p_N(j, k-2, n-1)-p_S(j, k+2, n-1).  
\end{split}
\end{align}

Note that the edge probabilities around a vertex sum to $1$. Using this fact for each vertex of the face~$(j,k,n-1)$ and 
expanding the function~$f$ at the point~$(j-1,k,n)$ by
\[f(j-1,k,n) =1 - p_E(j-2, k, n-1)  - p_W(j, k, n-1) - p_N(j-1, k-1, n-1)-p_S(j-1, k+1, n-1),\] 
and similarly at points~$(j+1,k,n), (j,k+1,n)$ and~$(j,k-1,n)$ we obtain
\begin{align}\label{eq:ffp_2}
\begin{split}
&f(j, k, n-1)=
-1 +f(j-1,k,n)+f(j+1,k,n)+f(j,k+1,n)+f(j,k-1,n)\\
&\,\,+p_E(j-2, k, n-1)  + p_W(j+2, k, n-1) + p_N(j, k-2, n-1) + p_S(j, k+2, n-1). \\
\end{split}
\end{align}

To complete the proof note that~\eqref{eq:ffp_1} together with~\eqref{eq:ffp_2} imply that~$f$ satisfies
 recurrence~\eqref{eq:rec-unif}  for~$n \geq 1$ as well.
\end{proof}

The following corollary 
gives a representation of the perfect t-embedding~$\T_n$ and its origami map~$\Or_n$ in terms of edge probabilities.

\begin{corollary}\label{cor:p_E_and_f_E}
Let $p_E, p_S, p_N$ and $p_W$ be edge probabilities as defined above, then 
\[p_E = f_E, \quad p_S = f_S, \quad p_N = f_N  \quad \text{ and } \quad p_W = f_W,\]
where the functions $f_E, f_S, f_N$ and $f_W$ are as defined in Section~\ref{sec:aztec_rec}. In particular, for~${j + k + n}$ odd
\begin{equation}\label{eq:T_f_e}
\T_n(j,k)=p_E(j,k,n)+\i p_N(j,k,n)-p_W(j,k,n)-\i  p_S(j,k,n),
\end{equation}
\begin{equation}\label{eq:O_f_e}
\Or_n(j,k)=p_E(j,k,n)+\i p_N(j,k,n)+p_W(j,k,n)+\i  p_S(j,k,n).
\end{equation}
\end{corollary}

\begin{proof}
Again denote $f(j, k, n) = 2 (p_E(j, k, n) - p_E(j-1,k,n-1))$. The claim $p_E = f_E$ follows from Lemma~\ref{lem:f_e_unif} and Lemma~\ref{lem:fund}, along with the fact~\cite[Equation (1.5)]{cohn-elki-prop-96} that
\begin{equation}\label{eqn:peformula}
p_E(j, k, n) = \frac{1}{2} \sum_{s = 0}^{n+1} f(j-s, k, n-s).
\end{equation}

For $W, N$ and $S$ the proof is similar.
\end{proof}

\begin{remark}
Equation \eqref{eqn:peformula} first appeared as Equation (1.5) in \cite{cohn-elki-prop-96}, where the authors use this representation to prove the convergence of $p_E, p_W, p_N, p_S$ in the $n \rightarrow \infty$ limit. 
\end{remark}

\begin{remark}
An important component in deriving the above formula came from the observation that the recurrence relation for the fundamental solution has appeared before. In fact, in~\cite{DFSG14} the authors derive a recurrence relation for the so-called density function of a weighted Aztec diamond, and this recurrence relation, in the special case when all weights are one, coincides with the relation appearing here. Although we don't expect a direct generalization of Corollary~\ref{cor:p_E_and_f_E} to hold if we vary the weights, we still believe that the density functions could be used also to describe the perfect t-embedding corresponding to certain weighted Aztec diamonds. 
\end{remark}

\section{Perfect t-embeddings of tower graphs}\label{sec:tower} The goal of this section is to introduce a construction of perfect t-embeddings of so-called \emph{tower graphs} and show how they are related to t-embeddings of Aztec diamonds. The construction is very similar to the case of the Aztec diamond and also based on the shuffling algorithm. The main difference is that in the shuffling algorithm for tower graphs there are two rounds of spider moves. We will show that perfect t-embeddings of Aztec diamonds can be obtained by a change of coordinates from the perfect t-embeddings of tower graphs. 

\subsection{The tower graph and its reduction}
Let us first recall a definition of a tower graph  introduced in~\cite[Section 4]{borodin2015random} and similarly to the Aztec diamond case introduce a reduced tower graph. 

A tower graph~$\TG_n$ of size~$n$ is a union of~$3n-1$ columns of faces with~$n$ columns of hexagonal faces separated by two columns of square faces, where the first column contain~$n$ hexagons,  and the last column has~$2n$ squares. More precisely, for~$p\geq 0$ the~$(3p+1)$st column contains~$n+p$ hexagons and it is adjacent to a column of~$n+p-1$ squares on the left (if~$p > 0$) and to a column of~$n+p+1$ squares on the right, as shown on Figure~\ref{fig:tower}.


Note that faces of tower graphs have a structure of a square lattice (send a length of all horizontal edges of hexagonal faces to zero to see this). So we can naturally assign square lattice coordinates to faces of tower graphs.
For lattice coordinates $(j, k)$ on the faces of the tower graph lattice we choose a coordinate axes $\widehat{j}, \widehat{k}$ as shown on Figure~\ref{fig:tower}. We declare
the origin~$(0, 0)$ to correspond to the following face: Starting at the ``bottom left''
hexagonal face, move~$n-1$ units in the~$\widehat{k}$ direction, and the face we end up at is
at $(0, 0)$. Note that the set of faces of a tower graph $\TG_n$ of size $n$ is given by $(j, k)$ such that 
\begin{equation}\label{eq:tower_coord}
\begin{cases}
 j+k < 2n \quad &\text{ for } j, k \geq 0,\\
 |j|+|k| <n \quad &\text{ for } j, k \leq 0,\\
 k-2j < 2n \quad &\text{ for } j<0, k>0,\\
 j-2k < 2n \quad &\text{ for } j>0, k<0,
\end{cases}
\end{equation}
where all faces of the domain with \[-j-k+2n=2 \mod 3\] are hexagons and all other faces are squares, see Figure~\ref{fig:tower}.

To obtain a reduced tower graph~$\TG'_n$ 
from a tower graph~$\TG_n$ of size $n$ one should make the following sequence of moves:
\begin{itemize}
\item[$\bullet$] contract degree $2$ white vertices of faces $(j, k)$ of $\TG_{n}$ with $k-2j = 2n-1$ or $j-2k = 2n-1$;
\item[$\bullet$] contract degree $2$ black vertices of faces $(j, k)$ of $\TG_{n}$ with  $j+k = 2n-1$ or $ |j|+|k| =n-1$.  
\end{itemize}
Note that,  in contrast to the Aztec setup, the reduced  tower graph has the same set of inner faces as the tower graph itself, since for the tower graph we consider double edges as degree two faces rather than single edges with twice the weight as in the Aztec case, 
see Figure~\ref{fig:tower}. Therefore we can associate coordinates to all inner faces of a reduced tower graph~$\TG'_n$ in the same way as we did for faces of tower graphs, i.e. using~\eqref{eq:tower_coord}. 
The boundary vertex of augmented dual~$(\TG'_{n})^*$ adjacent to the dual vertex indexed by~$(0,-(n-1))$ we index by~$(0,-n)$. Similarly we index the other three boundary vertices of the augmented dual~$(\TG'_{n})^*$ by~$(0,2n)$, $(2n,0)$ and~$(-n,0)$.

\subsection{Recurrence relation for perfect t-embeddings of tower graphs}
The construction of perfect t-embeddings of tower graphs mimics the one for Aztec diamonds. The main idea is the same as in Section~\ref{sec:aztec_rec}: We use the shuffling algorithm~\cite{borodin2015random, GIClusterMod} to construct the tower graph of size~$n+1$ from the tower graph of size~$n$ using elementary transformations.

One can reconstruct the tower graph of size $n+1$ from the size $n$ tower graph using two rounds of spider moves, 
see~\cite[Figure 13]{tower_localstat}.  And similarly one can get the reduced tower graph~$\TG'_{n+1}$ from~$\TG'_{n}$. Since we construct perfect t-embeddings of reduced graphs, we only describe the shuffling algorithm for reduced ones, see Figure~\ref{fig:reduced_tower_update}. We refer the interested reader to e.g.~\cite{tower_localstat} for the shuffling algorithm for tower graphs themselves.  Due to two rounds of spider moves in the shuffling algorithm, we use a discrete time parameter $m$ which is related to $n$ as $m = 2 n$. 
Denote by~$\widetilde\T_{n}(j,k)$ the position of the perfect t-embedding of the vertex of~$(\TG'_{n})^*$ indexed by~$(j,k)$.
And in addition to $\widetilde{\T}_{n}(j,k)$ we will define~$\widetilde{\T}(j, k, m)$, the position of the dual vertex $(j, k)$ in the t-embedding at time $m$. For~$m=2n$ the values of $\widetilde{\T}(j, k, m)$ simply coincide with $\widetilde{\T}_{n}(j,k)$ and for odd~$m$ we will define~$\widetilde{\T}(j, k, m)$  recursively. 

\old{We fix the positions of boundary vertices of the augmented dual~$(\TG'_{n})^*$ at~$1, i, -1$ and~$-i$:
\begin{equation}
\widetilde{\T}_{n}(2n,0)=1, \quad \widetilde{\T}_{n}(0,2n)=i, 
\quad \widetilde{\T}_{n}(-n,0)=-1, \quad \widetilde{\T}_{n}(0,-n)=-i.
\end{equation}
We also fix boundary conditions for~$\widetilde{\T}(j, k, m)$ as follows
\begin{align*}
&\widetilde{\T}(2n, 0, 2n-2)=\widetilde{\T}(2n+1, 0, 2n-1)=1;\\
&\widetilde{\T}(0, 2n, 2n-2)=\widetilde{\T}(0, 2n+1, 2n-1)=i;\\
&\widetilde{\T}(-n, 0, 2n-2)=\widetilde{\T}(-n-1, 0, 2n-1)=-1;\\
&\widetilde{\T}(0, -n, 2n-2)=\widetilde{\T}(0,-n-1, 2n-1)=-i.
\end{align*}}

We start out with~$m = 2$ (i.e.~$n=1$), note that~$\TG_1$ has three interior faces~${\{(0, 0), (1,0), (0,1)\}}$. 
 The positions of the inner vertices of~$(\TG'_{1})^*$ are ~$\widetilde{\T}_1(0,0)=\widetilde{\T}(0, 0,  2) = 0$, $\widetilde{\T}_1(0,1)=\widetilde{\T}(0,1, 2)= \frac{i}{2}$ and~$\widetilde{\T}_1(1,0)=\widetilde{\T}(1,0, 2) = \frac{1}{2}$.

Define~$\HS_{m}$ to be the set of faces $(j, k)$ such that 
\[\begin{cases}
 j + k \leq  m  \quad &\text{ for } j, k \geq 0,\\
 -j - k \leq \frac{m}{2}  -1 \quad &\text{ for } j, k \leq 0,\\
 j - 2 k \leq  m \quad &\text{ for } j>0, k<0,\\
 k - 2 j \leq  m \quad &\text{ for } j<0, k>0.
\end{cases}\]
Note that the set~$\HS_{2n+1}$ coincides with the set of inner faces of $\TG_{n+1}$.

\old{
\begin{figure}
 \includegraphics[width=0.65\textwidth]{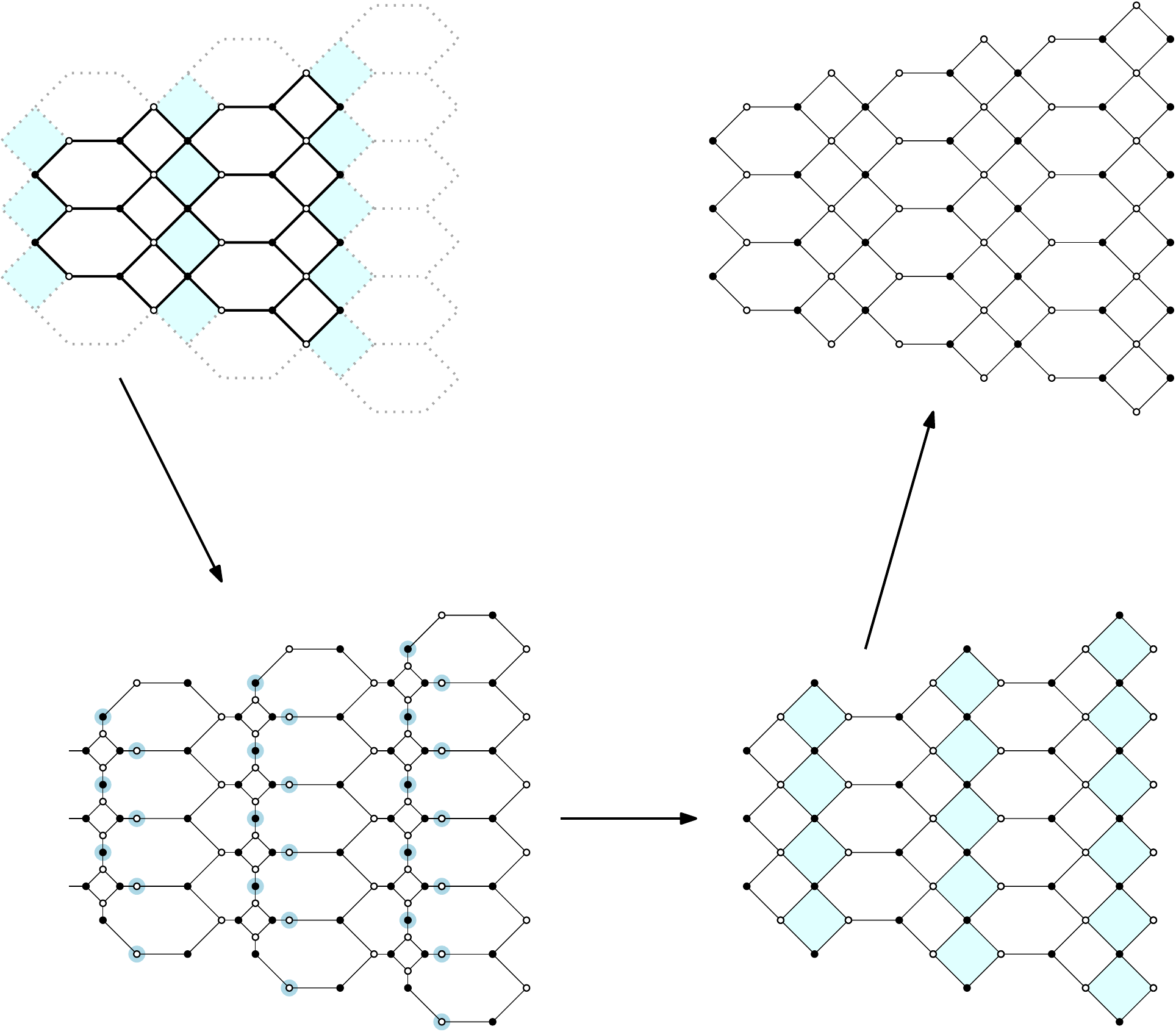}
\caption{The sequence above illustrates the sequence of elementary transformations which, starting with a decorated tower graph of size~$n$, produces a tower of size~$n+1$ for $n=2$. The decoration is shown in dotted grey, faces where one should apply a spider move and vertices degree two which should be contracted emphasised in light blue. The last arrow corresponds to a combination of spider moves and contracting degree two vertices after it. }
\label{fig:tower_shuffle}
\end{figure}
}

\begin{figure}
 \begin{center}
 \includegraphics[width=0.75\textwidth]{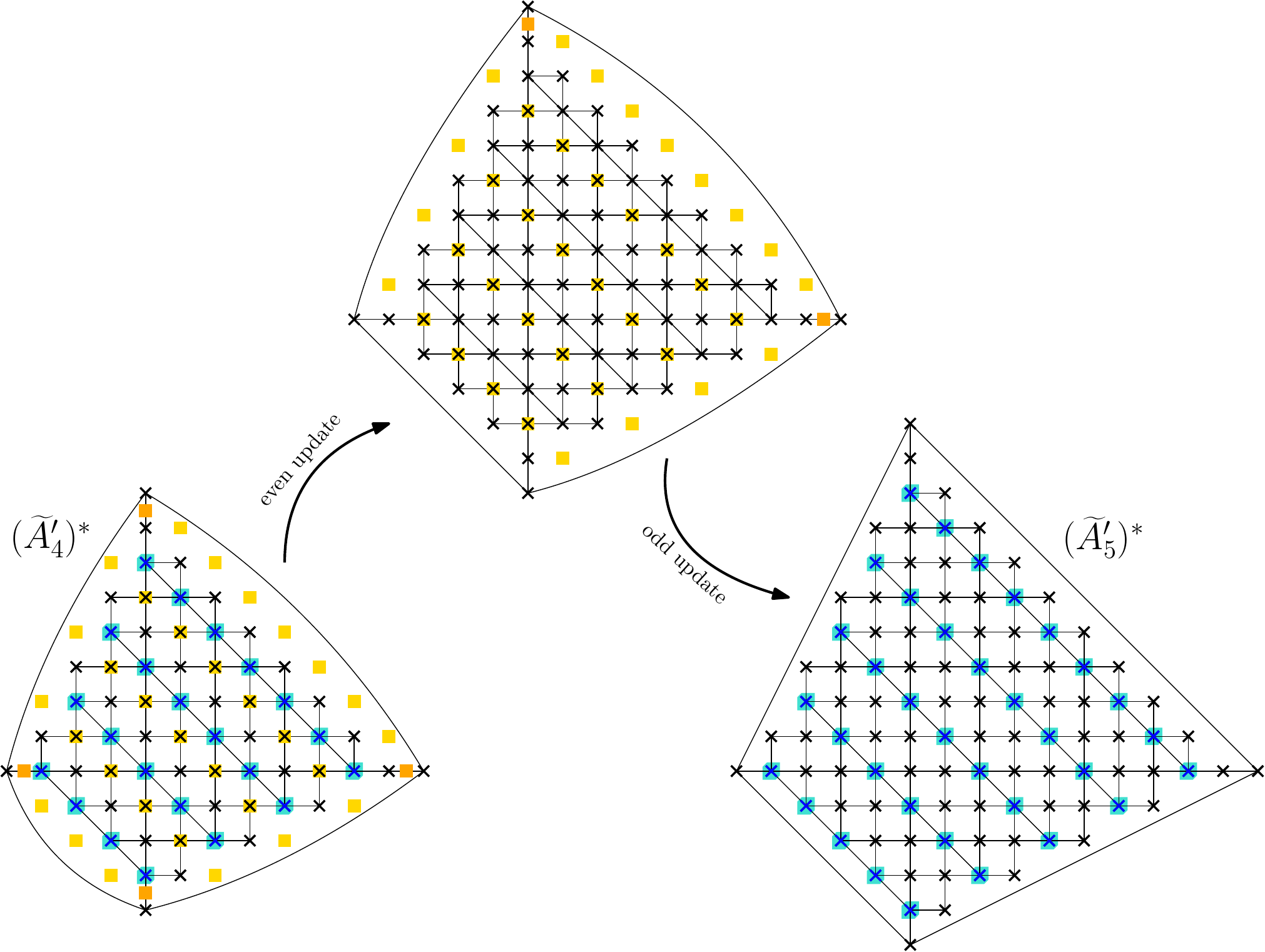}
 \end{center}
\caption{ Even and odd update rules to obtain the augmented dual of~$\TG'_{5}$ from the augmented dual of~$\TG'_{4}$. Hexagonal faces of reduced tower graphs emphasised in  turquoise.
Update rules: we apply central move at yellow faces and corner updates at orange points. }\label{fig:update}
\end{figure}

The following proposition is an analog of Proposition~\ref{prop:t_rec_aztec} for tower graphs.

\begin{proposition}\label{prop:t_rec} 
The perfect t-embedding~$\widetilde{\T}_{n+1}(j,k)$ can be obtained from~$\widetilde{\T}_{n}(j,k)$ using the following update rules

{\bf Boundary conditions:} for all $n\in\mathbb{Z}_{>0}$ we fix
\begin{equation*}
\widetilde{\T}_{n}(2n,0)=1, \quad \widetilde{\T}_{n}(0,2n)=i, 
\quad \widetilde{\T}_{n}(-n,0)=-1, \quad \widetilde{\T}_{n}(0,-n)=-i.
\end{equation*}

{\bf Even update:} For all $(j, k)\in\HS_{2n}$ define \[\widetilde{\T}(j, k, 2n):=
\begin{cases}
\widetilde{\T}_{n}(j,k) &\text{ if } (j,k) \text{ is a vertex of } (\TG'_{n})^*,\\
0 &\text{ otherwise}.
\end{cases}\]

\begin{enumerate}[(a)]
\item Boundary conditions. Set
\begin{align*}
\widetilde{\T}(2n+1, 0, 2n+1)=1;& \quad \widetilde{\T}(-n-1, 0, 2n+1)=-1;\\
\widetilde{\T}(0, 2n+1, 2n+1)=i;& \quad \widetilde{\T}(0,-n-1, 2n+1)=-i.
\end{align*}
\item For each $(j, k)\in\HS_{2n}$ such that 
\[\{j,k\}\neq\{0,2n\}, \quad \{j,k\}\neq\{-n,0\} \quad \text{ and } \quad -j - k + (2n+1) = 1 \mod 3\]
 we update $\widetilde{\T}$ according to the recurrence:
\begin{multline*}
\widetilde{\T}(j, k, 2n+1)+\widetilde{\T}(j, k, 2n) =\frac{1}{2} \Big( \widetilde{\T}(j-1, k, 2n) + \widetilde{\T}(j+1, k, 2n)\Big.\\
\Big. + \widetilde{\T}(j, k-1, 2n) + \widetilde{\T}(j, k+1, 2n) \Big).
\end{multline*}
Otherwise update $\widetilde{\T}(j, k, 2n+1) = \widetilde{\T}(j, k, 2n)$.
\item \emph{Corner Update}. Perform updates
\begin{align*} 
&\widetilde{\T}(2n, 0, 2n+1) = \tfrac{1}{2}\left( \widetilde{\T}(2n-1, 0, 2n) + \widetilde{\T}(2n, 0, 2n)\right),\\
&\widetilde{\T}(0, 2n, 2n+1) = \tfrac{1}{2}\left( \widetilde{\T}(0, 2n-1, 2n) + \widetilde{\T}(0, 2n, 2n)\right),\\
&\widetilde{\T}(-n, 0, 2n+1) = \tfrac{1}{2}\left( \widetilde{\T}(-n+1, 0, 2n) + \widetilde{\T}(-n, 0, 2n)\right),\\
&\widetilde{\T}(0, -n, 2n+1) = \tfrac{1}{2}\left( \widetilde{\T}(0, -n+1, 2n) +  \widetilde{\T}(0, -n, 2n)\right).
\end{align*}
 \end{enumerate}
 \item[ ] 
{\bf Odd update:} Let $\widetilde{\T}(j, k, 2n+1)$ be as obtained by even update. 
\begin{enumerate}[(a)]
\item Boundary conditions. Set
\begin{align*}
\widetilde{\T}(2n+2, 0, 2n+2)=1;& \quad \widetilde{\T}(-n-1, 0, 2n+2)=-1;\\
\widetilde{\T}(0, 2n+2, 2n+2)=i;& \quad \widetilde{\T}(0,-n-1, 2n+2)=-i.
\end{align*}
\item For each $(j, k)\in\HS_{2n+1}$ such that 
\[\{j,k\}\neq\{0,2n+1\} \quad \text{ and } \quad  -j - k + 2n+2 = 1 \mod 3\]
 we update $\widetilde{\T}$ according to the recurrence:
\begin{multline*}
\widetilde{\T}(j, k, 2n+2)+\widetilde{\T}(j, k, 2n+1) =\frac{1}{2} \Big(  \widetilde{\T}(j-1, k, 2n+1) +  \widetilde{\T}(j+1, k, 2n+1) \Big.\\
\Big. + \widetilde{\T}(j, k-1, 2n+1) + \widetilde{\T}(j, k+1, 2n+1) \Big).
\end{multline*}
Otherwise update $\widetilde{\T}(j, k, 2n+2) = \widetilde{\T}(j, k, 2n+1)$.
\item \emph{Corner Update}. Perform updates
\begin{align*}
\widetilde{\T}(2n+1, 0, 2n+2) = \frac{1}{2}( \widetilde{\T}(2n, 0, 2n+1) + \widetilde{\T}(2n+1, 0, 2n+1)), \\
\widetilde{\T}(0, 2n+1, 2n+2) = \frac{1}{2}( \widetilde{\T}(0, 2n, 2n+1) +  \widetilde{\T}(0, 2n+1, 2n+1)).
\end{align*}
\end{enumerate}
Now~$\widetilde{\T}_{n+1}(j,k)$ is given by  $\widetilde{\T}(j, k, 2n+2)$.
\end{proposition}

\begin{remark}
At any step where $\widetilde{\T}(j, k, m)$ has not been defined but its value is used in an update, we set it to $0$.
\end{remark}

\begin{figure}
 \begin{center}
 \includegraphics[width=0.8\textwidth]{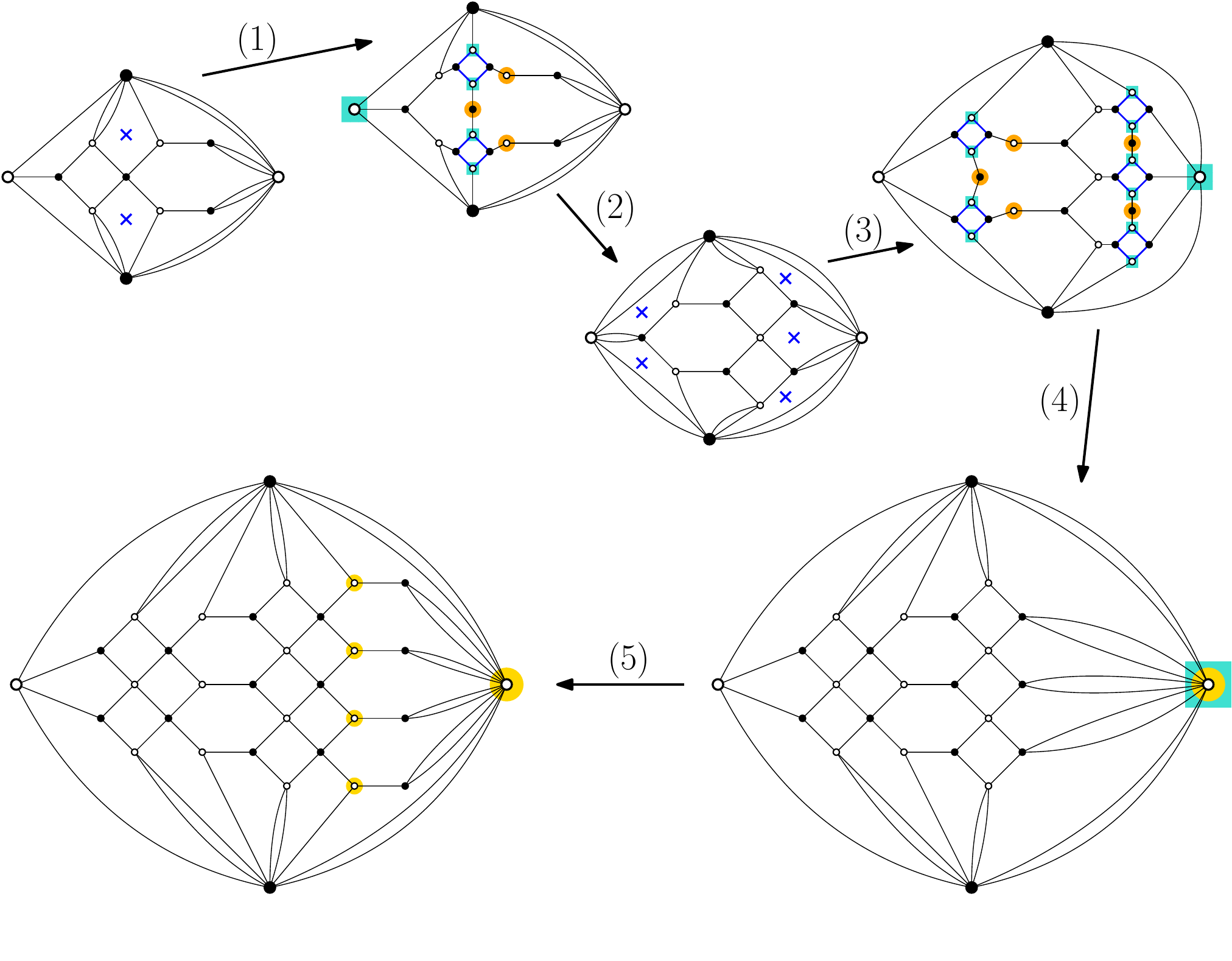}
$\,$ \vline $\,$ $\,$ 
  \includegraphics[width=0.15\textwidth]{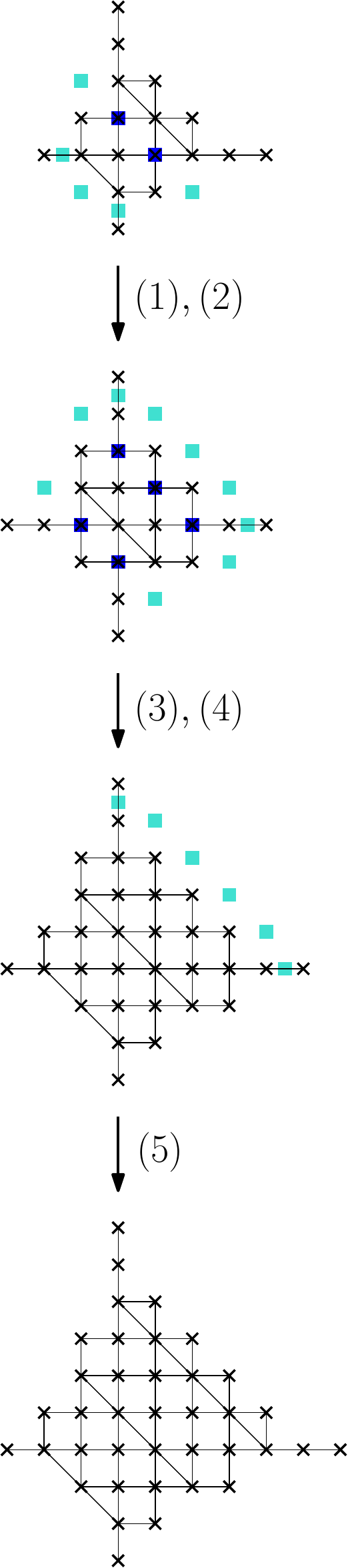}
 \end{center}
\caption{{\bf Left: }A sequence of~$5$ steps of elementary transformations producing reduced tower graph of size~$3$ from a reduced tower graph of size~$2$. Apply spider move at all faces marked with blue crosses; contract degree~$2$ vertices emphasised in orange; multiply all weights on edges adjacent to vertices emphasised in turquoise colour and double all edges with weight~$2$. The final step also includes splitting the marked boundary white vertex into vertices emphasised in yellow. Weights of black edges are~$1$ and of blue ones are~$\tfrac12$. {\bf Right:} Schematic representation of the same elementary transformations for corresponding augmented dual graphs.
Recall, that contracting/splitting vertices of the graph do not affect positions of the dual vertices in the t-embedding.
}\label{fig:reduced_tower_update}
\end{figure}

\begin{proof} The proof mimics the proof of~\cite[Proposition 2.4]{Ch-R}. To obtain reduced tower graph  of size $n+1$ from a reduced tower graph of size $n$ one should make the following sequence of moves, see Figure~\ref{fig:reduced_tower_update}: 
\begin{enumerate}
\item apply a spider move at each face indexed by $(j, k)$ with $2n-1-j-k = 1 \mod 3$;
\item note that spider moves in the previous step do change the set of faces, now:
\begin{itemize}
\item[$\bullet$] apply the following gauge transformation at all new white vertices obtained by the spider moves in the previous step and at the white boundary vertex adjacent to boundary faces indexed by $(-n,0)$ and $(0,-n)$: multiply all edge weights around these vertices by $2$; 
\item[$\bullet$] contract all degree~$2$ vertices obtained by spider moves;
\item[$\bullet$] split all edges of weight~$2$ to double edges with weight~$1$ each;
\end{itemize}
\item apply a spider move at each face indexed by $(j, k)$ with $2n-j-k = 1 \mod 3$;
\item in the obtained graph:
\begin{itemize}
\item[$\bullet$] at all new white vertices obtained by spider moves in the previous step and at the white boundary vertex adjacent to boundary faces indexed by $(2n,0)$ and $(0,2n)$ multiply all edge weights around these vertices by $2$; 
\item[$\bullet$] contract all degree~$2$ vertices obtained by spider moves;
\item[$\bullet$] split all edges of weight~$2$ to double edges with weight~$1$ each;
\end{itemize}
\item split the white boundary vertex adjacent to boundary faces indexed by~$(2n+1,0)$ and~$(0,2n+1)$ to~$2n+1$ white vertices as shown on Figure~\ref{fig:reduced_tower_update}; in the obtained graph:
\begin{itemize}
\item[$\bullet$] at the white boundary vertex adjacent to boundary faces indexed by~$(2n+1,0)$ and~$(0,2n+1)$ multiply all edge weights by~$2$; 
\item[$\bullet$] split all edges of weight~$2$ to double edges with weight~$1$ each.
\end{itemize}
\end{enumerate}

Recall that t-embeddings are preserved under elementary transformations and in the above steps we never apply spider move to dual vertices adjacent to boundary ones, therefore t-embeddings stay perfect. 

The update rules of t-embeddings~$\widetilde{\T}$ correspond to the update steps of reduced tower graphs. 
To see this let us split both $\operatorname{Even}$ and $\operatorname{Odd}$ updates described in the proposition to $\operatorname{Even}_1 \sqcup \operatorname{Even}_2$ and $\operatorname{Odd}_1 \sqcup \operatorname{Odd}_2$ in the following way:  

 {\bf Even$_1$ update: } Apply Even~(b) for all 
$(j, k)\in\HS_{2n}$ such that 
\[\{j,k\}\neq\{0,2n\}, \quad \{j,k\}\neq\{-n,0\}, \quad -j - k + 2n+1 = 1 \mod 3 \quad \text{ and } \quad j+k\neq 2n\]
and apply the corner update Even~(c) at points $(-n,0,2n-1)$ and $(0,-n,2n-1)$ only.

{\bf Even$_2$ update: } Apply Even~(b) for all 
$(j, k)\in\HS_{2n}$ such that 
\[\{j,k\}\neq\{0,2n\} \quad \text{ and } \quad j+k= 2n\]
and also apply the corner update Even~(c) at points $(2n,0,2n+1)$ and $(0,2n,2n+1)$.

 {\bf Odd$_1$ update: } Apply Odd~(b) for all 
$(j, k)\in\HS_{2n+1}$ such that 
\[\{j,k\}\neq\{0,2n+1\}, \quad  -j - k + 2n+2 = 1 \mod 3 \quad  \text{ and } \quad  j+k\neq 2n+1. \]

 {\bf Odd$_2$ update: } Apply Odd~(b) for all 
$(j, k)\in\HS_{2n+1}$ such that 
\[\{j,k\}\neq\{0,2n+1\} \quad  \text{ and } \quad  j+k = 2n+1 \]
and also apply the corner update Odd~(c).

Note that the positions of points obtained in Even$_2$ update are not used in Odd$_1$ update, therefore one can first apply Odd$_1$ and then Even$_2$. 
To conclude, note that Even$_1$ update corresponds to the update steps~(1)--(2) of reduced tower graphs;  Odd$_1$ and Even$_2$ updates correspond to the update steps~(3)--(4); and finally, Odd$_2$ update corresponds to update step~(5) of reduced tower graphs.
\end{proof}

\begin{remark} 
 Given the proof one may wonder why
the update rules in Proposition~\ref{prop:t_rec} are divided to Even and Odd in the way described in the statement 
and not in the way described in the proof. The statement is formulated in this way 
to make it easier to apply this proposition in the next section.
\end{remark}

Let~$\widetilde \Or_{n}(j,k)$ be the origami map of~$(\widetilde A_{n}')^*$, where the root white face is the one adjacent to~$\widetilde \T_n(n,0)$ and~$\widetilde \T_n(0,n)$. Also, let~$\widetilde \Or(j, k, m)$ denote the image of~$\widetilde{\T}(j, k, m)$ under the origami map corresponding to~$\widetilde{\T}$. Then the following holds.

\begin{proposition}\label{prop:o_rec_tower}
For $n \geq 1$ the origami map $\widetilde \Or_{n+1}$ can be constructed from $\widetilde \Or_n$ using the same update rules {\bf Even update} (a)-(c) and {\bf Odd update} (a)-(c) as in Proposition~\ref{prop:t_rec}, with boundary conditions 
\begin{equation}\label{eq:or_bdry_tower}
 \widetilde \Or_n(2n,0) = \widetilde \Or_n(-n,0) = 1\quad\text{ and } \quad \widetilde \Or_n(0,2n)= \widetilde \Or_n(0,-n) = \i .
 \end{equation}
 Boundary conditions in (a) of each step should be changed accordingly as well.
\end{proposition}
\begin{proof}
This proposition follows from Proposition~\ref{prop:t_rec} combined with Corollary~\ref{cor:or_elem}.
\end{proof}

\old{
\begin{lemma}
Let $\tilde{f}:\Lambda_+ \to\mathbb{R}$ be given by 
$\tilde{f}(j, k, m) := \widetilde{\T}(j, k, m).$ 
Then the following holds.
\begin{multline*}
\tilde{f}(j, k, m+1) + \tilde{f}(j, k, m-2)= \frac{1}{2} \Big(\tilde{f}(j+1, k, m-1) + \tilde{f}(j, k+1,m-1)\\
+ \big( \tilde{f}(j-1, k, m) +  \tilde{f}(j, k-1,m)\big) \Big) 
\end{multline*}

for all $(j, k, m+1) \in \Lambda_+ \setminus \{\text{corners}\}$, and at the corners one has
\[
\begin{cases}
\tilde{f}(m+1, 0, m) = \frac{1}{2} (\tilde{f}(m, 0, m-1) + 1), \\
\tilde{f}(0, m+1, m) = \frac{1}{2} (\tilde{f}(0, m, m-1) + i), \\
\tilde{f}(-(m+2)/2, 0, m) = \frac{1}{2} (\tilde{f}(-m/2, 0, m-2) - 1), \\
\tilde{f}(0, -(m+2)/2, m) = \frac{1}{2} (\tilde{f}(0, -m/2, m-2)  - i).
\end{cases}\]
\end{lemma}
\begin{proof} It is easy to check using update rules for $\T(j, k, m)$.
\end{proof}
}

\old{Let us perform a change of coordinates. Define
 \[G(j, k, m') := F(j, k, \frac{3}{2} m' - \frac{1}{2}(j + k) - 5/2)\]
 at points $(j, k, m') \in \mathbb{Z}^3$ such that
 $ j + k + m' = 1  \mod 2.$ In the next section we show that such a function $G$ gives a position of t-embeddings of reduced Aztec diamonds.
 }

\subsection{Perfect t-embeddings of Aztec diamonds and tower graphs}
In this section we prove Theorem~\ref{thm:Aztec_Tower} by performing a change of coordinates.

Let $\partial \widetilde\Lambda_+$ be the set of points $(j, k, m)\in\mathbb{Z}^2\times\mathbb{Z}_{\geq0}$ such that 
\[\begin{cases}
j + k = m - 1 \quad &\text{ for } j, k \geq 0,\\
 -j - k = \frac{m-1}{2} \quad &\text{ for } j, k < 0, \, m \text{ odd},\\
 j - 2 k = m-1 \quad &\text{ for } j>0, k<0,\\
 k - 2 j = m-1 \quad &\text{ for } j<0, k>0.
\end{cases}\]
Note that  $m - j - k = 1 \mod 3$ for all $(j, k, m)\in \partial \widetilde\Lambda_+.$ Let $\widetilde\Lambda_+ := \widetilde{\Lambda}_{\operatorname{int}} \cup \partial \widetilde\Lambda_+$, where ${\widetilde\Lambda}_{\operatorname{int}}$ is the set of points $(j, k, m)\in\mathbb{Z}^2\times\mathbb{Z}_{\geq0}$ such that 
\[\begin{cases}
m - j - k = 1 \mod 3,\\
j + k < m - 1,\\
 -j - k < \frac{m-1}{2},\\
 j - 2 k < m-1 \quad \text{ and }  \quad k - 2 j < m-1.\\
\end{cases}\]

Let $\widetilde\Lambda:=\{(j, k, m)\in\mathbb{Z}^3 \, | \, m - j - k = 1 \mod 3\}$.
Given~$(b_E,b_N,b_W,b_S) \in \CC^4$ 
the following conditions define the function $\tilde{f}:\widetilde\Lambda \to \CC$ uniquely. 
\begin{enumerate}
\item For all $(j,k)\in\mathbb{Z}^2$ 
\begin{equation} \label{bc_tower}
\tilde{f}(j,k,1)=0, \quad \tilde{f}(j,k,0)=0 \quad\text{ and }\quad \tilde{f}(j,k,-1)=0;
\end{equation}
\item If $(j, k, m+1) \in \widetilde\Lambda$ then
\begin{equation} \label{eq:rec-unif_tower}
\begin{split}
\tilde{f}&(j,k,m+1)+\tilde{f}(j,k,m-2)\\
 &- \frac12\big(\tilde{f}(j+1,k,m-1)+\tilde{f}(j,k+1,m-1)+\tilde{f}(j-1,k,m)+\tilde{f}(j,k-1,m)\big)\\
&\quad=
\frac12\big(\delta_{j,m}\delta_{k,0}\cdot b_E+
\delta^{2}_{m}\delta_{j,-\tfrac{m}{2}}\delta_{k,0}\cdot b_W+
\delta_{j,0}\delta_{k,m}\cdot b_N+
\delta^{2}_{m}\delta_{j,0}\delta_{k,-\tfrac{m}{2}}\cdot b_S\big),
\end{split}
\end{equation}
where by $\delta^{2}_{m}$ we denote $\delta_{\{m = 0 \mod 2\}}$.
\end{enumerate}

Note that $\tilde{f}(j,k,n)=0$ for all $(j,k,n)\in \widetilde\Lambda\smallsetminus\widetilde\Lambda_+$.

\begin{remark}\label{rmk:tower_f}
Similarly to Remark~\ref{rmk:aztec_f} the function $\widetilde\T(j,k,m)$ defined in Proposition~\ref{prop:t_rec} can be seen as the solution~$\tilde{f}(j,k,n)$ of the system~\eqref{bc_tower}--\eqref{eq:rec-unif_tower} with the boundary conditions~$(1,\, \i , -1, -\i )$. Here we use that $\widetilde\T(j,k,m)=\widetilde\T(j,k,m-1)$ for $m - j - k \neq 2 \mod 3$.
\end{remark}

The following theorem is a restatement of Theorem~\ref{thm:Aztec_Tower} in the introduction.

\begin{theorem}\label{thm:T_aztec_tower} Let $\T_{n}$ be the positions of the perfect t-embedding of the vertices of the augmented dual~$(A'_{n+1})^*$ of a reduced Aztec diamond of size~$n+1$ and let~$\widetilde{\T}(j,k,n)$ be as defined in the previous section. Similarly, let~$\Or_n$ be the positions of vertices of~$(A'_{n+1})^*$ under the origami map of~$\T_n$, and let~$\widetilde{\Or}(j,k,n)$ be as defined in the previous section. Then for all $j+k+n$ odd
\begin{align*}
\T_n(j,k) &=\widetilde{\T}\left(j, \,k, \,\tfrac32 n - \tfrac12(j+k)-\tfrac12\right) \\
\Or_n(j,k) &=\widetilde{\Or}\left(j, \,k, \,\tfrac32 n - \tfrac12(j+k)-\tfrac12\right) .
\end{align*}
\end{theorem}

\begin{remark}
Note that for $j+k+n$ odd the value $\tfrac32 n - \tfrac12(j+k)-\tfrac12$ is integer. Therefore the right hand side of each equation in the above proposition is well defined.
\end{remark}

\begin{proof} 
Recall that $\T_n(j,k)$ and $\Or_n(j,k)$ can each be seen as the solution~$f(j,k,n)$ of the system~\eqref{bc}--\eqref{eq:rec-unif} with the boundary conditions~$(0, 1,\i , -1, -\i )$ and $(0, 1,\i , 1, \i )$, respectively. Therefore, due to Remark~\ref{rmk:tower_f}, it is enough to check that for all $j+k+n$ odd
\[f(j,k,n)=\tilde{f}\left(j, \,k, \,\tfrac32 n - \tfrac12(j+k)-\tfrac12\right),\]
where~$\tilde{f}$ is the solution of the system~\eqref{bc_tower}--\eqref{eq:rec-unif_tower} with the boundary conditions~$(1,\, \i , -1, -\i )$ (or~$(1,\, \i , -1, -\i )$, for the origami map).

Note that for $n$ negative none of the points~$\left(j, \,k, \,\tfrac32 n - \tfrac12(j+k)-\tfrac12\right)$ lie in $\widetilde\Lambda_+$, and therefore both $f$ and $\tilde{f}$ vanish for $n\leq0$. To match the boundary conditions~\eqref{bc} and~\eqref{bc_tower} it remains to check that~$f$ is zero for all~$j+k+n$ odd such that~$\tfrac32 n - \tfrac12(j+k)-\tfrac12\leq 1$. Recall that all non-zero values of $f$ appear inside  the set $|j|+|k|\leq n$. So we need to check that in the intersection of these sets both functions vanishes. The first inequality is equivalent to~$3n-3\leq j+k$ using the inequality $j + k \leq |j| + |k| $ we see that the only point in the intersection of these sets is the point $(j,k,n)=(0,0,1)$. And since $f(0,0,1)=0$ for $b_0=0$, we obtain that~$f$ vanishes for all~$j+k+n$ odd such that~$\tfrac32 n - \tfrac12(j+k)-\tfrac12\leq 1$.

To finish the proof note that the recurrence relation~\eqref{eq:rec-unif} gives relation~\eqref{eq:rec-unif_tower} by performing a coordinate change~$(j,k,n) \mapsto \left(j, \,k, \,\tfrac32 n - \tfrac12(j+k)-\tfrac12\right)$.
\end{proof}

\section{Verification of assumptions}\label{sec:assumptions}
The goal of this section is to prove Theorem~\ref{thm:Aztec_assumptions} under the \emph{rigidity condition}, Assumption~\ref{assumption:bounded} stated below. The fact that perfect t-embeddings~$\T_n$ of uniformly weighted Aztec diamonds fulfill the rigidity condition follows from Lemma~\ref{lem:bdd}. We postpone the proof of this lemma until Section~\ref{sec:complicated_stuff}.

More precisely, in Section~\ref{sec:T_n_conv}, we obtain the leading order term in the asymptotic expansion of the perfect t-embeddings and origami maps, both for uniformly weighted Aztec diamond and tower graphs. 
Then, in Section~\ref{sec:lip_expfat} we show that Assumptions~\ref{assumption:Lip} and~\ref{assumption:Exp_fat} of Theorem~\ref{thm:CLR2thm} hold for any sequence~$\T_n$ of perfect t-embeddings satisfying the rigidity condition.
 Finally, in Section~\ref{sec:main_thm_proof}, we prove Theorem~\ref{thm:Aztec_assumptions} under the rigidity condition.

\subsection{Convergence of t-embeddings and origami maps.}\label{sec:T_n_conv}
In this section we prove Proposition~\ref{prop:eqn:to_conv}. The proposition implies that the sequence~$(\T_n,\Or_n')$ of perfect t-embeddings and origami maps of uniformly weighted Aztec diamonds converges to the maximal surface in $\mathbb{R}^{2,1}$. We also show that the corresponding sequence of maps of tower graphs~$(\widetilde\T_n,\widetilde\Or_n')$ converges to the same maximal surface, so Corollary~\ref{cor:convergence_tower} holds. In the process we also prove Corollary~\ref{Cor_main_frozen}, i.e. we show that under perfect t-embeddings, frozen regions  map to points in the scaling limit for the uniformly weighted Aztec diamond. 

As discussed in Section~\ref{sec:intro}, let us define the other version of the origami map of the reduced Aztec diamond~$A'_{n+1}$  by
\begin{align}\label{def:Or_prime}
\Or'_n(j,k):=\e^{i\tfrac{\pi}{4}}\left(\Or_n(j,k) - \tfrac{1+i}{2}\right),
\end{align}
which is a composition of a translation and rotation of the origami map $\Or_n$ defined in Section~\ref{sec:aztec}. Note that~$\Or'_n$ still satisfy Proposition~\ref{prop:o_rec} with the boundary conditions
\begin{equation}\label{eq:or_real_bdry}
 \Or'_n(n,0) = \Or'_n(-n,0) = \tfrac{\sqrt{2}}{2}\quad\text{ and } \quad\Or'_n(0,n)=\Or'_n(0,-n) = -\tfrac{\sqrt{2}}{2}.
 \end{equation}
Recall, that~$\T_n$ maps the four outer augmented dual vertices to~$1, \i, -1, -\i$.

Note that the domain~$\Omega_{\T_n}$ covered by~$\T_n$ for each~$n$ is equal to~$\diamondsuit \coloneqq \{z = x+ \i y: |x | +| y| < 1\}$.  Let us denote by~$\partial \diamondsuit$ the boundary of the domain~$\diamondsuit$. Then
 the graph of~$\Or_{n}'$ restricted to~$\partial \diamondsuit$ is the non-planar quadrilateral, which we denote by~$C_{\diamondsuit}$, 
 in~$\mathbb{C} \times \mathbb{R}$ with vertices
\[\left(1, \tfrac{\sqrt{2}}{2}\right), \quad
\left(i, -\tfrac{\sqrt{2}}{2}\right), \quad
\left(-1, \tfrac{\sqrt{2}}{2}\right) \quad
\text{ and } \quad
\left(-i, -\tfrac{\sqrt{2}}{2}\right).\]
We define~$S_{\diamondsuit}$ as the maximal surface in the Minkowski space~$\mathbb{R}^{2,1}$ with boundary contour~$C_{\diamondsuit}$. This is the unique space-like surface in~$\mathbb{R}^{2,1}$ with this boundary contour which has vanishing mean curvature.

In Section~\ref{sec:aztec_prob} we described a connection between perfect t-embeddings~$\T_n$  of reduced Aztec diamond~$A'_{n+1}$  of size~$n+1$ and edge probabilities of the Aztec diamond~$A_n$ of size~$n$. More precisely, recall that 
\begin{equation}\label{eq:T_f_e_reminder}
\T_n(j,k)=p_E(j,k,n)+\i p_N(j,k,n)-p_W(j,k,n)-\i  p_S(j,k,n),
\end{equation}

\begin{align}\label{eq:O_f_e_reminder}
\Or_{n}'(j,k)= & \tfrac{1}{\sqrt{2}} \left( p_E(j,k,n)- p_N(j,k,n)+p_W(j,k,n) -  p_S(j,k,n) \right) \\
&+ \i \tfrac{1}{\sqrt{2}} \left(  -1 + p_E(j,k,n) + p_N(j,k,n) + p_S(j,k,n) + p_W(j,k,n) \right). \notag
\end{align}
where~\eqref{eq:O_f_e_reminder} is obtained by substituting~\eqref{eq:O_f_e} into~\eqref{def:Or_prime}.
The convergence of edge probabilities according to the uniform distribution of the Aztec diamond was shown in
\cite{cohn-elki-prop-96}. We use this result to show the convergence of perfect t-embeddings of uniformly weighted Aztec diamonds and tower graphs.

\begin{proposition}[\cite{cohn-elki-prop-96}]\label{prop:edge_prob_convergence}
Let~$p_E(j,k,n)$ be the probability that the east edge of the face~$(j, k)$ is present in a uniformly random dimer cover of Aztec diamond~$A_n$ of size~$n$. Then, for all~$j+k+n$ odd and~$(j/n,k/n)$ staying away from~$\left(\tfrac12,\pm \tfrac12\right)$ one has
\begin{equation}\label{eq:CEP}
p_E(j,k,n)=o(1)+
\begin{cases}
\delta_{(j/n)>1/2} &\text{ if } (j/n)^2+(k/n)^2\geq\tfrac12,\\
\tfrac{1}{2}+\tfrac{1}{\pi}\tan^{-1}\left(\frac{2(j/n)-1}{\sqrt{1-2(j/n)^2-2(k/n)^2}}\right) &\text{ if } (j/n)^2+(k/n)^2<\tfrac12,
\end{cases}
\end{equation}
uniformly on compact subsets as~$n\to\infty$. 
\end{proposition}

\begin{remark}
Similarly, one get formulas for~$p_N$,~$p_W$ and~$p_S$. Indeed, the symmetry of the uniformly weighted Aztec diamond implies
\[p_E(j,k,n)=p_N(-k,j,n)=p_W(-j,-k,n)=p_S(k,-j,n).\]
\end{remark}

Note that the previous proposition, together with the representations~\eqref{eq:T_f_e_reminder}, \eqref{eq:O_f_e_reminder} of~$\T_n$ and~$\Or_n'$ in terms of edge probabilities, is enough to compute the limits~$z$ and~$\vartheta$. However, we want to match the limits with the limits predicted in~\cite{Ch-R}. To that end, we introduce some notation before we prove Proposition~\ref{prop:eqn:to_conv} (Propositio~\ref{prop:main_theorem_2} below). Following~\cite[Section 3]{Ch-R} 
we define the functions~$z(x,y)$ and~$\vartheta(x,y)$ for~$(x,y)\in \DAz\coloneqq \{x^2+y^2<\tfrac{1}{2}\}$ by
\begin{align}
z(x,y) &\coloneqq\Psi_E(x,y)+i\Psi_N(x,y)-\Psi_W(x,y)-i\Psi_S(x,y), \label{eq:limit_t}\\
\vartheta(x,y) &\coloneqq\tfrac{1}{\sqrt{2}}\left(\Psi_E(x,y)-\Psi_N(x,y)+\Psi_W(x,y)-\Psi_S(x,y)\right),\label{eq:limit_o}
\end{align}
where the function~$\Psi_E$ is given by
\begin{equation}\label{eq:psi_e}
\Psi_E(x,y)\coloneqq\int_0^1 \Psi_0(x-s,y,1-s)ds,
\end{equation}
with
\[\Psi_0(x,y,t)\coloneqq\begin{cases}
\tfrac{1}{\pi}(t^2-2x^2-2y^2)^{-1/2} & \text{ if } x^2+y^2<\tfrac12 t^2,\\
0 & \text{ otherwise. }
\end{cases}\]
In~\cite{Ch-R} the function~$\Psi_E$ was defined as the conjectured limit of the right hand side of~\eqref{eq:f_E_unif}, using that~$f_0$ can be seen as the fundamental solution to the discrete wave equation. Similarly we define~$\Psi_N(x,y)$,~$\Psi_W(x,y)$, and~$\Psi_S(x,y)$.

\begin{remark}\label{rmk:lor_min}
It was shown in~\cite[Proposition 3.1]{Ch-R}, that the graph $(z, \vartheta(z))$ is precisely equal to the maximal surface~$S_{\diamondsuit}$ with boundary contour~$C_{\diamondsuit}$, where~$z$ and~$\vartheta$ are defined by~\eqref{eq:limit_t}--\eqref{eq:limit_o}.
\end{remark}

\begin{proposition}\label{prop:main_theorem_2}
Let~$\T_n$ and~$\Or_n'$ be the perfect t-embedding and origami map of the reduced Aztec diamond~$A_{n+1}'$, then
\begin{align}\label{eqn:tconv_2}
\T_n(j,k)=&z(j/n,k/n)+o(1), \\ 
\mathcal O_n'(j,k)=&\vartheta(j/n,k/n)+o(1), \label{eqn:oconv_2}
\end{align}
as~$n\to \infty$, where~$z$ and~$\vartheta$ are defined by~\eqref{eq:limit_t} and~\eqref{eq:limit_o}. The convergence is uniform for~$(j/n,k/n)$ on compact subsets of~$\mathcal A\coloneqq\{|x|+|y|<1\}$.
\end{proposition}
\begin{proof}
A computation shows that the integral~\eqref{eq:psi_e} is equal to the limit of the right hand side of \eqref{eq:CEP}, which proves the limits
\begin{align}
p_E(j,k,n)=\Psi_E(j/n,k/n)+o(1) \text{ as } n\to\infty, \label{eq:p_psi_E}\\
p_N(j,k,n)=\Psi_N(j/n,k/n)+o(1) \text{ as } n\to\infty,\label{eq:p_psi_N}\\
p_W(j,k,n)=\Psi_W(j/n,k/n)+o(1) \text{ as } n\to\infty, \label{eq:p_psi_W}\\
p_S(j,k,n)=\Psi_S(j/n,k/n)+o(1) \text{ as } n\to\infty, \label{eq:p_psi_S}
\end{align}
where the convergence is uniform for~$(j/n, k/n)$ on compact subsets of~$\mathcal A$. In addition, note that~$\Psi_E + \Psi_N + \Psi_W + \Psi_S \equiv 1$, see, e.g.,~\cite{cohn-elki-prop-96}. The statement now follows by comparing~$\T_n$ and~$\mathcal O'_n$ with~$z$ and~$\vartheta$, respectively, as expressed in~\eqref{eq:T_f_e_reminder},~\eqref{eq:O_f_e_reminder},~\eqref{eq:limit_t} and~\eqref{eq:limit_o}.
\end{proof}

Recall that the origami map can be seen as a map on $\diamondsuit$. Similarly, $\vartheta$ can be defined as a function of $z \in \diamondsuit$, as we describe in the proof below.

\begin{corollary}\label{Cor:main_lor_min}
We have the uniform convergence on compact subsets of~$\diamondsuit$
$$\Or_n'(z) \rightarrow \vartheta(z)$$
 where the graph of~$\vartheta$ is the maximal surface~$S_{\diamondsuit} \subset \mathbb{R}^{2,1}$ with boundary contour~$C_{\diamondsuit}$.
\end{corollary}

\begin{proof}
The convergence of graphs~$(\T_n, \Or_n')$ to a maximal surface in~$\R^{2,1}$, essentially follows from  Proposition~\ref{prop:main_theorem_2} together with Remark~\ref{rmk:lor_min}. We give details below for completeness. 
It is noted in~\cite{Ch-R} that~$z(x,y)$ diffeomorphically maps the liquid region~$\frac{1}{\sqrt{2}} \mathbb{D} \subset \mathcal{A}$ to the domain~$\diamondsuit$; it follows from this, and from the fact that~$z$ and~$\vartheta$ are constant on each frozen region, that~$\vartheta$ can also be viewed as a function of~$z \in \diamondsuit \subset \mathbb{C}$ which satisfies the property that~$\vartheta(z(x,y)) = \vartheta(x,y)$ for any~$(x,y) \in \mathcal{A}$. The function~$\vartheta$ is real valued, and it is shown in~\cite[Proposition~3.1]{Ch-R} that the graph of~$\vartheta$ is the maximal surface~$S_{\diamondsuit}$. Thus, all that remains is to deduce the (uniform on compact subsets) convergence~$\Or_n'(z) \rightarrow \vartheta(z)$ as functions on the domain~$\diamondsuit \subset \C$ from~\eqref{eqn:tconv_2} and~\eqref{eqn:oconv_2}.

In order to do this, let us suppose we have some compact subset~$K \subset \diamondsuit$. Define~$\mathcal{K}$ as a compact subset of~$\frac{1}{\sqrt{2}} \mathbb{D}$ strictly containing an~$\eps$-neighborhood of the compact subset~$z^{-1}(K)$, for some small enough~$\eps > 0$. Then, for a fixed point~$z_0  \in K \subset \diamondsuit$, we denote by~$(j, k)$ the vertex of~$( A_{n+1}')^*$ such that~$|(j/n,k/n) - z^{-1} (z_0) |$ is minimal. Using this notation, we have
\begin{align*}
\Or_n'(z_0) &= \Or_n'(z(j/n,k/n)) + O(1/n) \\
&= \Or_n'(j, k) + o(1) \\
&= \vartheta(j/n, k/n) + o(1) \\
&= \vartheta( z(j/n, k/n) ) + o(1) \\
&= \vartheta( z_0 ) + o(1) 
\end{align*}
and we must show that the error accumulated in each line can be taken to be uniformly small for~$z_0 \in K$. To do this, we note two facts. First, there exists~$N > 0$ independent of~$z_0 \in K$, such that for all~$n > N$, we will have~$(j/n,k/n) \in \mathcal{K}$. Second, by the previous observation  and~\eqref{eqn:tconv_2} in Proposition~\ref{prop:main_theorem_2}, we have
\begin{equation}\label{eqn:star}
\mathcal{T}_n(j, k) = z(j/n,k/n) + o(1)
\end{equation}
uniformly over~$z_0 \in K$. Now, from the fact that~$\Or_n'$ is~$1$-Lipschitz and~$z$ is smooth, it is clear that the error in the first line is uniformly small. In the second line, the error comes from approximating~$z(j/n,k/n)$ by~$\T_n(j, k)$, and so the error in the second line is uniformly small by~\eqref{eqn:star}. The error in the third line is uniform by the uniform convergence of~$\Or_n'$ to~$\vartheta$ on~$\mathcal{K}$. The fourth line has no additional error by definition, and in the fifth line the additional error is uniformly small because~$z$ is smooth and~$\vartheta$ is Lipschitz with constant~$1$ on the entire domain. This concludes the proof.
\end{proof}

Another direct consequence of the representation of~$\T_n$ and~$\mathcal O_n'$ in terms of edge probabilities, together with Proposition~\ref{prop:edge_prob_convergence} is Corollary~\ref{Cor_main_frozen}, which is stated as Corollary~\ref{Cor_main_frozen_2} below.

\begin{corollary}\label{Cor_main_frozen_2}
As~$n \rightarrow \infty$, the perfect t-embedding and the origami map of each frozen region of the uniformly weighted Aztec diamond maps to a point. More precisely, for all~$j+k+n$ odd 
one has
\begin{align*}
(\T_n(j,k),\mathcal{O}_n'(j,k))&=o(1)+
\begin{cases}
(1,1/\sqrt{2}) &\text{ if }\, (j/n, k/n) \in \mathcal A \cap  \{x^2 + y^2 \geq 1/2\} \cap  \{  x > 1/2 \}\\
(i,-1/\sqrt{2}) &\text{ if }\, (j/n, k/n) \in \mathcal A \cap  \{x^2 + y^2 \geq 1/2\} \cap  \{  y > 1/2 \}\\
(-1,1/\sqrt{2}) &\text{ if }\, (j/n, k/n) \in \mathcal A \cap  \{x^2 + y^2 \geq 1/2\} \cap  \{  x < - 1/2 \}\\
(-i,-1/\sqrt{2}) &\text{ if }\, (j/n, k/n) \in \mathcal A \cap  \{x^2 + y^2 \geq 1/2\} \cap  \{  y < -1/2 \}
\end{cases}
\end{align*}
where the~$o(1)$ error is uniform for~$(j/n, k/n)$ in compact subsets as~$n\to\infty$. 
\end{corollary}

We continue by discussing the effect of the above results on the perfect t-embeddings of the tower graphs. 

Recall Theorem~\ref{thm:T_aztec_tower}, which can be restated as
\begin{align*}
\widetilde{\T}_n(j, k) &= \T_{\frac{1}{3} (1 + j + k + 4 n)}(j, k) \\
\widetilde \Or_{n}(j,k) &=  \Or_{\frac{1}{3} (1 + j + k + 4 n)}(j, k)
\end{align*}
for suitable~$(j,k,n)$, where~$\widetilde \T_n$ and~$\widetilde \Or_n$ are the perfect t-embedding and origami map of the tower graph defined in Section~\ref{sec:tower}. We consider, as we did in the case of the Aztec diamond, another version of the origami map,
\begin{equation*}
\widetilde\Or'_n(j,k)\coloneqq\e^{i\tfrac{\pi}{4}}\left(\widetilde\Or_n(j,k) - \tfrac{1+i}{2}\right).
\end{equation*}
Note that 
$\widetilde \Or'_{n}(j,k) =  \Or'_{\frac{1}{3} (1 + j + k + 4 n)}(j, k).$

The convergence of the perfect t-embedding and the origami map of the Aztec diamond together with Theorem~\ref{thm:T_aztec_tower} implies that the pairs~$(\widetilde\T_n, \widetilde\Or_n')$ converges to the same maximal surface as~$(\T_n,\Or_n')$.

\begin{corollary}\label{cor:tower_coord} 
We have the convergence of tower graph origami maps, uniform on compact subsets of~$\diamondsuit$,$$\widetilde\Or_n'(z) \rightarrow \vartheta(z),$$ where the graph of~$\vartheta$ is the maximal surface~$S_{\diamondsuit}$ with boundary contour~$C_{\diamondsuit}$.
\end{corollary}

\begin{proof}Note that,
by Proposition~\ref{prop:main_theorem_2}, the sequence~$(\T_n(nx,ny), \Or'_n(nx,ny))$ converges uniformly to~$(z(x,y),\vartheta(x,y) )$ on compact subsets of~$|x|+|y|<1$, as~$n\to\infty$. Then Theorem~\ref{thm:T_aztec_tower} implies that
\begin{equation}\label{eq:limit_tower}
\big(\widetilde\T_n(nx,ny), \widetilde\Or'_n(nx,ny)\big)\to\left(z(\tfrac{3x}{4+x+y},\tfrac{3y}{4+x+y}),\vartheta(\tfrac{3x}{4+x+y},\tfrac{3y}{4+x+y}) \right),
\end{equation}
as~$n\to\infty$, and as~$(x,y)$ range over the tower graph rescaled by~$1/n$, which is defined by
\[\begin{cases}
- 1 \,< \, x + y \, < \, 2, \\
x - 2 y < 2 \quad \text{and} \quad y - 2 x < 2, 
\end{cases}\]
the right hand side of~\eqref{eq:limit_tower} range over~$S_\diamondsuit$. Note that as~$(x, y)$ ranges over this set, the points~$\left(\frac{3x}{4+(x+y)}, \frac{3y}{4+(x+y)}\right)$ ranges over the rescaled Aztec diamond domain~$\mathcal A$. 

Since~$\widetilde \T_n$ and~$\widetilde \Or_n'$ converge uniformly on compact subsets of the tower graph, it follows that~$\widetilde \Or_n'$ converges on compact subsets of~$\diamondsuit$. See the proof of Theorem~\ref{thm:main_thm}, where we give this argument for the Aztec diamond, for details.
\end{proof}

We have observed above that the limiting maximal surface for perfect t-embeddings of tower graphs is \emph{the same} as the one for the Aztec diamond,  see Figure~\ref{fig:surface_both}.
Thus, if we equip the tower graph's liquid region with the metric of~$S_{\diamondsuit}$, then its conformal parameterization should be equal to that of the Aztec diamond up to the change of coordinates described in the proof of Corollary~\ref{cor:tower_coord}, which identifies the two liquid regions.

In addition, we know from \cite{Ch-R} that the conformal parameterization of the liquid region obtained from the limiting maximal surface for the Aztec diamond agrees with the Kenyon--Okounkov conformal structure. We now check the following:

\begin{claim}\label{claim:same_structure}
The change of coordinates of Corollary~\ref{cor:tower_coord}
$$(x, y) \mapsto \left(\frac{3x}{4+(x+y)}, \frac{3y}{4+(x+y)}\right)$$
diffeomorphically maps the liquid region of the tower graph to that of the Aztec diamond, and transforms the conformal structure describing the Gaussian fluctuations on the Aztec diamond into the conformal structure describing the Gaussian fluctuations on the tower graph.
\end{claim}

 Due to~\cite{BoutillierLiSHL}, the conformal structure describing the height fluctuations on the tower graph, after switching the coordinate system to~$(x, y)$ coordinates, is given by 
\begin{equation}\label{eqn:tower_cs}
\xi_{\text{tower}}(x, y) =
 \frac{-3 x - 3 y + 
 \sqrt{ 18 x^2 + 18 y^2 - (4 + x + y)^2}}{4 + 4 y - 2 x},
 \end{equation}
 which maps the liquid region of the tower graph to the upper half plane.
   In more detail, to obtain this formula one should solve \cite[Equation (4.3)]{BoutillierLiSHL}
   with~$\mathbf{m}_{\omega}$ equal to the uniform measure on~$[0,1]$, as this is the boundary condition corresponding to tower graphs. Then changing from~$(\chi, \kappa)$ coordinates to~$(x,y)$ coordinates gives the formula for~$\xi_{\text{tower}}$ written above but with~$x$ and~$y$ swapped; however, observe that the limiting height function fluctuations are invariant under this swap because of the invariance of the uniform dimer model on the tower graph under the reflection about~$y = x$. The complex coordinate \eqref{eqn:tower_cs} also can be seen to agree with the Kenyon--Okounkov conformal structure by choosing a fundamental domain for the tower graph and then solving the complex Burgers equation (see \cite{OkounkovKenyon2007Limit} for details on the complex Burgers equation, and see \cite{tower_localstat} for the solution via characteristics in this particular case).

Next, it is known that the conformal structure for the limiting Gaussian free field of the Aztec diamond is induced by the following diffeomorphism of the liquid region to the upper half plane
\begin{equation}\label{eq:conformal_structure_aztec}
\xi_{\text{aztec}}(x, y) =\frac{x + y - \sqrt{-1 + 2 x^2 + 2 y^2}}{-1 + x - y},
\end{equation}
see \cite{CKBAztec, bufetov2016fluctuations}. Note that if we map \eqref{eq:conformal_structure_aztec} to the unit disc, under the conformal map
\begin{equation}
z\mapsto \frac{1-\i}{\sqrt{2}}\frac{z-\i}{z+\i},
\end{equation}
we obtain the conformal parametrization of the surface~$S_\diamondsuit$ given in \cite{Ch-R}. We can directly check that
   \begin{align}\label{eqn:coord_change}
  \xi_{\text{aztec}}\left(\frac{3x}{4+(x+y)}, \frac{3y}{4+(x+y)}\right) &=  \frac{-3 x - 3 y + \sqrt{
 18 x^2 + 18 y^2 - (4 + x + y)^2}}{4 + 4 y - 2 x} \\
 &= \xi_{\text{tower}}(x, y). \notag
   \end{align}
  The calculation above, along with the fact that the liquid region in either domain is the region where~$\xi_\text{aztec/tower}$ satisfies~$\text{Im}(\xi_\text{aztec/tower}) > 0$, verifies Claim~\ref{claim:same_structure}.

  \begin{remark} The claim has two consequences. First, it implies that in the case of tower graphs, the Kenyon--Okounkov conformal structure and the one obtained from the maximal surface agree; they are both induced by~$\xi_{\text{aztec}}\left(\frac{3x}{4+(x+y)}, \frac{3y}{4+(x+y)}\right)$ (c.f. Theorem~\ref{thm:Aztec_assumptions} and the proof of Corollary~\ref{cor:tower_coord}). Second, under the change of coordinates in \eqref{eqn:coord_change}, the complex structures on tower graphs and Aztec diamonds, and thus their limiting height fluctuation fields, are identified. The perfect t-embedding machinery, together with Theorem~\ref{thm:T_aztec_tower}, is what reveals this non-trivial relationship.
  \end{remark}

\subsection{$\LipKd$ and~$\ExpFat$ under structural rigidity of~$\mathcal{T}_n$}\label{sec:lip_expfat}
In this section we assume that~$\T_n$ is a sequence of perfect t-embeddings of weighted bipartite planar 
graphs~$\G_n, n\to\infty$ (not necessary Aztec diamonds or tower graphs). Assume that~$\operatorname{deg}(v)<D$ for each vertex~$v$ of the dual graph~$\G^*_n$, where~$D>0$ is a constant that does not depend on~$n$.   In addition, we assume that the sequence of discrete domains~$\Omega_{\T_n}$ converges to a simply-connected convex bounded domain~$\Omega$ (in the Hausdorff sense). And as before, we denote the origami map of~$\T_n$ by~$\Or_n$. We also fix a sequence of positive reals~$\mu_n \rightarrow 0$.

\begin{assumption}[Rigidity condition]\label{assumption:bounded} 
Given a compact set~$\mathcal K\subset\Omega$, there exist positive constants~$N_\mathcal{K}, C_\mathcal{K}$ and~$\varepsilon_\mathcal{K}$ which only depend on~$\mathcal{K}$,
such that for all pairs of adjacent vertices~$v, v'$ of the dual graph~$\G_n^*$
such that both~$\mathcal{T}_n(v)$ and~$\mathcal{T}_n(v')$ are contained in~$\mathcal{K}$ we have
\[\frac{\mu_n}{C_\mathcal{K}}  \leq |\mathcal{T}_n(v') - \mathcal{T}_n(v) | \leq \mu_n C_\mathcal{K}\]
for all~$n>N_\mathcal{K}$. 
In addition the angles of the faces of the perfect t-embedding inside~$\mathcal{K}$ are contained in~$(\varepsilon_\mathcal{K}, \pi-\varepsilon_\mathcal{K})$ for all~$n >N_\mathcal{K}$.
\end{assumption}

\begin{proposition}\label{prop:lip}
Suppose the above assumption holds for a compact set~$\mathcal{K}\subset\Omega$. Then for all~$n>N_\mathcal{K}$ there exist constants~$\kappa=\kappa(\mathcal{K}) \in (0,1)$ and~$C'=C'_\mathcal{K} > 0$ such that for~$\delta = C' \mu_n$ the condition~$\LipKd$ holds for the perfect t-embedding~$\T_n$, i.e. 
\begin{equation*}\label{eq:lip_5_1}
|\mathcal{O}_n(z) - \mathcal{O}_n(z')| < \kappa |z - z'|
\end{equation*}
for each~$z, z' \in \mathcal{K}$ with~$|z - z'| > \delta$.
\end{proposition}

\begin{figure}
\centering
\includegraphics[scale=0.8]{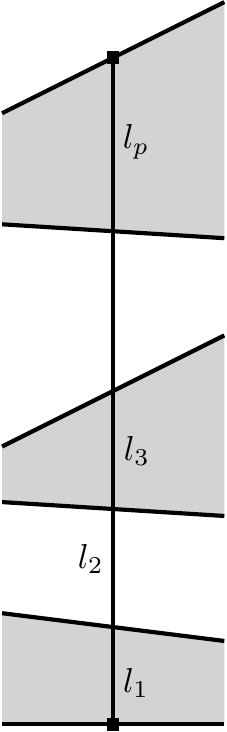}
\caption{The line segment~$\ell$ divided to smaller segments~$\ell_i$ by edges of the t-embedding in the proof of Proposition~\ref{prop:lip}.}
\label{fig:l_i}
\end{figure}

\begin{proof}
We want to show that we can choose~$C$ and~$\kappa$ such that for any line segment~$\ell$ connecting two points~$z, z' \in \mathcal{K}$ with~$|z - z'| > C \mu_n$, we have 
\begin{align}\label{eq:Lip_l}
|\mathcal{O}_n(z) - \mathcal{O}_n(z')| <  \kappa\cdot\text{length}(\ell).
\end{align}
Note that it is enough to show~\eqref{eq:Lip_l} for all segments~$\ell$ contained in~$\mathcal{K}$ that do not contain any vertices of~$\mathcal{T}_n$, and that they start and end on some edges of~$\mathcal{T}_n$. 
Indeed, recall that~$d\Or_n$ is linear on each face form and~$\mathcal{O}_n(z) - \mathcal{O}_n(z')=\int_{\ell} \; d \mathcal{O}_n,$ therefore we can get~\eqref{eq:Lip_l} for any other~$\ell$ by linearity.

For a fixed~$n > N_\mathcal{K}$ consider a line segment~$\ell$ contained in~$\mathcal{K}$ which does not contain any vertices of~$\mathcal{T}_n$, and that it starts and ends on an edge of~$\mathcal{T}_n$. 
Then, if~$f_1,f_2, \dots, f_p$ are the faces crossed by~$\ell$ (in consecutive order), we define~$\ell_i$ to be the segment of~$\ell$ contained in~$f_i$, see Figure~\ref{fig:l_i}. 

\smallskip

\textbf{Claim 1}: There exists~$C_{\operatorname{long}} > 0$ (independent of the line segment~$\ell$) such that there are no more than~$D$ segments~$\ell_i$ with~$\text{length}(\ell_i) \leq \mu_n/C_{\operatorname{long}}$ in a row. 

\begin{proof}[Proof of Claim 1]
Choose~$C_{\operatorname{long}}\gg C_\mathcal{K}$. If~$\text{length}(\ell_i) < \mu_n/C_{\operatorname{long}}$, there must be a vertex~$v_i$ of the t-embedding such that~$\text{dist}(\ell_i, v_i) < c \mu_n/C_{\operatorname{long}}$, for some constant~$c = c(\varepsilon_\mathcal{K})\geq 1$ since all faces of the t-embedding are convex polygons with sides of length at least~$\frac{1}{C_\mathcal{K}} \mu_n$ and angles bounded from~$0$ and~$\pi$. 

As a result, if~$\ell_i, \ell_{i+1}, \dots, \ell_{i+k}$ all have length less than~$\mu_n/C_{\operatorname{long}}$, then we get a sequence of vertices~$v_m$ with~$|v_m - v_{m+1}| < 4 c \mu_n/C_{\operatorname{long}}$. Choose~$C_{\operatorname{long}}$ large enough so that this is impossible unless all of the~$v_m$ are equal. So there is a vertex~$v = v_i = v_{i+1} = \cdots = v_{i + k}$ which is a vertex of each of the faces~$f_i, f_{i+1}, \dots, f_{i+k}$. Thus~$k+1$ must be bounded by the degree of the graph, which is~$D$, and thus we have shown the claim.
\end{proof}

Using the~$C_{\operatorname{long}}$ from above, we call a segment~$\ell_i$ \emph{long} if~$\text{length}(\ell_i) > \mu_n/ C_{\operatorname{long}}$. 


\smallskip

\textbf{Claim 2}:  It suffices to show that there are constants~$C', \kappa' > 0$ (independent of~$\ell$) such that for all~$z,z'\in \mathcal{K}$ with~$|z - z'| > C' \mu_n$ one has
\begin{equation}\label{eqn:longLip}
\Big| \sum_{\ell_i  \text{ long}} \int_{\ell_i} \; d \mathcal{O} \Big| 
<
 \kappa' \Big| \sum_{\ell_i  \text{ long}} \int_{\ell_i} \; d z \Big|.
\end{equation}

\begin{proof}[Proof of Claim 2]
Set 
\[L_\mathcal{O} := \Big| \sum_{\ell_i  \text{ long}} \int_{\ell_i} \; d \mathcal{O} \Big|, 
\quad
S_{\mathcal{O}} := \Big| \sum_{\ell_i  \text{ not long}} \int_{\ell_i} \; d \mathcal{O} \Big|,\] 
recall that~$d\T$ is simply~$d z$ and similarly define 
\[L_\mathcal{T} := \Big| \sum_{\ell_i  \text{ long}} \int_{\ell_i} \; d z \Big|,
\quad
 S_\mathcal{T} := \Big| \sum_{\ell_i  \text{ not long}} \int_{\ell_i} \; d z \Big|.\] 
 Note that~$| \sum_{\ell_i } \int_{\ell_i} \; d z |= \sum_{\ell_i } |\int_{\ell_i} \; d z|$.
 
 Assuming we have~$C', \kappa'$ such that equation \eqref{eqn:longLip} holds,  we are looking for constants~$C, \kappa$ such that for any segment~$\ell\in \mathcal{K}$ with~$\text{length}(\ell) > C \mu_n$
\[L_\mathcal{O} + S_{\mathcal{O}} < \kappa (L_{\mathcal{T}} + S_{\mathcal{T}}).\]

Choose~$C$ large enough so that~$|z - z'| > C \mu_n$ implies that we must have at least one long segment in~$\ell$. The statement of Claim~$1$ implies that there exists~$c=c(\mathcal{K})>0$ (independent of~$\ell$) such that 
\[ S_{\mathcal{T}} \,<\, c L_{\mathcal{T}}.\] Therefore since~$S_{\mathcal{O}} \leq S_{\mathcal{T}}$, we get
\begin{align*}
\frac{L_\mathcal{O} + S_{\mathcal{O}}}{L_{\mathcal{T}} + S_{\mathcal{T}}} \,<\, 
\frac{L_\mathcal{O} + c L_{\mathcal{T}} }{(c+1) L_{\mathcal{T}}} 
\,<\,  \left( \frac{1}{c+1} \kappa' + \frac{c}{c+1} \right),
\end{align*}
so if we set~$\kappa = \frac{1}{c+1} \kappa' + \frac{c}{c+1}  < 1$, then the Lipschitz condition holds with~$C, \kappa$.
\end{proof}

\textbf{Claim 3}: There exists~$C', \kappa'$, independent of~$\ell$, such that \eqref{eqn:longLip} holds.

\begin{proof}[Proof of Claim 3]
By choosing~$C$ large enough we can make~$\ell$ to contain at least three long segments. 
It suffices to show that, if~$\text{length}(\ell) \geq C \mu_n$ for appropriately chosen~$C$, there exists~$\kappa'$ such that for any three consecutive long edges~$\ell_{i_{j-1}}, \ell_{i_{j}}, \ell_{i_{j+1}}$, the following holds
\begin{equation}\label{eqn:3term}
\Big| \sum_{s = -1}^1 \int_{\ell_{i_{j+s}}} \; d \mathcal{O} \Big| < \kappa' \sum_{s = -1}^1 \text{length}(\ell_{i_{j+s}}). 
\end{equation}
Assume that~$\ell_{i_{j}}$ is contained in a white face~$w$. WLOG we also assume that~$\ell$ is a vertical segment and~$\eta_w=1.$

\smallskip

\begin{figure}
\centering
\includegraphics[scale=0.8]{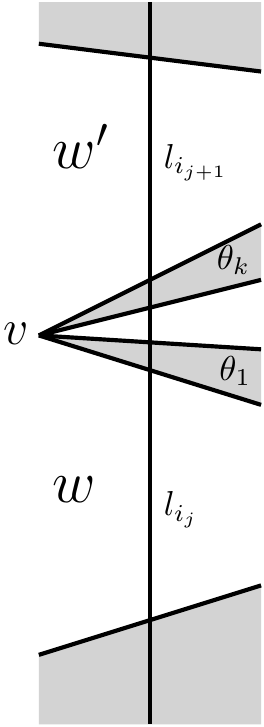}
$\qquad\qquad\qquad$
\includegraphics[scale=0.8]{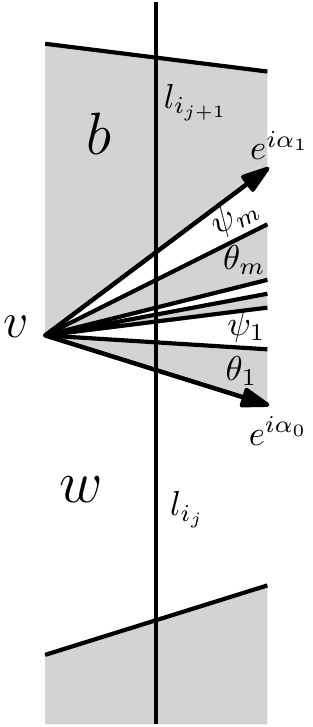}
$\qquad\qquad\qquad$
\includegraphics[scale=0.8]{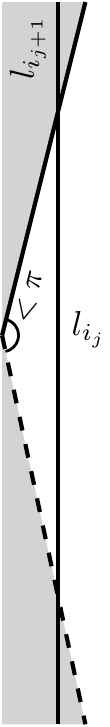}
\caption{Schematic representation of cases of location of consecutive long edges~$\ell_{i_{j}}$ and~$\ell_{i_{j+1}}$. {\bf Left:} Case~$1$. {\bf Middle:} Case~$2$. {\bf Right:} Case~$3$: the line shown in dashed is not possible.}
\label{fig:l_i_cases}
\end{figure}

{\bf Case 1:} Assume first, that the next long segment~$\ell_{i_{j+1}}$ is contained in a white face~$w'$. Then 
\begin{align*}
\Big|\int_{\ell_{i_{j}}} \; d \mathcal{O} + \int_{\ell_{i_{j+1}}} \; d \mathcal{O}\Big| 
&= \Big| \int_{\ell_{i_{j}}} dz + \int_{\ell_{i_{j+1}}}  {\eta_{w'}}^2 d{z}\Big| \\
&= \left | i |\ell_{i}| + i {\eta_{w'}}^2 |{\ell_{i+1}}|  \right| .
\end{align*}
Note that two white faces are not adjacent, therefore there are some not long segments between~$\ell_{i_{j}}$ and~$\ell_{i_{j+1}}$, see Figure~\ref{fig:l_i_cases}. Due to the proof of Claim~$1$ all these short segments belong to faces adjacent to the same vertex~$v$. Let~$\theta_1, \ldots, \theta_k$ be angles of the black faces intersecting~$\ell$ adjacent to the vertex~$v$ (as shown in Figure~\ref{fig:l_i_cases}). Then
\[\eta_{w'}=\exp (-2i(\theta_1 + \ldots + \theta_k)).\]
Recall that all angles~$\theta_s$ are bounded from~$0$ by Assumption~\ref{assumption:bounded}, therefore
\[0 < \varepsilon_K \leq \theta_1 + \ldots + \theta_k \leq \pi-\varepsilon_K <\pi\]
since the sum of  angles of all black faces around~$v$ is~$\pi$. This implies that~$|1-\eta^2_{w'}|$ is bounded from zero, and therefore there exists~$\kappa'_{0}$ such that 
\[\Big|\int_{\ell_{i_{j}}} \; d \mathcal{O} + \int_{\ell_{i_{j+1}}} \; d \mathcal{O}\Big| <
\kappa'_{0}
\Big|\int_{\ell_{i_{j}}} \; d z + \int_{\ell_{i_{j+1}}} \; d z\Big|.
\] 
Hence, since for long~$\ell_j$ we have~$\mu_n \lesssim \text{length}(\ell_{j}) \lesssim \mu_n$, there exists~$\kappa'$ such that~\eqref{eqn:3term} holds.

\smallskip

Now assume that the next long segment~$\ell_{i_{j+1}}$ is contained in a black face~$b$. There are two cases when that happens.

\smallskip

{\bf Case~$2$.} Let us first consider the case when the faces~$w$ and~$b$ are not adjacent. This case is similar to the Case~$1$. Indeed, note that in this case 
\begin{align*}
\Big|\int_{\ell_{i_{j}}} \; d \mathcal{O} + \int_{\ell_{i_{j+1}}} \; d \mathcal{O}\Big| 
&= \Big| \int_{\ell_{i_{j}}} dz + \int_{\ell_{i_{j+1}}}  \overline\eta_{b}^2 d\overline{z}\Big| \\
&= \left | i |\ell_{i}| - i \overline\eta_{b}^2 |{\ell_{i+1}}|  \right| .
\end{align*}
Therefore, similar to the first case, it suffices to show that~$|1+\overline\eta^2_{b}|$ is bounded from zero.

Let~$v, \theta_1, \ldots, \theta_m$ be as in Case~$1$. Denote by~$\psi_1, \ldots, \psi_m$ the angles of the white faces intersecting~$\ell$ adjacent to the vertex~$v$. Let also~$\e^{i\alpha_0}$ and~$\e^{i\alpha_1}$ be the directions of the edges of the t-embedding adjacent to~$v$ and the faces~$w$ and~$b$ correspondingly, see Figure~\ref{fig:l_i_cases}. Then we have
\[\overline{\eta}_b^2 = \exp(2 \i (\alpha_1 - \theta_1 - \cdots - \theta_m)).\]
Note that 
\begin{align*}
\frac{\pi}{2} - \sum_{j=1}^m \theta_j  
\,\,\geq \,\,
\alpha_1 -   \sum_{j=1}^m \theta_j  
= \alpha_0 +   \sum_{j=1}^m \psi_j  
\,\, \geq \,\,
-\frac{\pi}{2} + \sum_{j=1}^m \psi_j,
\end{align*}
therefore~$|1+\overline\eta^2_{b}|$ is bounded from zero.

\smallskip

{\bf{Case~$3$:}} Finally, assume that the faces~$w$ and~$b$ are adjacent. Then it can happen that the edge of the t-embedding between~$w$ and~$b$ is almost parallel to~$\ell$. This would imply that 
\[\Big|\int_{\ell_{i_{j}}} \; d \mathcal{O} + \int_{\ell_{i_{j+1}}} \; d \mathcal{O}\Big| 	\approx
\Big|\int_{\ell_{i_{j}}} \; d z + \int_{\ell_{i_{j+1}}} \; d z\Big|.
\] Recall, that we want to show that~\eqref{eqn:3term} holds for the sum along three consecutive long edges. Therefore, it suffices to show that the `bad scenario' of Case~$3$ cannot happen twice in a row. Note that at least one of these edges has to be away from the direction of the segment~$\ell$ since all angles of the t-embedding are bounded from~$0$ and~$\pi$.
This finishes the proof of Claim~$3$.
\end{proof}

The second and third claim together imply the proposition, so we are done.
\end{proof}

\begin{remark}
The proof above does not rely on the fact that the t-embeddings are perfect t-embeddings. In other words, Proposition~\ref{prop:lip} holds for any sequence of t-embeddings satisfying the rigidity condition~\ref{assumption:bounded}. 
\end{remark}

\begin{corollary}\label{cor:lip_under_rigiditi}
Suppose Assumption~\ref{assumption:bounded} holds for all compact sets~$\mathcal{K}\subset \Omega$ with~$\mu_n=\tfrac{1}{n}$. Then~$\T_{n}$ satisfies Assumption~\ref{assumption:Lip}, with the sequence~$\delta=\delta_{n}=\tfrac{\log n}{n}$.
\end{corollary}

Next, we prove that under Assumption~\ref{assumption:bounded}, the~$\ExpFat$ condition holds as well.

\begin{proposition}\label{prop:exp_fat_under_rigiditi}
In the setup of Corollary~\ref{cor:lip_under_rigiditi} the sequence of perfect t-embeddings~$\T_n$ satisfies assumption~$\ExpFat$, Assumption~\ref{assumption:Exp_fat}, with~$\delta=\delta_{n}=\tfrac{\log n}{n}$.
\end{proposition}

\begin{proof}
We claim that~$\ExpFat$ holds with an arbitrary choice of the splittings of black and white faces of the sequence of perfect t-embeddings~$\mathcal{T}_n$. Indeed, the non-degeneracy of the angles and the assumption about edge lengths of the embedding implies that for large enough~$C$, all edges~$e$ of faces in~$\mathcal{K}$ of a splitting of~$\mathcal{T}_n$ are also bounded as
\[ \frac{1}{C}n^{-1}  \leq \text{length}(e) \leq C n^{-1}\]
for all~$n$ large enough. One may also see that for any splitting, the angles of all triangles will also be uniformly bounded away from~$0$ and~$\pi$ in~$\mathcal{K}$, for all~$n$ large enough.

Thus for~$\delta=\delta_{n}=\tfrac{\log n}{n}$ it is clear that $\ExpFat$ holds, since for large enough~$n$, in each triangle one can inscribe a circle of radius~$\geq c n^{-1}$, for some~$c$ small enough which depends only on the bound on lengths and angles that we have in~$\mathcal{K}$. Hence, we may, for instance take~$\delta'=\tfrac{1}{\log n}$.
\end{proof}

\begin{remark}\label{rem:main_thm_1}
We will show in Section~\ref{sec:complicated_stuff} that perfect t-embeddings~$\T_n$ of the uniformly weighted Aztec diamonds~$A_{n+1}$ satisfy the rigidity condition with~$\mu_n=\tfrac{1}{n}$.  Together with Corollary~\ref{cor:lip_under_rigiditi} and Proposition~\ref{prop:exp_fat_under_rigiditi} this would conclude hypothesis~\ref{item:c} of Theorem~\ref{thm:Aztec_assumptions}.
\end{remark}

\subsection{Proof of Theorem~\ref{thm:Aztec_assumptions}}\label{sec:main_thm_proof}

\old{
Note that the domain~$\Omega_{\T_n}$ covered by~$\T_n$ for each~$n$ is equal to~$\diamondsuit \coloneqq \{z = x+ \i y: |x | +| y| < 1\}$.  Let us denote by~$\partial \diamondsuit$ the boundary of the domain~$\diamondsuit$. Then
 the graph of~$\Or_{n}'$ restricted to~$\partial \diamondsuit$ is the non-planar quadrilateral, which we denote by~$C_{\diamondsuit}$, 
 in~$\mathbb{C} \times \mathbb{R}$ with vertices
\[\left(1, \tfrac{\sqrt{2}}{2}\right), \quad
\left(i, -\tfrac{\sqrt{2}}{2}\right), \quad
\left(-1, \tfrac{\sqrt{2}}{2}\right) \quad
\text{ and } \quad
\left(-i, -\tfrac{\sqrt{2}}{2}\right).\]
We define~$S_{\diamondsuit}$ as the Lorentz-minimal surface with boundary contour~$C_{\diamondsuit}$. This is the unique space-like surface in~$\mathbb{R}^{2,1}$ with this boundary contour which has vanishing mean curvature.}

In this section we give the precise statement of the main theorem, Theorem~\ref{thm:Aztec_assumptions}, and prove it assuming  Lemma~\ref{lem:bdd} holds.


 Let~$\diamondsuit$, $S_{\diamondsuit}$ and~$C_{\diamondsuit}$  be as defined in Section~\ref{sec:T_n_conv}. 

\begin{theorem}\label{thm:main_thm}
Let~$\T_n$ be the sequence of perfect t-embeddings of the reduced uniformly weighted Aztec diamonds~$A_{n+1}'$ defined in Section~\ref{sec:aztec}, with corresponding origami maps~$\Or_n'$ given by~\eqref{def:Or_prime}. The hypotheses~\ref{item:a},~\ref{item:b}, and~\ref{item:c} of Theorem~\ref{thm:CLR2thm} hold for the sequence~$\T_n$, with the choice~$\delta_n = \tfrac{\log n}{n}$. 
\end{theorem}

\begin{proof}
Hypothesis~\ref{item:a} of Theorem~\ref{thm:CLR2thm} is immediate 
since the domain~$\Omega_{\T_n}$ covered by~$\T_n$ for each~$n$ is equal to~$\diamondsuit$, 
 and hence  the limiting domain~$\Omega = \diamondsuit$.

Hypothesis~\ref{item:b} is the content of Corollary~\ref{Cor:main_lor_min}.

Finally, hypothesis~\ref{item:c}, i.e. the validity of the assumptions~$\ExpFat$ and~$\LipKd$ on compact subsets with~$\delta_n = \frac{\log n}{n}$, follows from the results of Sections~\ref{sec:lip_expfat} and~\ref{sec:complicated_stuff}. More precisely, in Section~\ref{sec:lip_expfat} we proved that any sequence of t-embeddings satisfying a certain~\emph{rigidity condition} satisfies hypothesis~\ref{item:c}; see Corollary~\ref{cor:lip_under_rigiditi} and Proposition~\ref{prop:exp_fat_under_rigiditi}.  It follows from Lemma~\ref{lem:bdd} that the perfect t-embeddings of the Aztec diamond satisfy this condition; for its proof see Lemmas~\ref{lem:bound_edges} and~\ref{lem:bound_angles}.
\end{proof}

\section{Validation of the rigidity condition}\label{sec:complicated_stuff}

In this section we prove that the rigidity condition, Assumption~\ref{assumption:bounded}, holds for the perfect t-embeddings~$\T_n$ of the uniform Aztec diamond, or more precisely, of the reduced Aztec diamond~$\mathcal G_n=A_{n+1}'$. 

Recall that~$\diamondsuit\coloneqq\{|x|+|y|<1\}$ is the image of the limit of the perfect t-embedding. Recall also that~$\mathcal A$ and~$\DAz$ are the rescaled Aztec diamond and the circle inscribed in~$\mathcal A$, see Section~\ref{sec:T_n_conv}. Let~$j,k,n,\iota,\iota'\in \ZZ$ with~$|\iota|+|\iota'|=1$. We denote the edges in the dual graph~$(A_{n+1}')^*$ by~$e_{\iota,\iota'}(j,k)$, where~$e_{\iota,\iota'}(j,k)$ is the edge between~$(j,k)$ and~$(j+\iota,k+\iota')$. We suppress the~$n$ dependence in the notation of the edges. We orient the edges such that they have a white vertex to the left. Recall that
\begin{equation}\label{eq:t-embedding_edge}
d \T_n (e_{\iota,\iota'}(j,k))=\pm\left(\T_n(j+\iota,k+\iota')-\T_n(j,k)\right),
\end{equation}
where the sign depends on the orientation of the edge. The sign will not play any role, so we will not specify it more explicitly. 

We prove Assumption~\ref{assumption:bounded} by first reformulate it in terms of compact subsets of~$\DAz$ instead of~$\diamondsuit$. Then we use a double integral formula for~$\T_n$, which is derived using Corollary~\ref{cor:p_E_and_f_E} and a double integral formula for the inverse Kasteleyn matrix, to prove the assumption.

As in Assumption~\ref{assumption:bounded},~$\mu_n$ denotes a sequence converging to zero as~$n\to \infty$. 
\begin{assumption}\label{assumption:bounded_aztec} 
Given a compact set~$\mathcal K\subset\DAz$, there exist positive constants~$N_\mathcal{K}, C_\mathcal{K}$ and~$\varepsilon_\mathcal{K}$ which only depend on~$\mathcal{K}$, such that for all edges~$e_{\iota,\iota'}(j,k)\subset \mathcal K$,
\[\frac{\mu_n}{C_\mathcal{K}}  \leq |d\T_n(e_{\iota,\iota'}(j,k))| \leq \mu_n C_\mathcal{K}\]
for all~$n>N_\mathcal{K}$. 
In addition, for~$n >N_\mathcal{K}$ and~$e_{\iota,\iota'}(j,k),e_{\tilde\iota,\tilde \iota'}(j,k)\subset \mathcal K$ the angle between two adjacent edges~$d\T_n(e_{\iota,\iota'}(j,k))$ and~$d\T_n(e_{\tilde\iota,\tilde \iota'}(j,k))$, which are adjacent to a common face, is contained in~$(\varepsilon_\mathcal{K}, \pi-\varepsilon_\mathcal{K})$. These angles are precisely the angles of the faces of the image of the perfect t-embedding.
\end{assumption}

\begin{lemma}\label{lem:compact_pullback}
Assumption~\ref{assumption:bounded_aztec} implies Assumption~\ref{assumption:bounded}.
\end{lemma}
\begin{proof}
It is enough to show that for each compact set~$\mathcal K_\diamondsuit \subset \diamondsuit$, there exists a compact set~${\mathcal K\subset \DAz}$ and~$N>0$, such that if~$n>N$ and~$\T_n(j,k)\in \mathcal K_\diamondsuit$, then~$\frac{1}{n+1}(j,k)\in \frac{1}{n+1}(A_{n+1}')^*\cap \mathcal K$.

Let~$\Gamma$ be a simple loop, which contains~$\mathcal K_\diamondsuit$ in its interior, and which has a distance~$\eps>0$ from~$\partial \diamondsuit$ and ~$\mathcal K_\diamondsuit$. Let~$\gamma=z^{-1}(\Gamma)\subset \DAz$ be the preimage of~$\Gamma$ under~$z$. Since~$z$ is an embedding,~$\gamma$ is also a simple loop. If we take~$\mathcal K$ as the closure of the interior of~$\gamma$, it has the above described properties.

Indeed, since~$z$ is continuous and~$\T_n$ converges to~$z$ uniformly on compact subsets, there is an~$N>0$ such that for all~$n\geq N$ there exists a simple loop~$\gamma_n$ in the dual graph~$(A_{n+1}')^*$ which approximates~$\gamma$ from within, such that~$\dist(\Gamma,\T_n(\gamma_n))<\eps/2$. Since~$\dist(\Gamma,\mathcal K_\diamondsuit)>\eps$ we get that~$\dist(\T_n(\gamma_n),\mathcal K_\diamondsuit)>\eps/2$. In particular~$\T_n(\gamma_n)\cap\mathcal K_\diamondsuit=\emptyset$. The fact that~$\T_n$ is an embedding implies that~$\T_n(\gamma_n)$ also is a simple loop, and that~$\T_n$ maps the interior (resp. exterior) of~$\gamma_n$ to the interior (resp. exterior) of~$\T_n(\gamma_n)$. We conclude that~$\mathcal K_\diamondsuit$ lies in the interior of~$\T_n(\gamma_n)$, and, hence, if~$\T_n(j,k)\in \mathcal K_\diamondsuit$, then~$\frac{1}{n+1}(j,k)\in \frac{1}{n+1}(A_{n+1}')^*\cap \mathcal K$.
\end{proof}

\subsection{Inverse Kasteleyn matrix and exact formulas for the perfect t-embedding}
From now on, and for the rest of this paper, we will use~$z$ as a variable and not as the limit of the perfect t-embedding.

We define the Kasteleyn matrix~$K$, see the discussion in Section~\ref{sec:definitions}, for the Aztec diamond~$A_n$. Recall that the white vertices are points of the form~$(j\pm \frac{1}{2},k\mp\frac{1}{2})$, and black vertices are points of the form~$(j\pm \frac{1}{2},k\pm\frac{1}{2})$, for~$(j,k)\in \ZZ^2$ and~$j+k+n$ is odd. We use the notation
\begin{equation*}
K(b,w)=K_{(j\pm \frac{1}{2},k\pm\frac{1}{2}),(j\pm \frac{1}{2},k\mp\frac{1}{2})},
\end{equation*}
if~$b=(j\pm \frac{1}{2},k\pm\frac{1}{2})$ and~$w=(j\pm \frac{1}{2},k\mp\frac{1}{2})$. We chose the Kasteleyn signs so that, for~$j_i+k_i+n$ odd,~$i=1,2$, the Kasteleyn matrix is given by
\begin{equation}\label{eq:kasteleyn}
K_{(j_1+\frac{1}{2},k_1+\frac{1}{2}),(j_2+\frac{1}{2},k_2-\frac{1}{2})}
= 
\begin{cases}
(-1)^{k_1+n+1}, & (j_2,k_2)=(j_1,k_1), \\
(-1)^{k_1+n+1}\i, & (j_2,k_2)=(j_1-1,k_1+1), \\
(-1)^{k_1+n}, & (j_2,k_2)=(j_1,k_1+2), \\
(-1)^{k_1+n}\i, & (j_2,k_2)=(j_1+1,k_1+1), \\
0, & \text{otherwise}. 
\end{cases}
\end{equation}
See Figure~\ref{fig:kasteleyn_sign}.

The local statistics, and in particular the edge probabilities, are expressible in terms of the inverse Kasteleyn matrix~$K^{-1}$. Indeed, if~$j+k+n$ is odd, then
\begin{equation}\label{eq:edge_prob_kasteleyn}
p_E(j,k,n)=K_{(j+\frac{1}{2},k+\frac{1}{2}),(j+\frac{1}{2},k-\frac{1}{2})}K^{-1}_{(j+\frac{1}{2},k-\frac{1}{2}),(j+\frac{1}{2},k+\frac{1}{2})}. 
\end{equation}
In a similar way we may express~$p_W$,~$p_S$ and~$p_N$ in terms of~$K$ and~$K^{-1}$. 

The inverse Kasteleyn matrix for the uniform Aztec diamond was first given by Helfgott~\cite{H19}, and later, for the Aztec diamond with a bias, by Chhita--Johansson--Young~\cite{CKBAztec}. We use the formula given in the latter reference. The inverse Kasteleyn matrix is given by
\begin{multline}\label{eq:inverse_kasteleyn}
K^{-1}_{(j_1+\frac{1}{2},k_1-\frac{1}{2}),(j_2+\frac{1}{2},k_2+\frac{1}{2})} 
\\ 
=
\begin{cases}
f_1((j_1, k_1), (j_2, k_2)), & j_1+k_1<j_2+k_2+2, \\
f_1((j_1, k_1), (j_2, k_2)) - f_2((j_1, k_1), (j_2, k_2)), & \text{otherwise}, 
\end{cases}
\end{multline}
where
\begin{align*}
f_1((j_1, k_1), (j_2, k_2)) &= \frac{(-1)^{\frac{1}{2}(k_1+k_2+2n)}}{(2 \pi \i)^2} \int_{\mathcal{E}_2} \int_{\mathcal{E}_1} \frac{w^{\frac{1}{2}(j_2+k_2+n+1)} }{z^{\frac{1}{2}(j_1+k_1+n+1)}} \\
&\times \frac{(1+z)^{\frac{1}{2}(k_1-j_1+n-1)} (z-1)^{\frac{1}{2}(j_1-k_1+n+1)}}{(1+w)^{\frac{1}{2}(k_2-j_2+n+1)}(w-1)^{\frac{1}{2}(j_2-k_2+n+1)}} \frac{\d z \d w}{w-z}, \\
&\vspace{1pt}\\
f_2((j_1, k_1), (j_2, k_2)) &= \frac{(-1)^{\frac{1}{2}(k_1+k_2+2n)}}{2 \pi \i}  \int_{\mathcal{E}_1} 
 \frac{(1+z)^{{\frac{1}{2}(j_2+k_2-j_1-k_1)}} z^{{\frac{1}{2}(k_2-j_2-k_1+j_1)}}}{(z+2)^{{\frac{1}{2}(k_2-j_2-k_1+j_1+2)}}}\d z,
\end{align*}
and~$\mathcal{E}_1$ is a small circle around~$0$ and~$\mathcal{E}_2$ is a small circle around~$1$, both with radius strictly less than~$1/2$ and oriented in positive direction, see Figure~\ref{fig:steepest_descent_curves}. The coordinate change between the formula given here and the one given in \cite[Theorem 2.3]{CKBAztec} is~$j_1=\frac{1}{2}(x_1-x_2)$,~$k_1=\frac{1}{2}(x_1+x_2-2n)$, where~$(x_1,x_2)$ are the coordinates in \cite{CKBAztec}, and similarly for~$(j_2,k_2)$. The parameter~$a$ in the same paper is set to~$1$ here.

\begin{figure}
\begin{center}
\begin{tikzpicture}[scale=.6]

\draw (-3,1)--(3,1);
\draw (-3,-1)--(3,-1);
\draw (-1,-3)--(-1,3);
\draw (1,-3)--(1,3);

\draw (-1,1) node[circle,draw=black,fill=white,inner sep=2pt]{};
\draw (1,-1) node[circle,draw=black,fill=white,inner sep=2pt]{};
\draw (-1,-1) node[circle,draw=black,fill=black,inner sep=2pt]{};
\draw (1,1) node[circle,draw=black,fill=black,inner sep=2pt]{};

 \draw (0,0) node {$(j,k)$};

 \draw (-2,1) node[above] {$-\i$};
 \draw (-2,-1) node[below] {$-\i$};
 \draw (2,1) node[above] {$-\i$};
 \draw (2,-1) node[below] {$-\i$};

 \draw (-1,2) node[left] {$-1$};
 \draw (-1,-2) node[left] {$-1$};
 \draw (1,2) node[right] {$-1$};
 \draw (1,-2) node[right] {$-1$};

 \draw (0,1) node[above] {$\i$};
 \draw (0,-1) node[below] {$\i$};
 \draw (-1,0) node[left] {$1$};
 \draw (1,0) node[right] {$1$};
\end{tikzpicture}
\end{center}
\caption{The Kasteleyn signs used in the definition of the Kasteleyn matrix. Note that the alternating product of the complex edge weights around each face is~$-1$. Here~$j+k+n$ and~$k+n$ are odd. \label{fig:kasteleyn_sign}}
\end{figure}
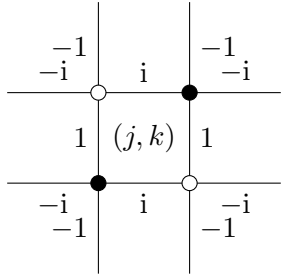

Recall that Corollary~\ref{cor:p_E_and_f_E} expresses the perfect t-embedding in terms of edge probabilities. So from the integral formula of the inverse Kasteleyn matrix, we obtain an integral formula for~$\mathcal T_n(j,k)$. 
\begin{lemma}\label{lem:finite_t_embedding}
For~$(x,y)\in \mathcal A$, we define the action function~$S(z)=S(z;x,y)$ as
\begin{equation*}
S(z)\coloneqq\frac{1}{2}(1+x+y)\log z-\frac{1}{2}(1-x+y)\log(z+1)-\frac{1}{2}(1+x-y)\log(z-1),
\end{equation*}
and
\begin{equation}\label{eq:g_t}
G_{\mathcal T}(z,w)\coloneqq\frac{(1-\i)(z-\i)(w-\i)}{(z-1)(z+1)w}. 
\end{equation}
Let~$(j,k)\in (A_{n+1}')^*$ with~$j+k+n$ odd, and set~$(x,y)=\frac{1}{n+1}(j,k)$. Then the perfect t-embedding
\begin{equation}\label{eq:t-embedding_odd}
\mathcal T_n(j,k)=\frac{1}{(2\pi \i)^2}\int_{\mathcal E_2}\int_{\mathcal E_1}\e^{(n+1)(S(w)-S(z))}G_{\mathcal T}(z,w)\frac{\,d z \,d w}{z-w}-1,
\end{equation}
where~$\mathcal{E}_1$ and~$\mathcal{E}_2$ are positively oriented circles around~$0$ and~$1$, respectively, with radius strictly less than~$1/2$.
\end{lemma}
\begin{proof}
We get from \eqref{eq:kasteleyn}, \eqref{eq:edge_prob_kasteleyn} and \eqref{eq:inverse_kasteleyn} that
\begin{multline*}
p_E(j,k,n)=(-1)^{k+n+1}f_1((j,k),(j,k))\\
= \frac{-1}{(2 \pi \i)^2} \int_{\mathcal{E}_2} \int_{\mathcal{E}_1} \frac{w^{\frac{1}{2}(j+k+n+1)} }{z^{\frac{1}{2}(j+k+n+1)}}
 \frac{(1+z)^{\frac{1}{2}(k-j+n-1)} (z-1)^{\frac{1}{2}(j-k+n+1)}}{(1+w)^{\frac{1}{2}(k-j+n+1)}(w-1)^{\frac{1}{2}(j-k+n+1)}} \frac{\,d z \,d w}{w-z} \\
=\frac{1}{(2 \pi \i)^2} \int_{\mathcal{E}_2} \int_{\mathcal{E}_1} \e^{(n+1)(S(w)-S(z))}\frac{1}{1+z} \frac{\,d z \,d w}{z-w}.
\end{multline*}
Similarly
\begin{multline*}
p_N(j,k,n)=(-1)^{k+n+1}\i f_1((j-1,k+1),(j,k))\\
= \frac{-1}{(2 \pi \i)^2} \int_{\mathcal{E}_2} \int_{\mathcal{E}_1} \e^{(n+1)(S(w)-S(z))}\frac{1}{z-1} \frac{\,d z \,d w}{z-w},
\end{multline*}

\begin{equation*}
p_W(j,k,n)
= \frac{1}{(2 \pi \i)^2} \int_{\mathcal{E}_2} \int_{\mathcal{E}_1} \e^{(n+1)(S(w)-S(z))}\frac{1}{w(z-1)} \frac{\,d z \,d w}{z-w}+1,
\end{equation*}
and
\begin{equation*}
p_S(j,k,n)
= \frac{1}{(2 \pi \i)^2} \int_{\mathcal{E}_2} \int_{\mathcal{E}_1} \e^{(n+1)(S(w)-S(z))}\frac{1}{w(z+1)} \frac{\,d z \,d w}{z-w}.
\end{equation*}
In the calculations of~$p_W$ we have used that~$f_2((j-1,k+1),(j-1,k-1))=(-1)^{k+n}$, and in the calculations of~$p_S$, that~$f_2((j,k),(j-1,k-1))=0$.

Summing up these integrals according to Corollary~\ref{cor:p_E_and_f_E}, proves the statement.
\end{proof}

As we saw in Section~\ref{sec:assumptions}, we already know the limit of the perfect t-embedding, by Corollary~\ref{cor:p_E_and_f_E} and the results of \cite{cohn-elki-prop-96}. However, to prove Assumption~\ref{assumption:bounded_aztec} for the uniform Aztec diamond, we need finer asymptotics, namely the leading order term of~$d \T_n (e_{\iota,\iota'}(j,k))$. For this we cannot rely on the results of~\cite{cohn-elki-prop-96}. Instead we use the formula in previous lemma and perform a steepest descent analysis. Before doing so, we extend the formula to faces where~$j+k+n$ is even. 

Let~$j+k+n$ and~$\iota+\iota'$ be odd, so that~$(j+\iota)+(k+\iota')+n$ is even. Then, by Proposition~\ref{prop:t_rec_aztec},
\begin{equation*}
\mathcal T_n(j+\iota,k+\iota')=\mathcal T_{n+1}(j+\iota,k+\iota').
\end{equation*}
We apply Lemma~\ref{lem:finite_t_embedding} to the right hand side and obtain
\begin{multline}\label{eq:t-embedding_even}
\mathcal T_n(j+\iota,k+\iota') \\
=\frac{1}{(2\pi \i)^2}\int_{\mathcal E_2}\int_{\mathcal E_1}\e^{(n+1)(S(w;x,y)-S(z;x,y))}\e^{S(w;\iota,\iota')-S(z;\iota,\iota')}G_{\mathcal T}(z,w)\frac{\,d z \,d w}{z-w}-1,
\end{multline}
where~$(x,y)=\frac{1}{n+1}(j,k)$. We have used that
\begin{equation*}
(n+2)S\left(w;\frac{j+\iota}{n+2},\frac{k+\iota'}{n+2}\right)=(n+1)S\left(w;\frac{j}{n+1},\frac{k}{n+1}\right)+S(w;\iota,\iota').
\end{equation*} 

The following lemma now follows from \eqref{eq:t-embedding_edge}, \eqref{eq:t-embedding_odd} and \eqref{eq:t-embedding_even}.
\begin{lemma}
Consider the same setting as in Lemma~\ref{lem:finite_t_embedding}. In addition, let~$|\iota|+|\iota'|=1$, and~$e_{\iota,\iota'}(j,k)\in (A_{n+1}')^*$. Then 
\begin{equation*}
d \T_n(e_{\iota,\iota'}(j,k))=\pm\frac{1}{(2\pi \i)^2}\int_{\mathcal E_2}\int_{\mathcal E_1}\e^{(n+1)(S(w)-S(z))}G_{\iota,\iota'}(z,w)G_{\mathcal T}(z,w) \,d z \,d w,
\end{equation*}
where the sign depends on the orientation of the edge, as in \eqref{eq:t-embedding_edge}, and
\begin{multline}\label{eq:g_direction}
G_{1,0}(z,w)=\frac{1}{z(w-1)}, \quad G_{0,1}(z,w)=-\frac{1}{z(w+1)}, \\
G_{-1,0}(z,w)=\frac{1}{w+1}, \quad \text{and} \quad G_{0,-1}(z,w)=\frac{1}{w-1}.
\end{multline}
\end{lemma}

\subsection{Steepest descent analysis}

The action function~$S(z)=S(z;x,y)$ has two critical points, and for~$(x,y)\in \DAz$ the critical points comes as complex conjugate pairs. We denote the critical point in the upper half plane by
\begin{equation*}
\xi=\xi(x, y) = \frac{y-x + \i\sqrt{1-(2x^2+2y^2)}}{1-x-y}.
\end{equation*}
We also define~$\theta_{\xi} \coloneqq \frac{1}{2} \arg(S''(\xi))$.
\begin{lemma}\label{lem:limit_integral}
Let~$\mathcal K \subset\DAz$ be a compact set,~$|\iota|+ |\iota'| = 1$, and define~$G(z, w)\coloneqq G_{\iota,\iota'}(z, w) G_{\mathcal{T}}(z,w)$, where~$G_{\mathcal{T}}$ and~$G_{\iota,\iota'}$ are defined by~\eqref{eq:g_t} and~\eqref{eq:g_direction}. For~$(j,k) \in \mathbb{Z}^2$ such that~$(x,y) = \frac{1}{n+1}(j, k) \in \mathcal K$, 
\begin{multline}\label{eq:double_integral}
\frac{1}{(2\pi\i)^2}\int_{\mathcal{E}_2}\int_{\mathcal{E}_1}\e^{(n+1)\big(S(w)-S(z)\big)}G(z,w)\,d z\,d w \\
=
-\frac{1}{n+1} \frac{1}{2\pi \i |S''(\xi)| } \bigg(  \frac{1}{\e^{2 \i \theta_{\xi}} } G( \xi, \xi)  + \e^{-2 \i (n+1)\im(S(\xi) )} G( \xi, \overline{\xi})  \\
-\frac{1}{\e^{-2 \i \theta_{\xi}} } G( \overline{\xi}, \overline{\xi})     -    \e^{2 \i (n+1)\im(S(\xi) )} G( \overline{\xi}, \xi) 
 \bigg)(1+o(1))
\end{multline} 
as~$n\to \infty$, where the error term is uniform for~$\frac{1}{n+1}(j, k)  \in \mathcal K$.
\end{lemma}

\begin{figure}
 \begin{center}
 \begin{tikzpicture}[scale=1]
    \draw (-2.5,0) node {\includegraphics[trim={0cm 0cm 1cm 3cm},clip, width=0.22\textwidth]{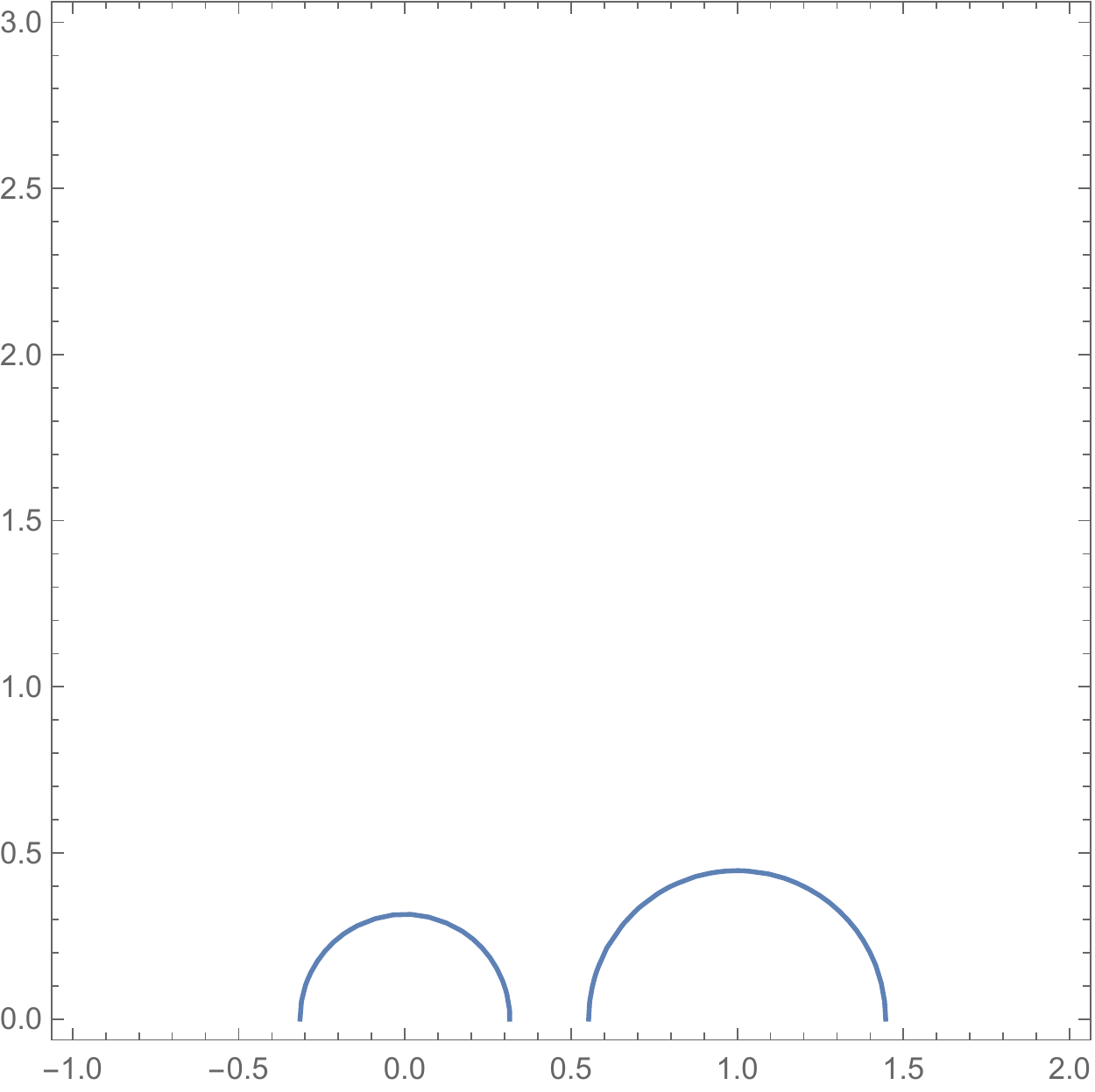}};
    \draw (2.5,0) node {\includegraphics[trim={0cm 0cm .2cm 3cm},clip, width=0.22\textwidth]{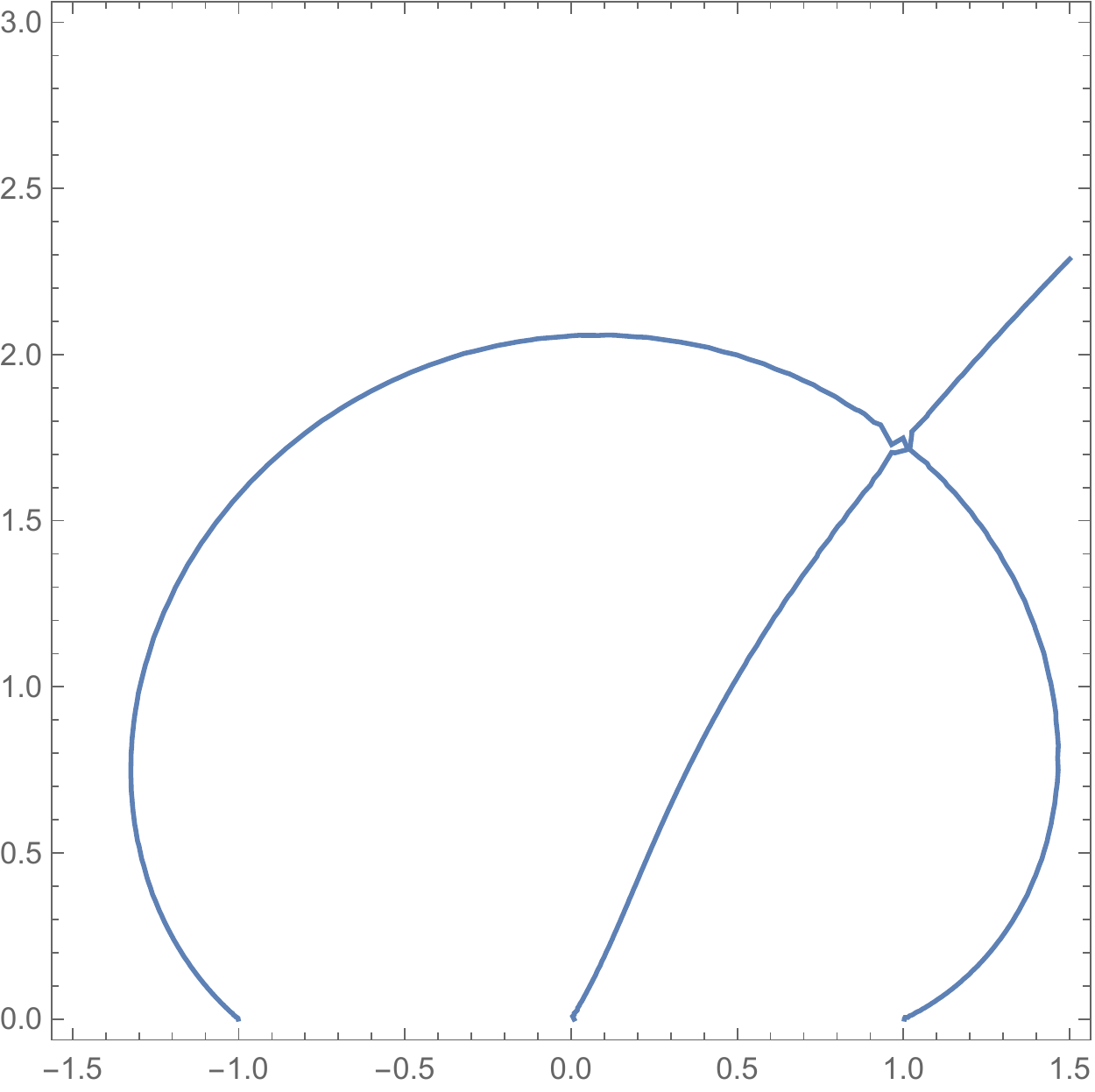}};
	\draw (-2.8,-.6) node{$\mathcal E_1$};
	\draw (-1.7,-.4) node{$\mathcal E_2$};
	\draw (2.1,.7) node{$\gamma_z$};
	\draw (2.7,0) node{$\gamma_w$};
  \end{tikzpicture}
 \end{center}
  \caption{Part of the curves~$\mathcal E_1$,~$\mathcal E_2$,~$\gamma_z$ and~$\gamma_w$. The curves are symmetric with respect to the real line. In the example given here~$x=.1$ and~$y=.5$.\label{fig:steepest_descent_curves}}
\end{figure}

\begin{proof}

We perform the proof in steps, and as the argument is standard we do not give all details.

\begin{enumerate}[1.]
\item  First we deform the contours, without crossing any residues, to the steepest descent and ascent contours, of~$\left|\e^{(n+1)S}\right|=\e^{(n+1)\re S}$, passing through~$\xi$ and~$\overline{\xi}$, given by the level lines~$\im (S(z)) = \im(S(\xi))$ and~$\im(S(z)) = \im(S(\overline{\xi}))$. We denote the~$z$ contour by~$\gamma_z$ and the~$w$ contour by~$\gamma_w$. The contours~$\gamma_z$ and~$\gamma_w$ are as follows: each consists of a piece in the upper half plane through~$\xi$, and the reflection of this piece in the lower half plane. The part of~$\gamma_w$ in the upper half plane goes from~$0$ to~$\infty$, and the part of~$\gamma_z$ in the upper half plane is an arc from~$1$ to~$-1$, see Figure~\ref{fig:steepest_descent_curves}. Since
\begin{equation*}
\e^{(n+1)(S(w)-S(z))}=\frac{w^{\frac{n+1}{2}(1+x+y)}(z+1)^{\frac{n+1}{2}(1-x+y)}(z-1)^{\frac{n+1}{2}(1+x-y)}}{(w+1)^{\frac{n+1}{2}(1-x+y)}(w-1)^{\frac{n+1}{2}(1+x-y)}z^{\frac{n+1}{2}(1+x+y)}},
\end{equation*}
we see that the poles of the integrand all occur at~$z = 0, w = \pm 1$, and the residue of the integrand at~$w = \infty$ is~$0$. Thus, as argued in \cite[Section 5]{CKBAztec}, we can indeed deform the contours~$\mathcal E_1$ and~$\mathcal E_2$ to~$\gamma_z$ and~$\gamma_w$, respectively, and in our case we do not pick up any residue because we have no pole at~$z=w$ in the integrand.

\item The contribution to the remaining integral comes from the union of~$4$ regions where~$(z, w)$ are both in a small neighborhood of the points~$\{\xi, \overline{\xi}\}$. Indeed, for small enough~$\eps>0$ we can truncate the integral to the parts of the contours in balls of radius~$(n+1)^{-\frac{1}{2} + \eps}$ around the critical points, at the cost of an error of order~$\e^{-c (n+1)^{\eps}}$ for some~$c >0$. We see this as follows: if we choose a pair of points~$\Xi_z, \Xi_w \in \{\xi, \overline{\xi}\} $, and write 
$$\tilde{z}= \sqrt{(n+1)}(z - \Xi_z), \quad \text{and} \quad \tilde{w} = \sqrt{(n+1)}(w - \Xi_w), $$
then we see that when the contours for~$\tilde{w}, \tilde{z}$ cross the boundary of this neighborhood, the norm of the integrand is of order 
$\exp(- c (n+1)^{2 \eps})$
for some small~$c > 0$. Since the contours are in the direction of steepest ascent and descent, the entire integral along the parts of the contour outside of these neighborhoods of critical points decays as~$\exp(-c (n+1)^{2 \eps})$.

\item Now we compute the leading order contribution coming from pairs of neighborhoods around points~$\{\xi, \overline{\xi}\}$. For each choice of pair~$\Xi_z, \Xi_w \in \{\xi, \overline{\xi}\}$, we define~$t, s$ by 
\begin{equation*}
\tilde{w} = \frac{t}{\i |S''(\xi)|^{1/2} \e^{\i \theta_{ \Xi_w} }}, \quad \text{and} \quad \tilde{z} = \pm  \frac{s}{|S''(\xi)|^{1/2} \e^{\i \theta_{ \Xi_z}  }},
\end{equation*}
where we use~$\theta_{ \Xi_w} = \frac{1}{2} \arg(S''(\Xi_w))$, and~$\theta_{\Xi_z} = \frac{1}{2} \arg(S''(\Xi_z)) $. We take the~$+$ when~$\Xi_z = \xi$ and~$-$ if~$\Xi_z = \overline{\xi}$. 

We Taylor expand~$S$ and~$G$ around~$(\Xi_z,\Xi_w)$, make the substitution given above, and then extend the curves of integration to~$\pm \infty$. We obtain that the left hand side of \eqref{eq:double_integral} is equal to 
\begin{multline*}
\frac{1}{(n+1) \i |S''(\xi)| } \sum_{\Xi_z, \Xi_w \in \{\xi, \overline{\xi}\}} \frac{\pm 1}{\e^{\i \theta_{\Xi_w}} \e^{\i \theta_{ \Xi_z}}} G(\Xi_z, \Xi_w) \e^{(n+1)\big(S(\Xi_w) - S(\Xi_z)\big)} \\
 \times \frac{1}{(2\pi \i)^2} \int_{-\infty}^\infty \int_{-\infty}^\infty \e^{-\frac{t^2}{2} - \frac{s^2}{2}}\,d s\,d t 
 \left(1  + \Ordo\left(n^{3\eps-1/2}\right) \right),
\end{multline*}
as~$n\to \infty$. The signs are taken according to the signs in the definition of~$\tilde z$. Hence, the leading order term is equal to
 \begin{multline*}
-\frac{1}{n+1} \frac{1}{2\pi \i |S''(\xi)| } \bigg(  \frac{1}{\e^{2 \i \theta_{\xi}} } G( \xi, \xi)  + \frac{1}{\e^{ \i \theta_{\xi}} \e^{ \i \theta_{\overline{\xi}}} }  \e^{-(n+1)\big(S(\xi) - S(\bar \xi)\big)} G( \xi, \overline{\xi})  \\
-\frac{1}{\e^{2 \i \theta_{\overline{\xi}}} } G( \overline{\xi}, \overline{\xi})     -   \frac{1}{ \e^{ \i \theta_{\overline{\xi}}} \e^{ \i \theta_{\xi}} }   \e^{(n+1)\big(S(\xi) - S(\bar\xi)\big)} G( \overline{\xi}, \xi) 
 \bigg) . 
 \end{multline*}
We observe that by real analyticity of~$S(z)$, this is exactly the expression in the theorem. 

It remains to observe that all implicit constants in the error terms can be taken to be continuous in~$(x, y)$, and thus can be taken to be uniform for~$\frac{1}{n+1} (j, k) \in \mathcal K$.
\end{enumerate}
\end{proof}

\subsection{Edges and angles}
For the bound of the size of the edges and angles in Assumption~\ref{assumption:bounded_aztec}, we use the leading term in Lemma~\ref{lem:limit_integral}. Recall that~$G=G_{\iota,\iota'}G_{\mathcal T}$, where~$G_{\mathcal{T}}$ and~$G_{\iota,\iota'}$ are defined by~\eqref{eq:g_t} and~\eqref{eq:g_direction},  and~$|\iota|+|\iota'|=1$. 

We can simplify the leading term in Lemma~\ref{lem:limit_integral} slightly. To that end we define~$h_{\iota,\iota'}$ and~$g_{\iota,\iota'}$ from \eqref{eq:g_direction} so that~$G_{\iota,\iota'}(z,w)=h_{\iota,\iota'}(z)g_{\iota,\iota'}(w)$. To specify the functions completely, we require~$h_{\iota,\iota'}(1)=1$. We denote 
\begin{equation}\label{eqn:defs} 
\alpha \coloneqq \frac{1}{\e^{2 \i \theta_{\xi}}}, \;\; \text{and} \;\; \beta \coloneqq \e^{-2 (n+1) \i \mathrm{Im}(S(\xi))}. 
\end{equation}

The explicit expression of~$G_\T$, \eqref{eq:g_t}, yield
\begin{equation*}
G(\xi,\xi)\alpha+G(\xi,\bar \xi)\beta=\frac{(1-\i)h_{\iota,\iota'}(\xi)(\xi-\i)}{(\xi-1)(\xi+1)}\left(\frac{g_{\iota,\iota'}(\xi)(\xi-\i)}{\xi}\alpha+\frac{g_{\iota,\iota'}(\bar \xi)(\bar \xi-\i)}{\bar \xi}\beta\right),
\end{equation*}
and
\begin{equation*}
G(\bar \xi,\bar \xi)\bar \alpha+G(\bar\xi,\xi)\bar \beta=\frac{(1-\i)h_{\iota,\iota'}(\bar \xi)(\bar \xi-\i)}{(\bar \xi-1)(\bar \xi+1)}\left(\frac{g_{\iota,\iota'}(\bar \xi)(\bar \xi-\i)}{\bar \xi}\bar \alpha+\frac{g_{\iota,\iota'}(\xi)(\xi-\i)}{\xi}\bar \beta\right).
\end{equation*}
Since~$|\alpha|=|\beta|=1$, we get that
\begin{equation*}
(G(\xi,\xi)\alpha+G(\xi,\bar \xi)\beta)=\alpha \beta\frac{(\bar \xi-1)(\bar \xi+1) h_{\iota,\iota'}(\xi)}{(\xi-1)(\xi+1) h_{\iota,\iota'}(\bar \xi)}\frac{\xi-\i}{\bar \xi-\i}(G(\bar \xi,\bar \xi)\bar \alpha+G(\bar \xi, \xi)\bar \beta).
\end{equation*}
In particular up to a sign, the numerator of the coefficient of~$\frac{1}{n+1}$ in Lemma~\ref{lem:limit_integral} is given by
\begin{multline}\label{eq:length_edge_factor}
G(\xi,\xi)\alpha+G(\xi,\bar \xi)\beta-G(\bar \xi,\bar \xi)\bar \alpha-G(\bar \xi, \xi)\bar \beta \\
=\left(\alpha \beta\frac{(\bar \xi-1)(\bar \xi+1) h_{\iota,\iota'}(\xi)}{(\xi-1)(\xi+1) h_{\iota,\iota'}(\bar \xi)}\frac{\xi-\i}{\bar \xi-\i}-1\right)(G(\bar \xi,\bar \xi)\bar \alpha+G(\bar \xi, \xi)\bar \beta)\\
=\frac{\bar \alpha(1-\i)h_{\iota,\iota'}(\bar \xi)g_{\iota,\iota'}(\bar \xi)(\bar \xi-\i)^2}{(\bar \xi-1)(\bar \xi+1)\bar \xi}
\left(\alpha \beta\frac{(\bar \xi-1)(\bar \xi+1) h_{\iota,\iota'}(\xi)}{(\xi-1)(\xi+1) h_{\iota,\iota'}(\bar \xi)}\frac{\xi-\i}{\bar \xi-\i}-1\right)
\left(1+\alpha\bar \beta\frac{g_{\iota,\iota'}(\xi)\bar \xi}{g_{\iota,\iota'}(\bar \xi)\xi}\frac{\xi-\i}{\bar \xi-\i}\right). 
\end{multline}

The following lemma proves the boundedness of the edges in Assumption~\ref{assumption:bounded_aztec}, with~$\mu_n=\frac{1}{n+1}$. 
\begin{lemma}\label{lem:bound_edges}
Let~$\mathcal K\subset\DAz$ be a compact set. There exist positive constants~$N_\mathcal{K}$,~$C_\mathcal{K}>1$ such that
\begin{equation*}
\frac{1}{n+1}\frac{1}{C_\mathcal{K}}  \leq |d\T_n(e_{\iota,\iota'}(j,k))| \leq \frac{1}{n+1}C_\mathcal{K},
\end{equation*}
for all edges~$e_{\iota,\iota'}(j,k)\subset \mathcal K$ and all~$n>N_\mathcal{K}$. 
\end{lemma}

\begin{proof}
Recall that~$|\iota|+|\iota'|=1$ for the edge~$e_{\iota,\iota'}(j,k)$, and we assume without loss of generality that~$j+k+n$ is odd. We set~$(x,y)=\frac{1}{n+1}(j,k)$. Note that for~$(x,y)\in \mathcal K$,~$\xi = \xi(x,y)$ is bounded away from the real line, and~$|S''(\xi)|$ is bounded away from zero and infinity. Since the limit in Lemma~\ref{lem:limit_integral} is uniform on compact subsets, it is sufficient to prove that
\begin{equation*}
G(\xi,\xi)\alpha+G(\xi,\bar \xi)\beta-G(\bar \xi,\bar \xi)\bar \alpha-G(\bar \xi, \xi)\bar \beta,
\end{equation*}
is bounded away from zero and infinity on~$\mathcal K$.

The statement now follows from \eqref{eq:length_edge_factor}. Since~$\xi(x,y)$ is in the upper half plane and bounded away from~$\mathbb{R}$, the first factor in \eqref{eq:length_edge_factor} is uniformly bounded away from zero and infinity. The same is true for the second and third factor, since
\begin{equation*}
\left|\alpha \beta\frac{(\bar \xi-1)(\bar \xi+1)}{(\xi-1)(\xi+1)}\frac{h_{\iota,\iota'}(\xi)}{h_{\iota,\iota'}(\bar \xi)}\right|=1
=\left|\alpha\bar \beta\frac{\bar\xi}{\xi}\frac{g_{\iota,\iota'}(\xi)}{g_{\iota,\iota'}(\bar\xi)}\right|,
\end{equation*}
where we used that~$h_{\iota,\iota'}$ and~$g_{\iota,\iota'}$ are real, while
\begin{equation*}
\left|\frac{\xi-\i}{\bar \xi-\i}\right|<c<1,
\end{equation*}
for some constant~$c$ which only depends on~$\mathcal K$. Hence, there exists a constant~$\widetilde C>1$ so that for~$(x,y)\in \mathcal K$,
\begin{equation*}
\frac{1}{\widetilde C}\leq |G(\xi,\xi)\alpha+G(\xi,\bar \xi)\beta-G(\bar \xi,\bar \xi)\bar \alpha-G(\bar \xi, \xi)\bar \beta|\leq \widetilde C,
\end{equation*}
which proves the statement.
\end{proof}

We end the section by proving the second part of Assumption~\ref{assumption:bounded_aztec}, that the angles are bounded away from~$0$ and~$\pi$. 
\begin{lemma}\label{lem:bound_angles}
Let~$\mathcal K\subset \DAz$ be a compact set. Then there exist constants~$N_{\mathcal K}>0$ and~${\varepsilon_\mathcal{K}\in (0,\pi)}$, such that if~$n >N_\mathcal{K}$, and~$\frac{1}{n+1}e_{\iota_1,\iota_1'}(j,k),\frac{1}{n+1}e_{\iota_2,\iota_2'}(j,k)$ are contained in~$\mathcal K$, the angle between two edges~$d\T_n(e_{\iota_1,\iota_1'}(j,k))$ and~$d\T_n(e_{\iota_2,\iota_2'}(j,k))$, which are adjacent to a common face, is contained in~$(\varepsilon_\mathcal{K}, \pi-\varepsilon_\mathcal{K})$. 
\end{lemma}

\begin{proof}
It is enough to prove the statement for~$j+k+n$ odd. Indeed, a face of~$\T_n((A_{n+1}')^*\cap \mathcal{K})$ is a convex quadrilateral with edges bounded from above and below, by Lemma~\ref{lem:bound_edges}. If we know the statement is true for~$j+k+n$ odd, then it means that two opposite angles in the quadrilateral are bounded from~$\pi$ and~$0$. It follows that the other two angles in the quadrilateral are bounded away from~$0$. Recall the angle condition in the definition of perfect t-embeddings, Definition~\ref{def:temp}, which states that the sum of all angles around a vertex corresponding to the corner of a black (white) face is equal to~$\pi$. Since we consider the square lattice, and since all angles are bounded away from~$0$, we conclude from the angle condition that all angles are also bounded away from~$\pi$.


By the above discussion we assume that $j+k+n$ is odd. We also assume that~$(\iota_1,\iota_1')$ follows by~$(\iota_2,\iota_2')$ when viewed as point on the unit circle oriented in positive direction. 
 
In addition to the notation \eqref{eqn:defs} we set
\begin{multline*}
A_i\coloneqq G_{\iota_i,\iota_i'}(\xi,\xi)G_\mathcal T(\xi,\xi), \quad B_i\coloneqq G_{\iota_i,\iota_i'}(\xi,\bar \xi)G_\mathcal T(\xi,\bar \xi), \\
 \widetilde A_i\coloneqq G_{\iota_i,\iota_i'}(\bar \xi,\bar \xi)G_\mathcal T(\bar \xi,\bar \xi), \quad \text{and} \quad
\widetilde B_i\coloneqq G_{\iota_i,\iota_i'}(\bar \xi,\xi)G_\mathcal T(\bar \xi,\xi).
\end{multline*}

Let us consider 
\begin{equation}\label{eq:scalar_product}
\frac{\overline{\left(A_1\alpha+B_1\beta-\widetilde A_1\bar \alpha-\widetilde B_1\bar \beta\right)}\left(A_2\alpha+B_2\beta-\widetilde A_2\bar \alpha-\widetilde B_2\bar \beta\right)}{\left|A_1\alpha+B_1\beta-\widetilde A_1\bar \alpha-\widetilde B_1\bar \beta\right|\left|A_2\alpha+B_2\beta-\widetilde A_2\bar \alpha-\widetilde B_2\bar \beta\right|},
\end{equation}
as by Lemma~\ref{lem:limit_integral} this is the leading order term of~$\e^{\i \theta}$, as~$n\to \infty$, where~$\theta$ is the angle between~$d\T_n(e_{\iota_1,\iota_1'}(j,k))$ and~$d\T_n(e_{\iota_2,\iota_2'}(j,k))$. To prove the statement we show that the imaginary part of \eqref{eq:scalar_product} is positive and bounded away from zero for~$(x,y)\in \mathcal K$. 

Since we consider consecutive points, we see from the explicit formulas \eqref{eq:g_direction} that either
\begin{equation*}
\frac{h_{\iota_1,\iota_1'}(\xi)}{h_{\iota_1,\iota_1'}(\bar\xi)}=\frac{h_{\iota_2,\iota_2'}(\xi)}{h_{\iota_2,\iota_2'}(\bar\xi)} \quad \text{or} \quad \frac{g_{\iota_1,\iota_1'}(\xi)}{g_{\iota_1,\iota_1'}(\bar\xi)}=\frac{g_{\iota_2,\iota_2'}(\xi)}{g_{\iota_2,\iota_2'}(\bar\xi)}.
\end{equation*}
Recall that~$h_{\iota_i,\iota_i'}$ is the factor of~$G_{\iota,\iota'}$ depending on~$z$ and~$g_{\iota_i,\iota_i'}$ the factor depending on~$w$. 

We assume first that~$h_{\iota_1,\iota_1'}(\xi)/h_{\iota_1,\iota_1'}(\bar\xi)=h_{\iota_2,\iota_2'}(\xi)/h_{\iota_1,\iota_1'}(\bar\xi)$, that is,
\begin{equation*}
h_{\iota_1,\iota_1'}(z)=h_{\iota_2,\iota_2'}(z)=\frac{1}{z}, \quad g_{\iota_1,\iota_1'}(w)=\frac{1}{w-1}\quad \text{and} \quad g_{\iota_2,\iota_2'}(w)=-\frac{1}{w+1},
\end{equation*}
or
\begin{equation*}
h_{\iota_1,\iota_1'}(z)=h_{\iota_2,\iota_2'}(z)=1, \quad g_{\iota_1,\iota_1'}(w)=\frac{1}{w+1}\quad \text{and} \quad g_{\iota_2,\iota_2'}(w)=\frac{1}{w-1}.
\end{equation*}
From \eqref{eq:length_edge_factor} we get that \eqref{eq:scalar_product} is equal to
\begin{equation}\label{eq:scalar_product_simplified}
\frac{\left(g_{\iota_1,\iota_1'}(\xi)+\overline{\alpha\bar\beta\frac{\bar \xi}{\xi}\frac{\xi-\i}{\bar \xi-\i}}g_{\iota_1,\iota_1'}(\bar \xi)\right)
\left(g_{\iota_2,\iota_2'}(\bar \xi)+\alpha\bar\beta\frac{\bar\xi}{\xi}\frac{\xi-\i}{\bar \xi-\i}g_{\iota_2,\iota_2'}(\xi)\right)}
{\left|g_{\iota_1,\iota_1'}(\xi)+\overline{\alpha\bar\beta\frac{\bar \xi}{\xi}\frac{\xi-\i}{\bar \xi-\i}}g_{\iota_1,\iota_1'}(\bar \xi)\right|
\left|g_{\iota_2,\iota_2'}(\bar \xi)+\alpha\bar\beta\frac{\bar\xi}{\xi}\frac{\xi-\i}{\bar \xi-\i}g_{\iota_2,\iota_2'}(\xi)\right|}.
\end{equation}  
We expand the product in the numerator. The sum of the crossing terms is real, so the imaginary part of \eqref{eq:scalar_product_simplified} is equal to
\begin{multline}\label{eq:scalar_product_imaginary}
\frac{\im\left(g_{\iota_1,\iota_1'}(\xi)g_{\iota_2,\iota_2'}(\bar\xi)+g_{\iota_1,\iota_1'}(\bar\xi)g_{\iota_2,\iota_2'}(\xi)\left|\frac{\xi-\i}{\bar\xi-\i}\right|^2\right)}
{\left|g_{\iota_1,\iota_1'}(\xi)+\overline{\alpha\bar\beta\frac{\bar \xi}{\xi}\frac{\xi-\i}{\bar \xi-\i}}g_{\iota_1,\iota_1'}(\bar \xi)\right|
\left|g_{\iota_2,\iota_2'}(\bar \xi)+\alpha\bar\beta\frac{\bar\xi}{\xi}\frac{\xi-\i}{\bar \xi-\i}g_{\iota_2,\iota_2'}(\xi)\right|} \\
=\frac{\left(1-\left|\frac{\xi-\i}{\bar\xi-\i}\right|^2\right)}
{\left|1+\overline{\alpha\bar\beta\frac{\bar \xi}{\xi}\frac{\xi-\i}{\bar \xi-\i}}\frac{g_{\iota_1,\iota_1'}(\bar \xi)}{g_{\iota_1,\iota_1'}(\xi)}\right|
\left|1+\alpha\bar\beta\frac{\bar\xi}{\xi}\frac{\xi-\i}{\bar \xi-\i}\frac{g_{\iota_2,\iota_2'}(\xi)}{g_{\iota_2,\iota_2'}(\bar \xi)}\right|}\frac{\im\left(g_{\iota_1,\iota_1'}(\xi)g_{\iota_2,\iota_2'}(\bar\xi)\right)}{|g_{\iota_1,\iota_1'}(\xi)||g_{\iota_2,\iota_2'}(\bar\xi)|}.
\end{multline}  
On compact subsets of the upper half plane, it follows, as in the proof of Lemma~\ref{lem:bound_edges}, by the inequality 
\begin{equation*}
\left|\frac{\xi-\i}{\bar \xi-\i}\right|<c<1,
\end{equation*}
that the numerator and denominator in the first factor of \eqref{eq:scalar_product_imaginary} are bounded away from zero and infinity.  Moreover, by the explicit expression of~$g_{\iota_i,\iota_i'}$, 
\begin{equation*}
\frac{\im\left(g_{\iota_1,\iota_1'}(\xi)g_{\iota_2,\iota_2'}(\bar\xi)\right)}{|g_{\iota_1,\iota_1'}(\xi)||g_{\iota_2,\iota_2'}(\bar\xi)|}=\frac{2\im \xi}{|\xi-1||\xi+1|}>C>0,
\end{equation*}
for some constant~$C$. This proves the statement if~$h_{\iota_1,\iota_1'}(\xi)/h_{\iota_1,\iota_1'}(\bar\xi)=h_{\iota_2,\iota_2'}(\xi)/h_{\iota_1,\iota_1'}(\bar\xi)$. 

If instead~$g_{\iota_1,\iota_1'}(\xi)/g_{\iota_1,\iota_1'}(\bar\xi)=g_{\iota_2,\iota_2'}(\xi)/g_{\iota_1,\iota_1'}(\bar\xi)$, then a similar argument still goes through.
\end{proof}

\bibliographystyle{alpha}
\bibliography{bib}

\end{document}